\documentclass[a4paper,UKenglish, autoref, thm-restate]{lipics-v2021}

\newif\iffinal
\finaltrue

\makeatletter
\AtBeginDocument{%
  \@ifpackageloaded{hyperref}
  {\def\@doi#1{\href{https://doi.org/#1}
      {\ttfamily https://doi.org/#1}\egroup}}
  {\def\@doi#1{\ttfamily https://doi.org/#1\egroup}}
  \def\doi{\bgroup\catcode`\_=12\relax\@doi}}
\makeatother

\newcommand{\LongVersion}[1]{}

\usepackage[svgnames,table]{xcolor}
\definecolor{darkblue}{rgb}{0, 0, 0.7}

\usepackage{color,soul}
\definecolor{lightblue}{rgb}{.8,.95,1}

\usepackage{pifont}
\usepackage{makecell} 
\usepackage{caption}
\usepackage{subcaption}
\usepackage{graphicx}
\usepackage{cancel}
\usepackage{mathtools}
\usepackage{tikz}
\usetikzlibrary{positioning,arrows.meta,decorations.pathmorphing,calc,shapes.multipart,automata,matrix,calc}
\usepackage{upgreek}
\usepackage{xspace}
\usepackage{multirow}
\usepackage{booktabs}
\usepackage[normalem]{ulem}
\usepackage{paralist}
\usepackage{enumitem}
\usepackage{wrapfig}
\usepackage{multidef}
\usepackage{amsmath,amssymb,amsfonts}
\usepackage{xcolor}
\usepackage{soul}
\usepackage{graphicx}
\usepackage{wrapfig}
\usepackage{booktabs}

\usepackage{hyperref}
\hypersetup{
    colorlinks=true,
    linkcolor=blue!60!black,
    citecolor=blue!60!black,
    urlcolor=blue!80!black
}

\usepackage[capitalise,nameinlink]{cleveref}
\multidef{\ensuremath{\mathsf{#1}}\xspace}{ack,rel}

\usepackage[ruled,vlined,linesnumbered]{algorithm2e}
	\SetKwInOut{Input}{input}
	\SetKwInOut{Output}{output}
\newcommand{\assign}{\leftarrow}
\newtheorem{assumption}{Assumption}

\crefname{algocfline}{\text{line}}{\text{lines}} %
\crefname{algocf}{Algorithm}{Algorithms}
\Crefname{algocf}{Algorithm}{Algorithms}

\newcommand{\eg}{e.g.,\xspace}

\newcommand{\ie}{i.e.,\xspace}

\newcommand{\wrt}{w.r.t.\xspace}

\usepackage{paralist} %

\newenvironment{ienumerate}
	{\ifdefined\VersionLong\begin{enumerate}\else\begin{inparaenum}[\itshape i\upshape)]\fi}
	{\ifdefined\VersionLong\end{enumerate}\else\end{inparaenum}\fi}

\newenvironment{oneenumerate}
	{\ifdefined\VersionLong\begin{enumerate}\else\begin{inparaenum}[1)]\fi}
	{\ifdefined\VersionLong\end{enumerate}\else\end{inparaenum}\fi}

\iffinal
  \newcommand{\checkmacros}[1]{#1}
\else
  \newcommand{\checkmacros}[1]{{\color{orange}#1}}
\fi
\newcommand{\card}[1]{\ensuremath{\left|#1\right|}}

\usepackage{blindtext}

\makeatletter
\def\BState{\State\hskip-\ALG@thistlm}
\makeatother
\tikzset{
     redblock/.style={rectangle, fill=red!40, text width=2em,
                   text centered, rounded corners, minimum height=1em},
     state/.style={circle, text width=2em,
                   text centered, minimum height=2em},
     basic/.style  = {draw, text width=2cm,  font=\sffamily, rectangle},
     arrow/.style={-{Stealth[]}}
     }

\iffinal
\newcommand{\add}[1]{#1}
\newcommand{\textreplace}[2]{#2}
\newcommand{\remove}[1]{}

\else
\usepackage{soul}
\sethlcolor{green}
\newcommand{\textreplace}[2]{{\color{blue}\st{\mbox{#1}} #2}}
\newcommand{\remove}[1]{{\color{red}\st{\mbox{#1}}}}
\newcommand{\add}[1]{{\color{blue}#1}}

\fi

\newcommand{\STA}{\checkmacros{\ensuremath{\mathcal{S}}}}
\newcommand{\STArule}{\checkmacros{\ensuremath{r}}}
\newcommand{\STArulelength}{\checkmacros{\ensuremath{m}}}
\newcommand{\STAConfig}{\checkmacros{\ensuremath{\nConfig_{\STA}}}}

\newcommand{\STAConfigdefnpre}{\ensuremath{\big((\loc_{1},\clockval_{1}), \ldots, (\loc_{\networksize},\clockval_{\networksize})\big)}}
\newcommand{\STAresetclocks}[2]{\checkmacros{\ensuremath{\clocks_{#1,#2}}}}
\newcommand{\STAruledefn}{\checkmacros{\ensuremath{ \big \langle \loc_{\STArule,1}  \xrightarrow[]{\guard_{\STArule,1},\STAresetclocks{\STArule}{1},\tatranlabel_{\STArule,1}}  \loc'_{\STArule,1} , \cdots, \loc_{\STArule,\STArulelength}  \xrightarrow[]{\guard_{\STArule,\STArulelength},\STAresetclocks{\STArule}{\STArulelength},\tatranlabel_{\STArule,\STArulelength}}  \loc'_{\STArule,\STArulelength} \big \rangle }}}

\newcommand{\stablecomputationbasecasedefn}{\nConfig_{\corrTA,0}  \rightarrow^* \nConfig_{\corrTA}}
\newcommand{\intercomputationdefn}{\nConfig_{\corrTA,pre} \rightarrow^* \nConfig_{\corrTA,mid} \rightarrow^* \nConfig_{\corrTA}}

\newcommand{\STAcomputation}{\checkmacros{\ensuremath{\pi_{\STA}}}}
\newcommand{\stablecomputation}{\checkmacros{\ensuremath{\pi}}}
\newcommand{\stablecomputationbasecase}{\checkmacros{\ensuremath{\pi}}}
\newcommand{\numOfOcc}[2]{\checkmacros{\ensuremath{\#(#1,#2)}}}

\newcommand{\locmid}{\checkmacros{\ensuremath{p}}}
\newcommand{\Locmid}{\checkmacros{\ensuremath{P}}}
\newcommand{\locsink}{\checkmacros{\ensuremath{\loc_\bot}}}

\newcommand{\cupdot}{\mathbin{\dot{\cup}}}

\newcommand{\tranleftgadget}{\checkmacros{\ensuremath{(\loc_{\STArule,i}, \guard_{\STArule,i}, \STAresetclocks{\STArule}{i} \cup \{\clocksync\}, \dummytranlabel_{r,i}, \locmid_{\STArule, i-1}, \locmid_{\STArule, i})}}}
\newcommand{\tranrightgadget}{\checkmacros{\ensuremath{(\locmid_{\STArule, i}, \top, \emptyset, \tatranlabel_{\STArule,i}, \locmid_{\STArule,\STArulelength}, \loc'_{\STArule,i})}}}
\newcommand{\Tranleftgadget}[1]{\checkmacros{\ensuremath{T_{#1,1}}}}
\newcommand{\Tranrightgadget}[1]{\checkmacros{\ensuremath{T_{#1,2}}}}
\newcommand{\corrTA}{\checkmacros{\ensuremath{\TA_{\STA}}}}
\newcommand{\corrnetwork}[1]{\checkmacros{\ensuremath{\corrTA^{#1}}}}
\newcommand{\corrTAdefn}{\checkmacros{\ensuremath{(\Loc_{\corrTA}, \locinit, \clocks \cup \{\clocksync\}, \Sigma, \Transitions, \invariant_{\corrTA})}}}

\newcommand{\dummytranlabel}{\checkmacros{\ensuremath{\iota}}}

\newcommand{\transequence}{\checkmacros{\ensuremath{TS}}}
\newcommand{\transequenceset}{\checkmacros{\ensuremath{\textit{TS-set}}}}
\newcommand{\Trangadget}[1]{\checkmacros{\ensuremath{TG_{#1}}}}

\newcommand{\STARules}{\checkmacros{\ensuremath{\mathcal{R}}}}
\newcommand{\STAdefn}{\checkmacros{\ensuremath{(\Loc,\locinit,\clocks,\invariant,\STARules)}}}
\newcommand{\nConfigdefn}{\checkmacros{\ensuremath{\big((\loc_1,\clockval_1), \ldots (\loc_m,\clockval_m)\big)}}}

\newcommand{\corrConfigdefnpre}{\checkmacros{\ensuremath{\big((\loc_{1},\altclockval_{1}), \ldots (\loc_{\networksize},\altclockval_{\networksize})\big)}}}
\newcommand{\corrConfig}{\checkmacros{\ensuremath{\mathfrak{c}_{\corrTA}}}}
\newcommand{\sync}{\mathsf{sync}}

\newcommand{\STAnetwork}[1]{\checkmacros{\ensuremath{\mathcal{S}^{#1}}}}

\newcommand{\nConfig}{\checkmacros{\ensuremath{\mathfrak{c}}}} %

\newcommand{\nConfigInit}{\checkmacros{\ensuremath{\hat{\nConfig}}}} %

\newcommand{\initialvaln}{\ensuremath{\mathbf{0}}}
\newcommand{\tatransition} [3]{{#1}\xrightarrow{#2}{#3}}
\newcommand{\resetclocks}{\checkmacros{\ensuremath{\clocks_r}}}

\newcommand{\network}[1]{\ensuremath{\TA^{#1}}}
\newcommand{\networksize}{\checkmacros{\ensuremath{n}}}
 
\newcommand{\dtntransition}[3]{#1 \xRightarrow{#2}#3}
\newcommand{\dtnconfig}{\ensuremath{\big((\loc_1,\clockval_1), \ldots, (\loc_{\networksize},\clockval_{\networksize})\big)}}
\newcommand{\dtnconfigsucc}{\ensuremath{\big((\loc_1',\clockval_1'),\ldots,(\loc_{\networksize}',\clockval_{\networksize}')\big)}}
\newcommand{\dtnconfigdelay}{\ensuremath{\big((\loc_1,\clockval_1+\delta),\ldots,(\loc_n,\clockval_n+\delta)\big)}}
\newcommand{\loc}{\checkmacros{\ensuremath{q}}}
\newcommand{\locinit}{\ensuremath{\hat{\loc}}}
\newcommand{\Loc}{\checkmacros{\ensuremath{Q}}}

\newcommand{\initloc}{\ensuremath{\hat{\loc}}}

\newcommand{\post}[1]{\checkmacros{\mathit{post}}(#1)}

\newcommand{\smartpar}[1]{\smallskip \noindent \textbf{#1.}}
\newcommand{\gclock}{\checkmacros{\ensuremath{t}}\xspace}

\newcommand{\init}{\ensuremath{\mathsf{init}}\xspace}
\newcommand{\loclisten}{\ensuremath{\mathsf{listen}}\xspace}
\newcommand{\locpost}{\ensuremath{\mathsf{post}}\xspace}
\newcommand{\locreading}{\ensuremath{\mathsf{reading}}\xspace}
\newcommand{\locdone}{\ensuremath{\mathsf{done}}\xspace}
\newcommand{\locerror}{\ensuremath{\mathsf{error}}\xspace}

\newcommand{\guard}{\checkmacros{\ensuremath{g}}}

\newcommand{\locguard}{\checkmacros{\ensuremath{\gamma}}}

\newcommand{\UG}[1]{\ensuremath{\checkmacros{\mathsf{UG}}(#1)}} %
\newcommand{\ugconfig}[1]{\ensuremath{\checkmacros{\mathfrak{c}}}}

\newcommand{\TA}{\checkmacros{\ensuremath{\mathit{A}}}}
\newcommand{\TAdefn}{\checkmacros{\ensuremath{\mathit{(\Loc, \hat{\loc}, \clocks,\Sigma, \Transitions, \invariant)}}}}

\usepackage{accents}
\newcommand{\lcomputation}{\checkmacros{\ensuremath{\rho}}}

\newcommand{\gcomputation}{\checkmacros{\ensuremath{\pi}}}

\newcommand{\totaltime}{\checkmacros{\ensuremath{\delta}}}

\newcommand{\Lg}{\checkmacros{\mathcal{L}}}

\newcommand{\mQz}{\checkmacros{\ensuremath{\mathbb{R}_{\geq 0}}}}

\newcommand{\Nats}{\checkmacros{\ensuremath{\mathbb{N}}}}
\newcommand{\Ints}{\checkmacros{\ensuremath{\mathbb{Z}}}}

\newcommand{\Mbound}{\checkmacros{\ensuremath{M}}}
\newcommand{\MFbound}{\checkmacros{\ensuremath{{M^{\text{\scalebox{0.7}{$\nearrow$}}}}}}}
\newcommand{\MFboundp}{\checkmacros{\ensuremath{{{M'}^{\text{\scalebox{0.7}{$\nearrow$}}}}}}}

\newcommand{\invariant}{\checkmacros{\ensuremath{\mathit{Inv}}}}
\newcommand{\intervalPr}{\checkmacros{\ensuremath{\mathcal{I}}}}
\newcommand{\intervalPrdefn}{\checkmacros{\ensuremath{\{i_1, \ldots, i_k\}}}}

\newcommand{\gta}{\checkmacros{GTA}\xspace}
\newcommand{\ngta}{\checkmacros{NGTA}\xspace}

\newcommand{\ngtas}{\checkmacros{NGTAs}\xspace}

\newcommand{\intfloor}[1]{\lfloor#1\rfloor}

\newcommand{\DTN}{DTN\xspace}
\newcommand{\disjunctivetimednetwork}{Disjunctive Timed Network\xspace}

\newcommand{\lcomputationdefn}{(\loc_0,\clockval_0) \xrightarrow{\delta_0,\sigma_0} \ldots \xrightarrow{\delta_{l-1},\sigma_{l-1}} (\loc_l,\clockval_l)}
\newcommand{\gcomputationdefn}{\nConfig_0 \xrightarrow{\delta_0, (i_0, \sigma_0)} \cdots \xrightarrow{\delta_{l-1}, (i_{l-1}, \sigma_{l-1})} \nConfig_l}

\newcommand{\transition} [3]{{#1}\xrightarrow{#2}{#3}}

\def\calL{\mathcal{L}}
\newcommand{\traceof}[1]{\mathsf{trace}(#1)}
\def\trace{tt}
\def\reduce{\mathsf{reduce}}

\usepackage{blindtext}

\newcommand{\clock}{\checkmacros{\ensuremath{c}}}
\newcommand{\clockval}{\checkmacros{\ensuremath{v}}}
\newcommand{\altclockval}{\checkmacros{\ensuremath{u}}}

\newcommand{\slot}{\checkmacros{\ensuremath{s}}}

\newcommand{\fractional}[1]{\ensuremath{\checkmacros{\mathsf{frac}}(#1)}}
\newcommand{\integral}[1]{\checkmacros{\ensuremath{\lfloor #1 \rfloor}}}
\newcommand{\clocks}{\checkmacros{\ensuremath{\mathcal{C}}}}
\newcommand{\clocksync}{\checkmacros{\ensuremath{\clock_\sync}}}

\newcommand{\clockcons}{\ensuremath{\checkmacros{\Psi}(\clocks)}}
\newcommand{\STAclockcons}{\ensuremath{\checkmacros{\Psi}(\clocks \cup \{ \clocksync\})}}
\newcommand{\ccons}{\ensuremath{\checkmacros{\psi}}}

\newcommand{\intconstant}{\checkmacros{\ensuremath{d}}}

\newcommand{\tatranlabel}{\checkmacros{\ensuremath{\sigma}}}
\newcommand{\tatran}{\checkmacros{\ensuremath{\tau}}}
\newcommand{\tatrandefn}{\ensuremath{(\loc, \guard, \resetclocks,\tatranlabel,\loc')}}

\newcommand{\tadelay}{\ensuremath{\delta}}
\newcommand{\tapath}{\rho}
\newcommand{\tapathdefn}{{(\loc_0,\clockval_0)} \rightarrow \ldots  \rightarrow{(\loc_l,\clockval_l)}}

\usepackage[most]{tcolorbox}

\def\sigmaerr{\ensuremath{\sigma_{\textsf{err}}}}
\newcommand{\dtnfix}[1]{\ensuremath{\checkmacros{\mathsf{fix}}(#1)}}

\newcommand{\regioninit}{\checkmacros{\ensuremath{\hat{\region}}}}
\newcommand{\regionstateinit}{\checkmacros{\ensuremath{(\locinit,\regioninit)}}}
\newcommand{\regionstatewithindex}[1]{\checkmacros{\ensuremath{(\loc_{#1},\region_{#1})}}}

\newcommand{\gtatrandefn}{\ensuremath{(\loc,\guard,\resetclocks,\tatranlabel,\locguard,\loc')}}
\newcommand{\gtatran}{\checkmacros{\ensuremath{\tau}}}

\newcommand{\Transitions}{\ensuremath{\checkmacros{\mathcal{T}}}}
\newcommand{\gtasetoftransitions}{\Transitions}

\newcommand{\send}{\mathsf{snd}}
\newcommand{\rcv}{\mathsf{rcv}}
\newcommand{\Synclabels}{\Lambda}
\newcommand{\synclabel}{\lambda}

\newcommand{\regionequalto}[1]{\simeq_{#1}}
\def\nextslot{\checkmacros{\ensuremath{\mathsf{next}}}}
\newcommand{\slotof}{\checkmacros{\ensuremath{\mathsf{slot}}}}
\def\post{\ensuremath{\checkmacros{\mathsf{Post}}}}
\def\posttime{\ensuremath{\post_{\geq 0}}}

\newcommand{\newslot}[2]{#1_{\slotof\mathrel{+}{#2}}}
\newcommand{\regionof}[2]{[#1]_{#2}}

\newcommand{\regionautomaton}{region automaton}
\newcommand{\RegionAutomaton}{Region Automaton}
\newcommand{\radelaylabel}{\checkmacros{\ensuremath{\mathsf{delay}}}}

\newcommand{\ra}[2]{\ensuremath{\checkmacros{\mathcal{R}}_{#1}(#2)}}

\newcommand{\rgtransition}[3]{\ensuremath{{#1}\xrightarrow{#2}{#3}}}
\newcommand{\genrgtransition}[2]{\ensuremath{{#1}\xrightarrow{}{#2}}}

\newcommand{\setofallzones}{\checkmacros{\ensuremath{\mathcal{Z}}}}
\newcommand{\zone}{\checkmacros{\ensuremath{z}}}

\newcommand{\DRA}[1]{\ensuremath{\mathcal{D}(#1)}}
\newcommand{\DRAdefn}{\checkmacros{\ensuremath{(\algtotalnodes,\regionstateinit, \Sigma, \algtotaledges)}}}

\newcommand{\dtnperiod}[1]{\ensuremath{\checkmacros{\mathsf{period}}_{#1}}}
\newcommand{\dtnregionautomaton}{DTN region automaton}
\newcommand{\proj}[2]{#1{\downarrow}_{#2}}
\newcommand{\eliminate}[2]{#1{\downarrow}_{-{#2}}}
\newcommand{\bound}[1]{{\Mbound(#1)}}
\newcommand{\boundmax}{{\Mbound_{\max}}}
\newcommand{\regionpath}{\rho_{\region}}

\newcommand{\regionpathdefn}{(\loc_0,\region_0) \xrightarrow{\sigma_0} \ldots \xrightarrow{\sigma_{l-1}} (\loc_l,\region_l)}

\newcommand{\dtnregionpath}{\rho_D}
\newcommand{\dtnregionpathdefn}{(\loc_0,\region_0) \xRightarrow{\sigma_0} \ldots \xRightarrow{\sigma_{l-1}}(\loc_l,\region_l)}

\newcommand{\lastzone}{z_l}
\newcommand{\lastregion}{\region_l}

\newcommand{\valuationsofdbm}[1]{\checkmacros{\ensuremath{[#1]}}}
\newcommand{\zonedbmdefn}{\checkmacros{\ensuremath{(Z_{\clock \clock'})_{\clock,\clock' \in \clocks_0}}}}
\newcommand{\zonedbmdefnprime}{\checkmacros{\ensuremath{(Z_{\clock \clock'}')_{\clock,\clock' \in \clocks_0}}}}
\newcommand{\zonedbmentry}{\checkmacros{\ensuremath{Z_{\clock \clock'}}}}
\newcommand{\zonedbmentryprime}{\checkmacros{\ensuremath{Z_{\clock \clock'}'}}}

\newcommand{\zonedbmentrydefn}{\checkmacros{\ensuremath{(\lessdot_{\clock \clock'},\intconstant_{\clock \clock'})}}}

\newcommand{\regiondbmforproof}{\checkmacros{\ensuremath{R}}}
\newcommand{\zonedbmforproof}{\checkmacros{\ensuremath{Z}}}
\newcommand{\zonedbmentryforproof}[2]{\checkmacros{\ensuremath{Z_{#1,#2}}}}
\newcommand{\regiondbmentryforproof}[2]{\checkmacros{\ensuremath{R_{#1,#2}}}}

\newcommand{\setRS}{\checkmacros{\ensuremath{W}}}
\newcommand{\setEdges}{\checkmacros{\ensuremath{E}}}

\newcommand{\algtotalnodes}{\checkmacros{\ensuremath{\setRS}}}

\newcommand{\algtotaledges}{\checkmacros{\ensuremath{\setEdges}}}

\newcommand{\regionstates}{region states}
\newcommand{\summaryautomaton}[1]{summary timed automaton}
\newcommand{\alglayernodes}[1]{\setRS_{#1}}

\newcommand{\alglayeredges}[1]{\setEdges_{#1}}

\newcommand{\indexofsecondrepeatinglayer}{\checkmacros{\ensuremath{l_0}}}
\newcommand{\indexoffirstrepeatinglayer}{\checkmacros{\ensuremath{i_0}}}

\def\indexl{l}

\newcommand{\algpathindex}[1]{\ensuremath{\checkmacros{\textsf{ind}}(#1)}}
\newcommand{\algtranindex}[1]{\ensuremath{\checkmacros{\textsf{ind}}(#1)}}

\newcommand{\dtndelay}{\checkmacros{\ensuremath{\checkmacros{\epsilon}}}}
\newcommand{\da}{\checkmacros{\ensuremath{\mathcal{S}(\TA)}}}

\def\raUGA{\ra{\MFbound}{\UG{\TA}}}

\newcommand{\styleComplexity}[1]{\textsf{#1}}
\newcommand{\EXPSPACE}{\styleComplexity{EXPSPACE}}

\usepackage[most]{tcolorbox}

\usepackage{leftidx}

\makeatletter
\NewDocumentCommand{\mynote}{+O{}+m}{%
  \begingroup
  \tcbset{%
    noteshift/.store in=\mynote@shift,
    noteshift=1.5cm
  }
   \begin{tcolorbox}[nobeforeafter,
    enhanced,
    sharp corners,
    toprule=1pt,
    bottomrule=1pt,
    leftrule=0pt,
    rightrule=0pt,
    colback=yellow!20,
    #1,
    left skip=\mynote@shift,
    right skip=\mynote@shift,
    overlay={\node[right] (mynotenode) at ([xshift=-\mynote@shift]frame.west) {\textbf{Note:}} ;},
    ]
    #2
  \end{tcolorbox}
  \endgroup
  }
\makeatother
\usetikzlibrary{decorations.pathmorphing,shapes}

\usetikzlibrary{arrows,automata,positioning, patterns} %

\tikzstyle{pta}=[auto, ->, >=stealth']
\tikzstyle{every node}=[initial text=]
\tikzstyle{location}=[rectangle, rounded corners, minimum size=12pt, draw=black, fill=blue!10, inner sep=3.5pt]
\tikzstyle{methodBox}=[rectangle, minimum size=12pt, draw=black, fill=green!10, inner sep=3.5pt]
\tikzstyle{location-lossy}=[rectangle, minimum size=12pt, draw=black, fill=brown!10, inner sep=3.5pt]
\tikzstyle{invariant}=[draw=black, dotted, fill=yellow, inner sep=1pt, node distance=0] %
\tikzstyle{locguard}=[fill=cyan!30, inner sep=1pt, node distance=0] %
\tikzstyle{final}=[double, fill=blue!50]

\tikzstyle{urgent}=[fill=yellow, thick, dotted] %
\tikzstyle{private}=[fill=red,thick]

\newcounter{sarrow}

\usepackage{pgfplots}
\pgfplotsset{compat=1.7}

\iffinal
  \newcommand{\todoinline}[2]{}
\else
  \newcommand{\todoinline}[2]{{\color{#1}{\textbf{#2}}}}
\fi

\newcommand{\sj}[1]{\todoinline{cyan!50!black}{SJ: #1}}

\newcommand{\region}{\checkmacros{\ensuremath{r}}}

\usepackage{thmtools}
\usepackage{thm-restate}

\usepackage{mfirstuc}

\newcommand{\defProblem}[3]
{%
\noindent%
  \fcolorbox{black}%
    {white}%
    {
	\begin{minipage}{.95\columnwidth}
		\textbf{#1:}\\
		\textsc{Input}: #2\\
		\textsc{Problem}: #3
	\end{minipage}
}

	\smallskip

}

\title{Parameterized Verification of Timed Networks\\ with Clock Invariants}
\titlerunning{Parameterized Verification of Timed Networks with Clock Invariants}

 \author{\'Etienne Andr\'e}{Universit\'e Sorbonne Paris Nord, LIPN, CNRS UMR 7030, F-93430 Villetaneuse, France \and Institut universitaire de France (IUF) \and \url{https://lipn.univ-paris13.fr/~andre/}}{}{https://orcid.org/0000-0001-8473-9555}{Partially supported by ANR BisoUS (ANR-22-CE48-0012)}
 \author{Swen Jacobs}{CISPA Helmholtz Center for Information Security, Germany \and \url{https://swenjacobs.github.io/}}{jacobs@cispa.de}{https://orcid.org/0000-0002-9051-4050}{}
 \author{Shyam Lal Karra}{CISPA Helmholtz Center for Information Security, Germany}{shyam.karra@cispa.de}{https://orcid.org/0009-0000-6859-4106}{}
 \author{Ocan Sankur}{Universit\'e de Rennes, CNRS, Inria, Rennes, France \and \url{https://people.irisa.fr/Ocan.Sankur/}}{ocan.sankur@irisa.fr}{https://orcid.org/0000-0001-8146-4429}{}
 \authorrunning{\'E. Andr\'e, S. Jacobs, S. Karra and O. Sankur} %

\ccsdesc[500]{Theory of computation~Concurrency}

\keywords{Networks of Timed Automata, Parameterized Verification, Timed Petri Nets} %

\nolinenumbers %

\EventEditors{John Q. Open and Joan R. Access}
\EventNoEds{2}
\EventLongTitle{42nd Conference on Very Important Topics (CVIT 2016)}
\EventShortTitle{CVIT 2016}
\EventAcronym{CVIT}
\EventYear{2016}
\EventDate{December 24--27, 2016}
\EventLocation{Little Whinging, United Kingdom}
\EventLogo{}
\SeriesVolume{42}
\ArticleNo{1}
\begin{document}

\maketitle              

\begin{abstract}	
	We consider parameterized verification problems for networks of timed automata (TAs) based on different communication primitives.
To this end, we first consider disjunctive timed networks (DTNs), i.e., networks of TAs that communicate via location guards that enable a transition only if there is another process in a certain location.
 We solve for the first time the case with unrestricted clock invariants, and establish that the parameterized model checking problem (PMCP) over finite local traces can be reduced to the corresponding model checking problem on a single~TA.
	Moreover, we prove that the PMCP for networks that communicate via lossy broadcast can be reduced to the PMCP for DTNs. 
	Finally, we show that for networks with $k$-wise synchronization, and therefore also for timed Petri nets, location reachability can be reduced to location reachability in DTNs. As a consequence we can answer positively the open problem from Abdulla et al.\ (2018) whether the universal safety problem for timed Petri nets with multiple clocks is decidable.

\end{abstract}

\section{Introduction}
Formally reasoning about concurrent systems is difficult, in particular if correctness guarantees should hold regardless of the number of interacting processes---a problem also known as \emph{parameterized verification}~\cite{AD16,DBLP:reference/mc/AbdullaST18}, since the number of processes is considered a parameter of the system.
Parameterized verification is undecidable in general~\cite{DBLP:journals/ipl/AptK86} and even in very restricted settings, \eg{} for safety properties of finite-state processes with rather weak communication primitives, such as token-passing or transition guards~\cite{DBLP:journals/ipl/Suzuki88,DBLP:conf/cade/EmersonK00}.
A long line of research has identified classes of systems and properties for which parameterized verification is decidable~\cite{DBLP:conf/cade/EmersonK00,DBLP:journals/jacm/GermanS92,EN03,DBLP:conf/lics/EsparzaFM99,DelzannoSZ10,DBLP:series/synthesis/2015Bloem,DBLP:journals/dc/AminofKRSV18,DBLP:journals/dc/EsparzaJRW21}, usually with finite-state processes.

Timed automata (TAs)~\cite{AD94} provide a computational model that combines real-time constraints with concurrency, and are therefore an expressive and widely used formalism to model real-time systems.
However, TAs are usually used to model a constant and \emph{fixed} number of system components.
When the number $n$ of components %
is very large or unknown,
considering their static combination becomes highly impractical, or even impossible if~$n$ is unbounded.
However, there are several lines of research studying networks with a parametric number of timed components (see \eg{} \cite{AJ03,BF13,ADRST16,ADFL19,abdulla_et_al:LIPIcs:2012:3874,AminofRZS15}).

One of these lines considers different variants of timed Petri nets (here, we consider the version defined in~\cite{AbdullaACMT18}), and networks of timed automata with $k$-wise synchronization~\cite{AJ03,DBLP:conf/lics/AbdullaDM04}, a closely related model.
Due to the expressiveness of the synchronization primitive, results for these models are often negative or limited to severely restricted cases. 
For example, in networks of timed automata with a controller process and multiple clocks per process, location reachability is undecidable (even in the absence of clock invariants that could force a process to leave a location)~\cite{DBLP:conf/lics/AbdullaDM04}.
The problem is decidable with a single clock per process and without clock invariants~\cite{AJ03}.
Decidability remains open for location reachability in networks \emph{without} a controller process and with multiple clocks (with or without clock invariants), which is equivalent to the universal safety problem for timed Petri nets that is mentioned as an open problem in~\cite{AbdullaACMT18}.

Another model that has received attention recently and is very important for the work we present is that of \emph{Disjunctive Timed Networks} (\DTN{}s)~\cite{SpalazziS20,AEJK24}.
It combines the expressive formalism of TAs with the relatively weak communication primitive of disjunctive guards~\cite{DBLP:conf/cade/EmersonK00}: transitions can be guarded with a location (called ``guard location''), and such a transition can only be taken by a TA in the network if another process is in that location upon firing.
\textreplace{For example, in the \DTN{} in \cref{fig:async-read-example}, the}{Consider the example in \cref{fig:async-read-example} which illustrates a process’s behavior within an asynchronous communication system, where tasks can be dynamically posted and data is read through shared input channels. The} transition from $\init$ to~$\locreading$ is guarded by
location~$\locpost$: 
for a process to take this transition, at least one other process must be in~$\locpost$.
\begin{wrapfigure}[13]{r}{0.48\textwidth}
\vspace{-.8em}
	\centering
	\scalebox{0.9}{
		\begin{tikzpicture}[font=\footnotesize, xscale=.7]
			\node [] (dummy) at (0,1){};
			\node (qinit) [location, initial text=] at (-1,1) {\init};
			\node (qlisten) [location] at (-1,-1) {$\loclisten$};
			\node (qpost) [location] at (-5.5,-1) {$\locpost$};
			\node (qerror) [location] at (1,-3) {$\locerror$};
			\node (qreading) [location] at (-3,-3) {$\locreading$};
			\node (qdone) [location] at (-7,-3) {$\locdone$};
			\node[invariant, below of=qdone, node distance=0.5cm] (doneinvariant) {$\clock \leq 1$};
			\node[invariant, below of=qreading, node distance=0.5cm] (readinginvariant) {$\clock \leq 2$};
			\path[- stealth]
				(qlisten) edge [] node [above] {$\sigma_1$} node [below]{$\clock = 4, \clock \assign 0$} (qpost);
			\path[-]
			(qdone) edge[] node [above] {} (-7,1);
			\path[->]
			(dummy) edge[]  (qinit);
			\path[- stealth]
			(-7,1) edge[] node [above] {$\sigma_6$} (qinit);
			\path[- stealth]
			(qpost) edge[] node [above,sloped] {$\sigma_2$} (qinit);
			\path[- stealth]
			(qreading) edge[] node [above,sloped] {$\sigma_4$} node [below,pos=0.55,rotate=-50] {$\clock \geq 2, \clock \assign 0$} (qpost);
			\path[- stealth]
			(qreading) edge[] node [above,sloped] {$\sigma_5$} node [below,sloped] {$\clock \geq 1$} (qdone);
			\path[- stealth]
			(qreading) edge[] node [below,sloped] {$\clock < 1$} node[above,pos=0.25]{$\sigmaerr$} 
				node[locguard, above, pos=0.7]{\locdone}
				(qerror);
			
			\path[-]
			(qinit) edge[] node [above] {} (1,-1);
			\path[- stealth]
			(1,-1) edge[] node [above,sloped, pos=0.65] {$\sigma_3$ } node [above,rotate=35,locguard, pos=0.35] {\locpost} node[below,rotate=35]{$\clock \assign 0$} (qreading);
			\path[- stealth]
			(qinit) edge[] node [above,sloped] {$\sigma_0$} node [below,sloped] {$\clock \geq 1$} (qlisten);
			
		\end{tikzpicture}

	}
	\caption{Asynchronous data read example}
	\label{fig:async-read-example}
\end{wrapfigure}%
Parameterized model checking of \DTN{}s was first studied in~\cite{SpalazziS20}, who considered local trace properties in the temporal logic MITL, and showed that the problem can be solved with a cutoff, \ie{} a number of processes that is sufficient to determine satisfaction in networks of any size.
However, their result is restricted to the case when guard locations do not have clock invariants.
This restriction is crucial to their proof, and they furthermore showed that statically computable cutoffs do not exist for the case when TAs can have clock invariants on all locations.

However, the non-existence of cutoffs does not imply that the problem is undecidable.
In~\cite{AEJK24}, \textreplace{we}{the authors} improved the aforementioned results by avoiding the construction of a cutoff system and instead using a modified zone graph algorithm.
Moreover, \textreplace{we}{they} gave sufficient conditions on the TAs to make the problem decidable even in the presence of clock invariants on guard locations.
However, these conditions are semantic, and it is not obvious how to build models that satisfy them;
for instance, our motivating example in \cref{fig:async-read-example} does not satisfy them. 
The decidability of the case without restrictions on clock invariants thus remained open.

In this paper, we show that \textreplace{location reachability and other parameterized verification problems}{properties of finite local traces (and therefore also location reachability)} are decidable for DTNs without restrictions on clock invariants.
Moreover, we show that \textreplace{location reachability in}{checking local trace properties of} systems with lossy broadcast communication~\cite{DelzannoSZ10,ADFL19} or with $k$-wise synchronization can be reduced to \textreplace{location reachability in}{checking local trace properties of} DTNs.
Note that our simulation of these systems by DTNs crucially relies on the power of clock invariants, and would not be possible in the previous restricted variants of DTNs.

To see why checking local trace properties of DTNs with invariants is technically difficult, consider first the easy case from~\cite{SpalazziS20}, where guard locations cannot have invariants.
In this case, it is enough to determine for every guard location $\loc$ the minimal time $\delta$ at which it can be reached: since a process cannot be forced to leave, $\loc$ can be occupied at any time in $[\delta,\infty)$, and transitions guarded by $\loc$ can be assumed to be enabled at any time later than $\delta$.
This is already the underlying insight of~\cite{SpalazziS20}, and in~\cite{AEJK24} it is embedded into a technique that replaces location guards with clock guards $\gclock \geq \delta$, where $\gclock$ is a clock denoting the time elapsed since the beginning.
In contrast, if guard locations can have invariants, a process in $\loc$ can be forced to leave after some time.
Therefore, the set of global times where $\loc$ can be occupied is an arbitrary set of timepoints, and it is not obvious how it can be finitely represented.

\smartpar{Detailed Example}
\add{We introduce an example that motivates the importance of clock invariants in modeling concurrent timed systems, and will be used as a running example.}
\textreplace{We consider an example}{It is} inspired by the verification of asynchronous programs~\cite{Ganty_2012}. 
In this setting, processes can be ``posted'' at runtime to solve a task, and will terminate upon completing the task.
Our example in \cref{fig:async-read-example} features one clock~$\clock$ per process; symbols~$\sigma_i$ and~$\sigmaerr$ are transition labels.
An unbounded number of processes start in the initial location~$\init$.
In the inner loop, a process can move to location \loclisten in order to see whether an input channel carries data.
Once it determines that this is the case, it moves to \locpost, thereby giving the command to post a process that reads the data, and then can return to~$\init$.
In the outer loop, if another process gives the command to read data, \ie{} is in \locpost, then another process can accept that command and move to \locreading.
After some time, the process will either determine that all the data has been read and move to \locdone, or it will timeout and move to \locpost to ask another process to carry on reading.
However, this scheme may run into an error if there are processes in \locdone and \locreading at the same time, modeled by a transition from \locreading to \locerror that can only be taken if \locdone is occupied.

While this example is relatively simple, checking reachability of location \locerror (in a network with arbitrarily many processes) is not supported by any existing technique. 
This is because clock invariants on guard locations are not supported at all by~\cite{SpalazziS20}, and are supported only in special cases (that do not include this example) by~\cite{AEJK24}.
Also other results that could simulate \DTN{}s do not support clock invariants at all \cite{AJ03,ADRST16}.

Moreover, note that clock invariants may be essential for correctness of such systems:
in a system $A^3$, consisting of three copies of the automaton in \cref{fig:async-read-example}, location \locerror is reachable; a computation that reaches \locerror is given in \cref{fig:example-run-gta-client-server}.
However, if we add a clock invariant $\clock \leq 0$ to location \locpost (forcing processes that enter \locpost to immediately leave it again), it becomes unreachable\footnote{To see this, consider the intervals of global time in which the different locations can be occupied: first observe that \locpost in the inner loop can now only be occupied in intervals $[4k, 4k]$ (for $k \in \Nats$), and therefore processes can only move into \locreading at these times. From there, they might move into \locpost after two time units, so overall \locpost can be occupied in intervals $[2k,2k]$ for $k \geq 2$, and \locreading in any interval $[2k,2(k+1)]$. Since clock $\clock$ is always reset upon entering \locreading, \locdone can only be occupied in intervals $[2k+1,2k+1]$ for $k \geq 2$, whereas for a process in \locreading the clock constraint on the transition to \locerror can only be satisfied in intervals $[2k,2k+1)$. Therefore, \locerror is not reachable with the additional clock invariant on \locpost.}.

\begin{figure*}%
    \centering
    \scalebox{0.9}{
\newcommand{\ft}{\scriptstyle}

\begin{tikzpicture}[font=\footnotesize]

    \node (p0-0) [location] at (0,0) {\init};

   

    \node (p0-1) [location] at (2,0) {\loclisten};
    \node (p0-2) [location] at (4,0) {\locpost};
    \node (p0-3) [location] at (12,0) {\locpost};
    \node (p0-0-c) [] at (0, -0.35) {$\ft \clock = 0$};
    \node (p0-1-c) [] at (2, -0.35) {$\ft \clock = 1$};
    \node (p0-2-c) [] at (4, -0.35) {$\ft \clock = 0$};
    \node (p0-3-c) [] at (12, -0.35) {$\ft \clock = 1$};

    \node (p1-0-c) [] at (0, -1.35) {$\ft \clock = 0$};
    \node (p1-1-c) [] at (4, -1.35) {$\ft \clock = 4$};
    \node (p1-2-c) [] at (6, -1.35) {$\ft \clock = 0$};
    \node (p1-3-c) [] at (8, -1.35) {$\ft \clock = 1$};
    \node (p1-4-c) [] at (12, -1.35) {$\ft \clock = 1$};

     \node (p1-0) [location] at (0,-1) {\init};
     \node (p1-1) [location] at (4,-1) {\init};
     \node (p1-2) [location] at (6,-1) {\locreading};
      \node (p1-4) [location] at (8,-1) {\locdone};
      \node (p1-5) [location] at (12,-1) {\locdone};
      
      \node (p2-0-c) [] at (0, -2.35) {$\ft \clock = 0$};
      \node (p2-1-c) [] at (8, -2.35) {$\ft \clock = 5$};
      \node (p2-2-c) [] at (10, -2.35) {$\ft \clock = 0$};
      \node (p2-3-c) [] at (12, -2.35) {$\ft \clock = 0$};
     \node (p2-0) [location] at (0,-2) {\init};
     \node (p2-1) [location] at (8,-2) {\init};
     \node (p2-2) [location] at (10,-2) {\locreading};
      \node (p2-4) [location] at (12,-2) {\locerror};
    \path[->]
        (p0-0) edge []
            node [above] {$\ft 1,\sigma_0$} (p0-1)
    ;
     \path[->]
        (p0-1) edge [] node [above] {$\ft 3, \sigma_1$} (p0-2)
    ;

    \path[->,dotted]
        (p0-2) edge [] node [above] {$\ft 1$} (p0-3)
    ;
    \path[->,dotted]
        (p1-0) edge [] node [above] {$\ft 4$} (p1-1)
    ;
    
    \path[->]
        (p1-1) edge [] node [above] {$\ft \sigma_3$} node [below,locguard] {\locpost} (p1-2)
    ;
    
    \path[->]
       (p1-2) edge []
           node [below] {}  
           node [above] {$\ft 1, \sigma_5$} (p1-4)
    ;
    
    \path[->,dotted]
       (p1-4) edge []
           node [below] {}  
           node [above] {$\ft 0$} (p1-5)
   ;
\path[->,dotted]
        (p2-0) edge [] node [above] {$\ft 5$} (p2-1)
    ;



    \path[->]
        (p2-1) edge [] node [above] {$\ft \sigma_3$} node [below,locguard]{\locpost} (p2-2)
    ;


   

    




    \path[->]
        (p2-2) edge [] node [above] {$\ft \sigma_{\mathsf{err}}$}  node [below,locguard] {\locdone}(p2-4)
    ;

















\end{tikzpicture}}
    \caption{Example of a computation in $\network{3}$ for $\TA$ as depicted by \cref{fig:async-read-example}.}
    \label[figure]{fig:example-run-gta-client-server}
\end{figure*}
\smartpar{Contributions}
We present new decidability results for parameterized verification problems with respect to three different system models as outlined below.
\begin{itemize}
    \item \smartpar{DTNs (\cref{sec:DTNRA})}
    For \DTN{}s, we show that, surprisingly, and despite the absence of cutoffs~\cite{SpalazziS20}, the parameterized model checking problem for \textreplace{local safety properties}{finite local traces} is decidable in the general case, without any restriction on clock invariants.
    Our technique circumvents the non-existence of cutoffs by constructing a modified region automaton, a well-known data structure in timed automata literature, such that communication via disjunctive guards is directly taken into account. In particular, we focus on analyzing the traces of a single (or a finite number of) process(es) in a network of arbitrary size. 
		
    While our algorithm uses some techniques from~\cite{AEJK24}, there are fundamental differences:
    in particular, we introduce a novel abstraction of global time into a finite number of ``slots'', which are elementary intervals with integer bounds, designed to capture the information necessary for disjunctive guard communication. 
    When a transition with a location guard is to be taken at a given slot, we check whether the given guard location appears in the same slot.
    It turns out that such an abstract treatment of the global time is sound: we prove that in this case, one can find a computation that enables the location guard at \emph{any} real time instant inside the given slot. Thus, the infinite set of points at which a location guard is enabled is a computable union of intervals; and we rely on this property to build a finite-state abstraction to solve our problem.
    
    \item  \smartpar{Lossy Broadcast Timed Networks (\cref{section:ltba})}
    We investigate the relation between communication with disjunctive guards and with \emph{lossy broadcast}~\cite{DelzannoSZ10,ADFL19}.
    For finite-state processes, it is known that lossy broadcast can simulate disjunctive guards wrt.\ reachability~\cite{DBLP:journals/corr/abs-2109-08315}, but the other direction is unknown.\footnote{\cite{DBLP:journals/corr/abs-2109-08315} considers IO nets which are equivalent to systems with disjunctive guards. It gives a negative result for a specific simulation relation, but does not prove that simulation is impossible in general.}
    As our second contribution, we establish the decidability of the parameterized model checking problem for 
    local trace properties in timed lossy broadcast networks.
    This result is obtained by proving that communication by lossy broadcast is equivalent to communication by disjunctive guards for networks of timed automata \emph{with} clock invariants for local trace properties. 

    \item \smartpar{Synchronizing Timed Networks and Timed Petri Nets (\cref{sec:stn})}
    Finally, we show that the location reachability problem for controllerless multi-clock timed networks with $k$-wise synchronization can be reduced to the location reachability problem for \DTN{}s with clock invariants.
		
		As a consequence, it follows that the universal safety problem for timed Petri nets with multiple clocks, stated as an open problem in~\cite{AbdullaACMT18}, is also decidable.
\end{itemize}
The proofs of the last two points above involve constructions that require clock invariants on guard locations.
This is why clock invariants are crucial in our formalism, which is a nontrivial extension of\remove{our previous work}~\cite{AEJK24}.
Note that in both cases we get decidability even for variants of the respective system models with clock invariants, which was not considered in~\cite{ADFL19} or~\cite{AbdullaACMT18}. 

For all of the above systems, location reachability can be decided in \EXPSPACE{}.

Due to space constraints, formal proofs of some of our results are deferred to the appendix.

\section{Preliminaries}\label{section:preliminaries}

\LongVersion{\subsection{System Model}}\label{sec:model}
Let $\clocks$ be a set of clock variables, also called \emph{clocks}. A \emph{clock valuation} is a mapping $\clockval : \clocks \rightarrow \mQz$. 
For a valuation~$\clockval$ and a clock~$\clock$, we denote the fractional and integral parts of $\clockval(\clock)$ by $\fractional{\clockval(\clock)}$ and $\integral{\clockval(\clock)}$ respectively.
We denote by~$\mathbf{0}$ the clock valuation that assigns $0$ to every clock, and by $\clockval+\delta$ for $\delta \in \mQz$ the valuation s.t.\ $(\clockval+\delta)(\clock)=\clockval(\clock)+\delta$ for all $\clock \in \clocks$.
Given a subset $\resetclocks \subseteq \clocks$ of clocks and a valuation~$\clockval$, $\clockval[\resetclocks \leftarrow 0]$ denotes the valuation~$\clockval'$ such that $\clockval'(\clock) = 0$ if $\clock \in \resetclocks$ and $\clockval'(\clock) = \clockval(\clock)$ otherwise.
We call \emph{clock constraints} $\clockcons$ the terms of the following grammar:
$\ccons \Coloneqq \top 
\mid \ccons \wedge \ccons 
\mid \clock \sim \intconstant 
\mid \clock \sim \clock' + \intconstant 
$
with $\intconstant \in \Nats$, $\clock,\clock' \in \clocks$ and ${\sim} \in \{<, \leq, =, \geq, > \}.$

A clock valuation $\clockval$ \emph{satisfies} a clock constraint~$\ccons$,
denoted by $\clockval \models \ccons$, if $\ccons$ evaluates to~$\top$ after replacing every $\clock \in \clocks$ with its value~$\clockval(\clock)$.

\begin{definition}
    A \emph{timed automaton (TA)} $\TA$ is a tuple $\TAdefn$
    where $\Loc$ is a finite set of locations with \emph{initial location} $\locinit$, $\clocks$ is a finite set of clocks,
    $\Sigma$ is a finite alphabet that contains a subset $\Sigma^-$ of special symbols, including a distinguished symbol $\epsilon \in \Sigma^-$,
    $\Transitions \subseteq \Loc \times \clockcons \times 2^\clocks \times  \Sigma \times \Loc$ is a transition relation, and
    $\invariant: \Loc \rightarrow \clockcons$ assigns to every location $\loc$ a \emph{clock invariant} $\invariant(\loc)$.
\end{definition}

TAs were introduced in~\cite{AD94} and clock invariants, also simply called invariants, in~\cite{HNSY-1994}.
We assume w.l.o.g.\ that invariants only contain upper bounds on clocks (as lower bounds can be moved into the guards of incoming transitions).
$\Sigma^-$ will be used to label silent transitions and unless explicitly specified otherwise (in \cref{section:ltba,sec:stn}), we assume that $\Sigma^-=\{\epsilon\}$.

\begin{example}
	If we ignore the location guards
        $\locpost$ (from~$\init$ to~$\locreading$)
        and $\locdone$ (from~$\locreading$ to~$\locerror$),
	then the automaton in \cref{fig:async-read-example} is a classical~TA with one clock~$\clock$.
	For example, the invariant of $\locdone$ is $\clock \leq 1$ and the transition from $\init$ to~$\locreading$ resets clock $\clock$.
\end{example}

A \emph{configuration} of a TA $\TA$ is a pair $(\loc,\clockval)$, where $\loc \in \Loc$ and $\clockval: \clocks \rightarrow \mQz$ is a clock valuation.
A \emph{delay transition} is of the form $\tatransition{(\loc,\clockval) }{\delta}{(\loc,\clockval+\delta)}$  for some delay {$\delta \in \mQz $}
such that $\clockval+\delta \models \invariant(\loc)$.
A \emph{discrete transition} is of the form  $\tatransition{(\loc,\clockval) }{\tatranlabel}{(\loc',\clockval')}$, where $\tatran=\tatrandefn \in \Transitions$, $\clockval \models \guard$,
$\clockval' = \clockval[\resetclocks \leftarrow 0]$
and $\clockval' \models \invariant(\loc')$.
A transition $\tatransition{(\loc,\clockval) }{\epsilon}{(\loc',\clockval')}$ is called an \emph{$\epsilon$-transition}.
We write $(\loc,\clockval)\xrightarrow{\delta,\tatranlabel}(\loc',\clockval')$ if there is a delay transition $\tatransition{(\loc,\clockval)}{\delta}{(\loc,\clockval+\delta)}$ followed by a discrete transition $\tatransition{(\loc,\clockval+\delta)}{\tatranlabel}{(\loc',\clockval')}$.

A \emph{timed path} of~$\TA$ is a finite sequence of transitions
 $\lcomputation = \lcomputationdefn$.
For a  timed path $\lcomputation = \lcomputationdefn$, 
let $\totaltime(\lcomputation) = \sum_{0 \leq i < l} \delta_i$ be the \emph{total time delay} of~$\lcomputation$. 
The \emph{length} of $\lcomputation$ is $2l$.
A configuration $(\loc,\clockval)$ has a \emph{timelock} if there is $b \in \mQz$ s.t.\ $\delta(\lcomputation) \leq b$ for every timed path $\lcomputation$
starting in~$(\loc,\clockval)$.
We write $(\loc_0,\clockval_0) \rightarrow^* (\loc_l,\clockval_l)$ if there is a  timed path $\lcomputation = \lcomputationdefn$;
$\lcomputation$ is a \emph{computation} if $\loc_0=\locinit$ and $\clockval_0=\mathbf{0}$.

The \emph{trace} of the timed path $\lcomputation$ is the sequence of pairs of delays and labels obtained by removing transitions with a label from $\Sigma^-$ and
adding the delays of these to the following transition (see \cref{appendix:trace}).
The \emph{language} of~$\TA$, denoted $\Lg(\TA)$, is the set of  traces of all of its computations.

We now recall guarded timed automata as an extension of timed automata with location guards, that will allow, in a network, to test whether some other process is in a given location in order to pass the guard.

\begin{definition}[Guarded Timed Automaton (\gta)]\label{definition:gTA}
A \emph{\gta} $\TA$ is a tuple $(\Loc,\locinit,\clocks,\Sigma, \Transitions,\invariant)$,
where $\Loc$ is a finite set of locations with \emph{initial location} $\locinit$,
    $\clocks$ is a finite set of clocks,
    $\Sigma$ is a finite alphabet that contains a subset $\Sigma^-$ of special symbols, including a distinguished symbol $\epsilon \in \Sigma^-$,
    $\Transitions \subseteq \Loc \times \clockcons \times 2^\clocks \times \Sigma \times  \left( \Loc \cup \{\top\} \right) \times \Loc$ is a transition relation, and
    $\invariant: \Loc \rightarrow \clockcons$ assigns to every location $\loc$ an \emph{invariant} $\invariant(\loc)$.
\end{definition}

Intuitively, a transition $\gtatran = \gtatrandefn \in \gtasetoftransitions$ takes the automaton from location $\loc$ to~$\loc'$;
$\gtatran$~can only be taken if \emph{clock guard}~$\guard$ and \emph{location guard}~$\locguard$ are both satisfied, and it resets all clocks in~$\resetclocks$.
Note that satisfaction of location guards is only meaningful in a \emph{network} of~TAs (defined below).
Intuitively, a location guard~$\locguard$ is satisfied if it is~$\top$ or if another automaton in the network currently occupies location~$\locguard$.
We say that $\locguard$ is \emph{trivial} if $\locguard = \top$.
We say location~$\loc$ \emph{has no invariant} if $\invariant(\loc)=\top$.

\begin{example}
In the \gta{} in \cref{fig:async-read-example}, the transition from~$\init$ to~$\locreading$ is guarded by location guard~$\locpost$.
The transition from~$\init$ to~$\loclisten$ has a trivial location guard (trivial location guards are not depicted in our figures).
Location~$\init$ has no invariant.
\end{example}

\begin{definition}[Unguarded Timed Automaton]
    Given a \gta{}~$\TA$, we denote by $\UG{\TA}$ the \emph{unguarded} version of~$\TA$, which is the TA obtained from~$\TA$ by removing location guards, and adding
    a fresh clock $\gclock$, called the global clock, that does not appear in the guards or resets.
    Formally, $\UG{\TA} = (\Loc,\locinit,\clocks{\cup}\{\gclock\},\Transitions',\invariant)$ with
    $\Transitions' = \big\{(\loc, \guard, \resetclocks, \tatranlabel, \loc') \mid \gtatrandefn \in \Transitions \big\}$.
\end{definition}

 For a \gta{} $\TA$,
 let $\TA^n$ denote a \emph{network of guarded timed automata}~(\ngta), consisting of $n$~copies of~$\TA$.
 Each copy of~$\TA$ in\LongVersion{ the \ngta}~$\TA^n$ is called a \emph{process}.

A \emph{configuration} $\nConfig$ of an \ngta
$\network{n}$ is a tuple
$\dtnconfig$, where every $(\loc_i,\clockval_i)$ is a configuration of~$\TA$.
The semantics of $\network{n}$ can be defined as a \emph{timed transition system} $(\mathfrak{C}, \hat{\nConfig},T)$, where $\mathfrak{C}$ 
denotes the set of all  configurations of $\network{n}$, $\hat{\nConfig}$ is the unique initial configuration $(\locinit,\mathbf{0})^n$, and the transition relation $T$ is the union of the following delay and discrete transitions:

\begin{description}[noitemsep,topsep=0pt]
    \item[delay transition] $ \transition {\dtnconfig} {\delta} {\dtnconfigdelay}$
    if $\forall i \in \{1, \dots, n\} :  \clockval_i+\delta \models \invariant(\loc_i)$, \ie{} we can delay $\delta \in \mQz $ units of time if all clock invariants are satisfied at the end of the delay.

        \item[discrete transition] $\transition{\dtnconfig } {(i, \tatranlabel)} {\dtnconfigsucc }$
        for some $i \in \{1, \dots, n\}$
    if
        \begin{oneenumerate}%
            \item $\transition{(\loc_i,\clockval_i)} {\tatranlabel} {(\loc_i',\clockval_i')}$ is a discrete transition of $\TA$ with             
            $\tatran = (\loc_i,\guard,\resetclocks,\tatranlabel,\locguard,\loc_i')$,
            \item $\locguard=\top$ or $\loc_j=\locguard$ for some $j \in \{1, \dots, n\}\setminus \{i\}$, and
            \item $\loc_j' = \loc_j$ and $\clockval_j'=\clockval_j$ for all $j \in \{1, \dots, n\}\setminus \{i\}$.
        \end{oneenumerate}%

That is, location guards $\locguard$ are interpreted as disjunctive guards: unless $\locguard=\top$, at least one other process needs to occupy location~$\locguard$ in order for process~$i$ to pass this guard.
\end{description}
 
 {We write $\nConfig \xrightarrow{\delta, (i,\tatranlabel)} \nConfig''$ for a delay transition $\nConfig \xrightarrow{\delta} \nConfig'$ followed by a discrete transition $\nConfig' \xrightarrow{(i,\tatranlabel)} \nConfig''$.
Then, a \emph{timed path} of~$\TA^n$ is a finite sequence $\gcomputation = \gcomputationdefn$.}

For a timed path $\gcomputation=\gcomputationdefn$, 
let $\totaltime(\gcomputation) = \sum_{0 \leq i < l} \delta_i$ be the total time delay of~$\gcomputation$.
The definition of timelocks extends naturally to configurations of~\ngtas.
A timed path $\gcomputation$ of~$\TA^n$ is a \emph{computation} if $\nConfig_0 = \nConfigInit$.
Its \emph{length} is equal to $2l$.%

We write $\loc \in \nConfig$ if $\nConfig=\big((\loc_1,\clockval_1),\ldots,(\loc_n,\clockval_n)\big)$ and $\loc = \loc_i$ for some $i \in \{1, \dots, n\}$, and similarly $(\loc,\clockval) \in \nConfig$.
We say that a location $\loc$ is \emph{reachable} in~$\TA^n$ if there exists a reachable configuration $\nConfig$ s.t.\ $\loc \in \nConfig$.

\begin{example}\label{example:computation}
    Consider the \ngta $\network{3}$ where $\TA$ is the \gta shown in \cref{fig:async-read-example}. 
    A computation $\gcomputation$ of this network is depicted in \cref{fig:example-run-gta-client-server}, in which a process reaches $\mathsf{error}$ with $\totaltime(\gcomputation)=5$. 
		The computation is 
            $\big((\mathsf{init},\clock=0),(\mathsf{init},\clock=0),(\mathsf{init},\clock=0)\big)$
            $\xrightarrow {1,(1,\sigma_0)}  \big((\mathsf{listen},\clock=1),(\mathsf{init},\clock=1),(\mathsf{init},\clock=1)\big)$
            $\xrightarrow {3,(1,\sigma_1)}  \big((\mathsf{post},\clock=0),(\mathsf{init},\clock=4),(\mathsf{init},\clock=4)\big)$
            $\xrightarrow {0,(2,\sigma_3)}  \big((\mathsf{post},\clock=0),(\mathsf{reading},\clock=0),(\mathsf{init},\clock=4)\big)$ 
            $\xrightarrow {1,(2,\sigma_5)}  \big((\mathsf{post},\clock=1),(\mathsf{done},\clock=1),(\mathsf{init},\clock=5)\big)$
            $\xrightarrow {0,(3,\sigma_3)}  \big((\mathsf{post},\clock=1),(\mathsf{done},\clock=1),(\mathsf{reading},\clock=0)\big)$
            $\xrightarrow {0,(3,\sigma_{\mathsf{err}})}  \big((\mathsf{post},\clock=1),(\mathsf{done},\clock=1),(\mathsf{error},\clock=0)\big)$.
            Therefore, $\mathsf{error}$ is reachable in~$\network{3}$.
\end{example}

The \emph{trace} of the timed path $\gcomputation$ is a sequence 
$\traceof{\gcomputation} = (\delta_0', (i_0', \sigma_0')) \ldots (\delta_{l-1}', (i_{l-1}', \sigma_{l-1}'))$
obtained by removing all discrete transitions $(j,\sigma_j)$ of $\gcomputation$ with $\sigma_j \in \Sigma^-$, and adding all delays of these transitions to the following discrete transition, yielding the~$\delta_j'$.
The \emph{language} of~$\TA^n$, denoted $\Lg(\TA^n)$, is the set of traces of all of its computations.

\begin{example}\label{lbl:example-trace}
For the computation $\gcomputation$ in \cref{example:computation},  $\traceof{\gcomputation}=\big(1,(1,\sigma_0)\big), \big(3,(1,\sigma_1)\big), \big(0,(2,\sigma_3)\big),
\big(1,(2,\sigma_5)\big),\big(0,(3,\sigma_{3})\big), \big(0,(3,\sigma_{\mathsf{err}})\big)$.

\end{example}
We will also use \emph{projections} of these global objects onto subsets of the processes.
That is, if $\nConfig=\big((\loc_1,\clockval_1),\ldots,(\loc_n,\clockval_n)\big)$ and $\intervalPr = \intervalPrdefn \subseteq \{1, \dots, n\}$, then $\proj{\nConfig}{\intervalPr}$ is the tuple $\big((\loc_{i_1},\clockval_{i_1}),\ldots,(\loc_{i_k},\clockval_{i_k})\big)$, and we extend this notation to computations 
$\proj{\gcomputation}{\intervalPr}$ by keeping only the discrete transitions of $ \intervalPr$ and by adding the delays of the removed discrete transitions to the delay of the following discrete transition of $\intervalPr$ (see \cref{appendix:NTAs} for a full definition).

We introduce a special notation for projecting to a single process and define, for any natural number $1\leq a \leq n$, $\proj{\gcomputation}{a}$ a computation of $\UG{A}$,
obtained from~$\proj{\pi}{\{a\}}$ by discarding the index $a$ from all transitions; 
that is, $\proj{\gcomputation}{a}$ has the form
$(\loc_0, \clockval_0) \xrightarrow{(\delta_{k_0}',\sigma_{k_0})}
\ldots
\xrightarrow{(\delta_{k_m}', \sigma_{k_m})} (\loc_{m+1}, \clockval_{m+1})$.
We also extend this to traces; that is, $\proj{\traceof{\gcomputation}}{a} = (\delta_{k_0}',\sigma_{k_0})\ldots (\delta_{k_m}', \sigma_{k_m})$, which is a trace of $\UG{A}$.
For a set of traces $L$, and set $\intervalPr$ of processes, we write 
$\proj{L}{\intervalPr} = \{ \proj{\trace}{\intervalPr} \mid \trace \in L\}$.

\add{Note that the projection of a computation is not necessarily a computation itself, since location guards may not be satisfied.} 

\begin{example}\label{lbl:example-trace-projection}
For the computation $\gcomputation$ in \cref{example:computation},   $\proj{\gcomputation}{3}=
		(\mathsf{init},\clock=0)
		\xrightarrow {(5,\sigma_3)} (\mathsf{reading},\clock=0)
		\xrightarrow {(0,\sigma_{\mathsf{err}})}  (\mathsf{error},\clock=0)$ and $\proj{\traceof{\gcomputation}}{3}= (5,\sigma_3), (0,\sigma_{\mathsf{err}})$.
\end{example}

A \emph{prefix} of a computation
$\gcomputation=\gcomputationdefn$, is a sequence
$\nConfig_0 \xrightarrow{\delta_0, (i_0, \sigma_0)} \cdots \xrightarrow{\delta_{l'}, (i_{l'}, \sigma_{l'})} \nConfig_{l'}$
with $l' \leq l-1$.
If $\gcomputation$ is a timed path and $d \in \mQz$, then $\gcomputation^{\leq d}$ 
denotes the maximal prefix of~$\gcomputation$ with $\totaltime(\gcomputation^{\leq d})\leq d$
, and similarly for timed paths $\lcomputation^{\leq d}$ of a single \gta{}.
For timed paths $\gcomputation_1$ of $\network{n_1}$ and $\gcomputation_2$ of $\network{n_2}$ with $\totaltime(\gcomputation_1) = \totaltime(\gcomputation_2)$,
we denote by $\gcomputation_1 \parallel \gcomputation_2$ their \emph{composition} into a timed path of $\TA^{n_1+n_2}$ 
whose projection to the first $n_1$ processes is~$\gcomputation_1$, and whose projection to the last~$n_2$ processes is~$\gcomputation_2$ (see \cref{appendix:NTAs}).

\begin{definition}[\disjunctivetimednetwork{}]
A given \gta{}~$\TA$ induces a \emph{disjunctive timed network (\DTN{})} 
$\TA^\infty$, defined as the following family of~\ngtas:
$\TA^\infty = \{\network{n} \mid n \in \Nats_{>0}\}$
(we follow the terminology and use abbreviations of~\cite{SpalazziS20}).
We define $\calL(\TA^\infty) = \bigcup_{n \in \Nats_{>0}} \calL(\network{n})$ and consider 
$\proj{\calL(\TA^\infty)}{I} = 
\bigcup_{n \in \Nats_{>0}} \proj{\calL(\network{n})}{I}$.
\end{definition}

\subsection{The Parameterized Model Checking Problem}
\label{sec:paramver}

We formalize properties of \DTN{}s as sets of traces that describe the intended behavior of a fixed number of processes running in a system with arbitrarily many processes.
That is, a \emph{local property} $\varphi$ 
of $k$ processes, also called a \emph{$k$-indexed property}, is a subset of $(\mQz \times ([1,k] \times \Sigma))^*$.
For $k=1$, for simplicity, we consider it as a subset of $(\mQz \times \Sigma)^*$.
We say that $\network{n}$ \emph{satisfies} a $k$-indexed local property $\varphi$, denoted $\network{n} \models \varphi$, if $\proj{\calL(\network{n})}{[1,k]} \subseteq \varphi$.
Note that, due to the symmetry of the system, it does not matter \emph{which} $k$ processes we project $\calL_{}(\network{n})$ onto, so we always project onto the first~$k$.

\defProblem
	{Parameterized model checking problem (PMCP)}
	{a \gta $\TA$ and a $k$-indexed local property $\varphi$}
	{Decide whether $\network{n} \models \varphi$ holds $\forall n\geq k$.}

Local trace properties allow to specify for instance any local safety property of a single process (with $I=[1,1]$), as well as mutual exclusion properties (with $I=[1,2]$) and variants of such properties for larger~$I$.

PMCP can be solved by checking whether $\proj{\calL(\network{\infty})}{[1,k]} \subseteq \varphi$.
Our solution consists in building a TA that recognizes
$\proj{\calL(\TA^\infty)}{[1,k]}$.
Note that language inclusion is undecidable on TAs~\cite{AD94}\add{, but many interesting problems are decidable.
These include MITL model checking~\cite{DBLP:reference/mc/BouyerFLMO018} and simpler problems such as} reachability:
given symbol $\sigma_0 \in \Sigma$,
the \emph{reachability PMCP} is the PMCP where
$\varphi$ is the set of traces that contain an occurrence of~$\sigma_0$.
Reachability of a \emph{location} of~$\TA$ can be solved by PMCP by choosing appropriate transition labels.

\vspace{-0.1cm}
\begin{example}
    In the example of \cref{fig:async-read-example},
    a natural local property we are interested in is the reachability of the label $\sigmaerr$.
    Formally, the local property for process 1 can be written as a 1-indexed property:
        $(\mQz {\times} ([1,1] {\times} \Sigma))^* \cdot \{(d, (1, \sigmaerr)) \mid d \in \mQz\} \cdot (\mQz {\times} ([1,1] {\times} \Sigma))^*$.

\end{example}
\vspace{-0.1cm}

\section{Model Checking \DTN{}s}\label{sec:DTNRA}

\subsection{Definitions}\label{ss:slots}
We recall here the standard notions of regions and region automata,
and introduce the \emph{slots} of regions which refer to the intervals of possible valuations of a global clock.

\smartpar{Regions}
Given~$\TA$, %
for all $\clock \in \clocks$,
let $\bound{\clock}$ denote
the maximal bound that $\clock$ is compared to:
$\bound{\clock}= \max \{ \intconstant \in \mathbb{Z} \mid $ 
``$\clock \sim \intconstant$'', ``$\clock - \clock' \sim \intconstant$'',
``$\clock' - \clock \sim \intconstant$'' appears in a guard or invariant of $\TA \}$(we set $\bound{\clock} = 0$ if this set is empty).
$\Mbound$ is called the \emph{maximal bound function for $\TA$}.
Define $\boundmax= \max \{\bound{\clock} \mid \clock \in \clocks\}$.
We say that two valuations $\clockval$ 
and $\clockval'$ are equivalent \wrt{} $\Mbound$, denoted by $\clockval \regionequalto{\Mbound} \clockval'$, if the following conditions hold for any clocks $\clock,\clock'$~\cite{PB-CD-EF-AP-cav2000,bouyer2004updatable}:

\begin{enumerate}[noitemsep,topsep=0pt]
\item either 
 $\integral{\clockval(\clock)} = \integral{\clockval'(\clock)}$ or
 $\clockval(\clock) > \bound{\clock}$ and $\clockval'(\clock) > \bound{\clock}$;
 \item if $\clockval(\clock),\clockval'(\clock) \leq \bound{\clock}$ then $\fractional{\clockval(\clock)} =0 \Longleftrightarrow \fractional{\clockval'(\clock)}=0$;
 \item if $\clockval(\clock)\leq \bound{\clock},\clockval(\clock') \leq \bound{\clock'}$ {then}
 $\fractional{\clockval(\clock)} \leq \fractional{\clockval(\clock')}$
 {$\Longleftrightarrow $} $\fractional{\clockval'(\clock)} \leq \fractional{\clockval'(\clock')}$;
 \item for any interval $I$ among $(-\infty,-\bound{\clock'}), [-\bound{\clock'},-\bound{\clock'}],
    (-\bound{\clock'},-\bound{\clock'}+1),
 \ldots, [\bound{\clock},\bound{\clock}], (\bound{\clock},\infty)$, we have
 $\clockval(\clock)-\clockval(\clock') \in I \Longleftrightarrow $
 $\clockval'(\clock)-\clockval'(\clock') \in I$.
\end{enumerate}

An \emph{$\Mbound$-region} is an equivalence class of valuations induced by $\regionequalto{\Mbound}$. We denote by $\regionof{\clockval}{\Mbound}$ the region to which $\clockval$ belongs.
We omit~$\Mbound$ when it is clear from  context.

For an $\Mbound$-region $\region$, if a valuation $\clockval \in \region$ satisfies a clock guard $\guard$, then every valuation in $\region$ satisfies~$\guard$.
We write $\region \models \guard$ to mean that every valuation in $\region$ satisfies~$\guard$.

Given an $\Mbound$-region~$\region$ and a clock~$\clock$, let $\proj{\region}{\clock}$ denote
the projection of the valuations of~$\region$ to~$\clock$,
\ie{} $\proj{\region}{\clock} = \{ \clockval(\clock) \mid \clockval \in \region\}$.
Given a valuation~$\clockval$ and a clock~$\clock \in \clocks$, let $\eliminate{\clockval}{\clock}$ denote
\emph{the projection of $\clockval$ to the clocks other than~$\clock$},
\ie{} $\eliminate{\clockval}{\clock}\colon \clocks \setminus \{\clock\} \rightarrow \mQz$
is defined by $\eliminate{\clockval}{\clock}(\clock')=\clockval(\clock')$ for all $\clock' \in \clocks \setminus \{c\}$.
By extension, given a region~$\region$ and a clock~$\clock$, let 
$\eliminate{\region}{\clock}= \{ \eliminate{\clockval}{\clock} \mid \clockval \in \region \}$.

\smartpar{\RegionAutomaton{}}
The \emph{\regionautomaton{}} of a TA~$\TA$ is a finite automaton with alphabet $\Sigma \cup \{\radelaylabel \}$,
denoted by $\ra{\Mbound}{\TA}$, defined as follows.

The \emph{region states} are pairs $(\loc,\region)$ where $\loc \in \Loc$ and $\region$ is an $\Mbound$-region.
The \emph{initial region state} is $\regionstateinit$ %
where $\locinit$ is the initial location of~$\TA$ and $\regioninit$ is the singleton
region containing $\initialvaln$.

There is a transition $\rgtransition {(\loc,\region)}{\radelaylabel}{(\loc,\region')}$ in $\ra{\Mbound}{\TA}$ iff
there is a transition $\tatransition{(\loc,\clockval)}{\tadelay}{(\loc,\clockval’)}$ in~$\TA$ for some
$\tadelay \in \mQz, \clockval \in \region$ and $\clockval’ \in \region’$. 
We say that~$\region'$ is a \emph{time successor} of~$\region$.
Note that we can have $\region'=\region$.
Furthermore, $(\loc,\region')$ is the \emph{immediate time successor} of $(\loc,\region)$ if
$\rgtransition {(\loc,\region)}{\radelaylabel}{(\loc,\region')}$,
$\region'\neq \region$,
and whenever $\rgtransition {(\loc,\region)}{\radelaylabel}{(\loc,\region'')}$,
we have 
$\rgtransition {(\loc,\region')}{\radelaylabel}{(\loc,\region'')}$.

There is a transition $\rgtransition{(\loc,\region)}{\tatranlabel}{(\loc',\region')}$ in $\ra{\Mbound}{\TA}$ iff
there is a transition $\tatransition{(\loc,\clockval)}{\tatranlabel}{(\loc’,\clockval’)}$ with label~$\tatranlabel$ in~$\TA$
for some
$\clockval \in \region$ and $\clockval’ \in \region’$.
We write $\genrgtransition{(\loc,\region)}{(\loc',\region')}$ if either 
$\rgtransition {(\loc,\region)}{\radelaylabel}{(\loc',\region')}$ or $\rgtransition{(\loc,\region)}{\tatranlabel}{(\loc',\region')}$ for some~$\tatranlabel \in \Sigma$.

A \emph{path} in $\ra{\Mbound}{\TA}$ is a finite sequence of  transitions 
$\rho_r = (\loc_0,\region_0)\xrightarrow{\sigma_0} \ldots \xrightarrow{\sigma_{n-1}} (\loc_n,\region_n) $ for some $n \geq 0$
where $\sigma_i \in \Sigma \cup\{\radelaylabel\}$.
A path of $\ra{\Mbound}{\TA}$ is a \emph{computation} if it starts from the initial region state.

It is known that $\ra{\Mbound}{\TA}$ captures the \emph{untimed} traces of~$\TA$,
\ie{} the projection of the traces of~$\TA$ to $\Sigma$~\cite{AD94}.

\smartpar{Slots}
Now, we can introduce slots.
We will show later that slots are a sufficiently precise abstraction of time for \DTN{}s.
In this paragraph, we assume that TAs have a distinguished \emph{global clock} $\gclock$ which is never reset and does not appear in clock guards.
We will thus consider a clock set $\clocks \cup \{\gclock\}$
(making $\gclock$ appear explicitly for clarity).

Let $N_\TA$ denote the number of pairs $(\loc, \region)$
where $\loc \in \Loc$ and $\region$ is an $\Mbound$-region (thus a region on the clock set $\clocks$ without~$\gclock$).
Recall that $N_\TA$ is exponential in $\lvert \clocks \rvert$~\cite{AD94,bouyer2004updatable}.
Let us consider a bound function~$\MFbound :\clocks \cup \{\gclock\} \rightarrow \mathbb{N}$ for~$\TA$
such that for $\clock \in \clocks \setminus \{\gclock\}$,
$\MFbound(\clock)=\bound{\clock}$,
and~$\MFbound(\gclock)=2^{N_\TA+1}$.
Throughout the paper, the bound functions will be denoted by $\MFbound(\cdot)$
whenever the clock set contains the distinguished global clock $\gclock$, and $\bound{\cdot}$ otherwise.
The former will be referred to as $\MFbound$-regions, and the latter as $\Mbound{}$-regions.

We define the \emph{slot of an $\MFbound$-region $\region$} as $\slotof(\region)=\proj{\region}{\gclock}$.
It is known that for any region~$\region$ (with any bound function)
and clock~$\clock$,
$\proj{\region}{\clock}$ is an interval~\cite{HerbreteauKSW11}.
Moreover, if $\clockval(\clock)$ for every $\clockval \in \region$ is below the maximal constant $\MFbound(\clock)$, then $\proj{\region}{\clock}$ is either a singleton interval of the form $[k,k]$, or an open interval of the form $(k,k+1)$ for some $k \in \Nats$.

For a slot~$\slot$, let us define $\nextslot(\slot)$ as follows.
\begin{inparaenum}%
    \item if $\slot=(k,k+1)$ for some $k \in \Nats$, then $\nextslot(\slot)=[k+1,k+1]$;
    \item if $\slot=[k,k]$ and $k < \MFbound(\gclock)$, then
    $\nextslot(\slot)=(k,k+1)$;
    \item if $\slot=[\MFbound(\gclock), \MFbound(\gclock)]$, then
    $\nextslot(\slot)=(\MFbound(\gclock),\infty)$.
    \item if $\slot=(\MFbound(\gclock),\infty)$ then $\nextslot(\slot)=\slot$.
\end{inparaenum}

We define the \emph{slot of a valuation}~$\clockval$ on $\clocks \cup\{\gclock\}$ as $\slotof(\region)$ where $\region$ is the (unique) $\MFbound$-region s.t.~$\clockval \in \region$.
Slots, seen as intervals, can be bounded or unbounded.

\begin{example}\label{example:slot-next}
    Consider the clock set $\{x,y,\gclock\}$ and the region $\region$ defined by
    $\intfloor{x} = \intfloor{y} = 1$, $\intfloor{\gclock} = 2$,
    $0<\fractional{x} < \fractional{y} < \fractional{\gclock}<1$ (with $M^\nearrow(\cdot)=4$ for all clocks).
    Then, $\slotof(\region) = (2,3)$.

    As a second example, assume $\MFbound(x) = 2$, $\MFbound(y) = 3$ and $\MFbound(\gclock)$ is, say, $2048$.
    Consider the region $\region'$ defined by
    $x > 2 \land \intfloor{y} = 1 \land 0 < \fractional{y} < 1 \land \intfloor{\gclock} = 2048 \land \fractional{\gclock} = 0$.
    Then, $\slotof(\region) = [2048,2048]$.
    In addition, $\nextslot(\slotof(\region)) = (2048, \infty)$.
\end{example}

We now introduce the \emph{shifting} operation which consists of increasing the global clock value,
without changing the values of other clocks.
\begin{lemma}\label{lemma:slots}
    Given any $\MFbound$-region~$\region$ 
    and $k\in\Ints$ such that ${\sup}\big(\slotof(\region)\big)+k \leq \MFbound(\gclock)$,
    and ${\inf}\big(\slotof(\region)\big)+k \geq 0$,
    there exists a $\MFbound$-region
    $\region'$ which satisfies $\slotof(\region') = \slotof(\region)+k$ and $\eliminate{\region'}{\gclock} = \eliminate{\region}{\gclock}$,
    and $\region'$ can be computed in polynomial time in the number of clocks.
\end{lemma}

The region~$\region'$ in \cref{lemma:slots} will be denoted by $\newslot{\region}{k}$.
We say that it is obtained by \emph{shifting the slot by~$k$ in $\region$}.
We extend this notation to sets of regions and sets of region states,
that is, $\newslot{W}{k} = \{ (\loc,\newslot{\region}{k}) \mid (\loc,\region) \in W\}$
where~$W$ is a set of region states.
For a set of region states $W$, we define $\eliminate{W}{\gclock}=\{(\loc,\eliminate{\region}{\gclock}) \mid (\loc,\region) \in W\}$.

\begin{example}
    Consider the clock set $\{x,y,\gclock\}$ and the region $\region$ defined in \cref{example:slot-next} satisfying
    $\intfloor{x} = \intfloor{y} = 1$, $\intfloor{\gclock} = 2$,
    $0<\fractional{x} < \fractional{y} < \fractional{\gclock}<1$
    (with $M^\nearrow_{max}=4$).
    Then, $\slotof(\region) = (2,3)$,
    and $\newslot{\region}{1}$ is defined by the same constraints as above
    except that $\intfloor{\gclock} = 3$, and $\slotof(\newslot{\region}{1})=(3,4)$.
\end{example}

\begin{remark}
    \label{remark:projected-regions-size}
    Recall that given a bound function, the number of regions is
    $O(|\clocks|! 2^{|\clocks|}\boundmax^{|\clocks|})$ 
    since regions determine an order of the fractional values of clocks, the subsets of clocks that have integer values, 
    and an integral value for each clock~\cite{PB-CD-EF-AP-cav2000}.
    The number of 
    $\MFbound$-regions is $O\left(|\clocks|!2^{|\clocks|} (\MFbound{(\gclock)})^{|\clocks|}\right)$,
    which is doubly exponential in $|\clocks|$ since $\MFbound(\gclock)$ is. %
    
    Crucial to our paper, however, is that the set of \emph{projections} $\eliminate{\region}{\gclock}$
    of the set of $\MFbound$-regions $\region$ has size exponential only. This can be seen as follows:
    our definition of regions from \cite{PB-CD-EF-AP-cav2000} uses a distinct maximum bound function for each clock.
    Thus, when constraints on $\gclock$ are eliminated, there only remain constraints on clocks $\clock \in \clocks\setminus\{\gclock\}$,
    with maximal constants $\Mbound(\clock)$ as in the original \gta $\TA$. We thus fall back to the set of regions of $\TA$
    of size $O(|\clocks|!2^{|\clocks|} \boundmax^{|\clocks|})$.
\end{remark}

\subsection{Layer-based Algorithm for the DTN Region Automaton}
\label{ss:MC-local}

We describe here an algorithm to compute a TA~$\da$ that recognizes the language $\proj{\calL(\TA^\infty)}{1}$.
We explain at the end of the section how to generalize the algorithm to compute $\proj{\calL(\TA^\infty)}{I}$ for an interval $I=[1,a]$ for~$a\geq 1$.

\begin{assumption}
    \label{assumption:nonblocking-2}
    We assume that the given \gta{} $\TA$ is \emph{timelock-free}, regardless of location guards.
		Formally, let $\TA'$ be obtained from~$\TA$ by removing all transitions with non-trivial location guards. 
		We require that no configuration of $\TA'$ has a timelock.
\end{assumption}

\begin{wrapfigure}{r}{0.27\textwidth}
\vspace{-1em}
    \centering
	\scalebox{0.9}{
	\begin{tikzpicture}[font=\footnotesize]
    \node (qinit) [location, initial text=] at (0,0) {$\locinit$};
    \node (q1) [location] at (3,0) {$\loc_1$};

    \node (q1inv) [invariant, below=of q1] {$\clock \leq 1$};

    \path[->]
        ($(qinit)+(0,0.5)$) edge (qinit)

        (qinit) edge  
            node [above] {$\clock \leftarrow 0$} (q1)
            
        (q1) edge [bend right=22]
						node [above, locguard,pos=0.5] {$\locinit$} (qinit)
    ;

\end{tikzpicture}
	}
	\caption{A \gta{} with timelock due to location guards.\vspace{-.5cm}}
	\label{figure:livelock} 	
\end{wrapfigure}

Note that this assumption 
guarantees that $\network{n}$ will be timelock-free for all~$n$.
\remove{Vice versa, if $A'$ is not timelock-free, then there are also timelocks in some $A^n$, specifically for $n=1$.}
Assuming timelock-freeness is not restrictive since a protocol cannot possibly block the physical time: time will elapse regardless of the restrictions of the design.
A network with a timelock is thus a design artifact, and just means the model is incomplete.
An incomplete model can be completed by adding a sink location to which processes that would cause a timelock can move, and regarding reachability the resulting model is equivalent to the original one.

\begin{example}\label{example:figure:livelock}

	The \gta{} in \cref{figure:livelock} does not satisfy \cref{assumption:nonblocking-2}, since $\TA'$ (where we remove transitions with non-trivial location guards) has a timelock at $(\loc_1, \clock =1)$.
\end{example}

The following assumption simplifies the proofs:
\begin{assumption}
    \label{assumption:unique-sigma}
    Each transition of \gta{} $\TA$ is labeled by a unique label different from~$\epsilon$.
\end{assumption}

Consider a \gta{}~$\TA$.
Our algorithm builds a TA capturing the language $\proj{\calL(\TA^\infty)}{1}$. 
The construction is based on $\MFbound$-region states
of $\UG{A}$; however, not all transitions of the region
automaton of $\UG{A}$ are to be added since location guards mean that some transitions are not enabled at a given region.
Unless otherwise stated, by region states we mean $\MFbound$-\regionstates{}. %
The steps of the construction are illustrated in \cref{fig:data-structures}.
From~$\TA$, we first obtain $\UG{\TA}$, and build the region automaton for~$\UG{A}$,
denoted by $\raUGA$. Then \cref{algo:dtnregionautomaton} builds the so-called \emph{\dtnregionautomaton{}} $\DRA{\TA}$ which is a finite automaton.
Finally we construct $\da$ which we refer to as the \emph{\summaryautomaton{}}, a timed automaton derived from $\DRA{\TA}$ by adding clocks and clock guards to $\DRA{\TA}$. 
Our main result is that $\da$  recognizes
the language $\proj{\Lg(A^\infty)}{1}$. %

\begin{figure*}[tb]
	\centering
	\resizebox{0.85\linewidth}{!}{%
			
	\begin{tikzpicture}[pta, bend angle=15, font=\small]

		\node[methodBox] at (0, 0)  (ds-A) {$\TA$};
		\node[methodBox] at (3.75, 0) (ds-UGA) {$\UG{\TA}$};
		\node[methodBox] at (8.5, 0) (ds-RA-UGA) {$\raUGA$};
		\node[methodBox] at (12.5, 0) (ds-DTN-RA) {$\DRA{\TA}$};
        \node[methodBox] at (14.5,0) (ds-AD) { $\da$};

		\path (ds-A) edge[] node[above]{$-$ location guards} node[below]{$+$ global clock}(ds-UGA);
        \path (ds-UGA) edge[] node[above]{Region automaton} (ds-RA-UGA);
        \path (ds-RA-UGA) edge[] node[above,sloped]{\cref{algo:dtnregionautomaton}} (ds-DTN-RA);
        \path (ds-DTN-RA) edge[sloped] node[above]{} (ds-AD);
        \draw[<->,thick, dashed](ds-A)  --(0,0.7)--(14.5, 0.7)-- (ds-AD);
        \node[] at (7.25, 1)  (langequivnode) {$\proj{\Lg(\TA^\infty)}{1}=\Lg(\da)$};


	\end{tikzpicture}
	}
	\caption{An overview of data structures in the paper}
	\label[figure]{fig:data-structures}
\end{figure*}

Intuitively, \cref{algo:dtnregionautomaton} computes region states reachable by a single process within the context of a network of arbitrary size. 
These region states are partitioned according to their slots.
More precisely, \cref{algo:dtnregionautomaton} computes (lines~\ref{algo:dtnregionautomaton:while-begin}-\ref{algo:dtnregionautomaton:while-end}) the sequence $(\alglayernodes{i},\alglayeredges{i})_{i\geq 0}$, where $\alglayernodes{i}$ is a set of $\MFbound$-region states of $\UG{\TA}$ having the same slot (written 
$\slotof(\alglayernodes{i})$), and $\alglayeredges{i}$ is a set of transitions from region states of $\alglayernodes{i}$ 
to either $\alglayernodes{i}$ or $\alglayernodes{i+1}$.
These transitions include $\epsilon$-transitions which correspond to delay transitions: if the slot does not change during the delay transition, then the transition goes to a region-state which is also in $\alglayernodes{i}$;
otherwise, it leaves to the next slot and the successor is in $\alglayernodes{i+1}$. 
During discrete transitions from~$\alglayernodes{i}$, the slot does not change, so the successor region-states are always inside~$\alglayernodes{i}$.
In order to check if a discrete transition with location guard $\locguard$ must be considered, the algorithm checks if some region-state
$(\locguard,\region_\locguard)$ was previously added to the \emph{same} layer $\alglayernodes{i}$. This means that some (other) process can be at $\locguard$ somewhere at a global time that belongs to $\slotof(\alglayernodes{i})$.
This is the nontrivial part of the algorithm: the proof will establish that if a process can be at location~$\locguard$ at \emph{some} time in a given slot $\slot$, then it can also be at $\locguard$ at \emph{any} time within $\slot$.

For two sets $\setRS_i,\setRS_j$ of region states of $\UG{\TA}$, let us define
$\alglayernodes{i} \approx \alglayernodes{j}$ iff 
$\alglayernodes{j}$ can be obtained from {$\alglayernodes{i}$} by shifting the slot, that is,
if there exists $k \in \Ints$ such that $\newslot{(\alglayernodes{i})}{k} = \alglayernodes{j}$. Recall that 
$\newslot{(\alglayernodes{i})}{k} = \alglayernodes{j}$ means that both sets contain the same
regions when projected to the local clocks $\clocks\setminus\{\gclock\}$.
This definition is of course symmetric.

\cref{algo:dtnregionautomaton} stops (line~\ref{algo:dtnregionautomaton:until}) when %
$\alglayernodes{i_0} \approx \alglayernodes{\indexofsecondrepeatinglayer}$ for some $\indexoffirstrepeatinglayer < \indexofsecondrepeatinglayer$
with both layers having singleton slots (this requirement could be relaxed
but this simplifies the proofs and only increases the number of iterations by a factor of~2).

\begin{algorithm}[t]
    \small
	\Input{\gta $\TA=\TAdefn$ and $\raUGA$}
	\Output{The \dtnregionautomaton{} of $\TA$}
	\BlankLine{} 
              
	Initialize $\slot \assign [0,0]$, 
    $\alglayernodes{0} \assign \{(\locinit,\hat{\region})\}$,
    $\alglayeredges{0} \assign \emptyset$, $\indexl \assign -1$

	\Repeat{$\exists \indexoffirstrepeatinglayer < \indexl : \alglayernodes{\indexl} \approx \alglayernodes{i_0}$, and $\slotof(\alglayernodes{\indexoffirstrepeatinglayer})$ 
    is a singleton%
    \label{algo:dtnregionautomaton:until}}{%
    $\indexl \assign \indexl+1$
        \label{algo:dtnregionautomaton:while-begin}\;

    Compute $(\alglayernodes{\indexl},\alglayeredges{\indexl})$ by applying the following rules until a fixed point is reached:
        \begin{itemize}[leftmargin=5.5mm]
            \item  Rule 1: For any $(\loc,\region) \in \alglayernodes{\indexl}$, let $(\loc,\region) \xrightarrow{\radelaylabel}(\loc,\region')$
            s.t.\ $\slotof(\region')=\slot$, do
            \\
            \qquad $\alglayernodes{\indexl} \assign \alglayernodes{\indexl} \cup \big\{(\loc,\region')\big\}$,
            and
            $\setEdges_{\indexl} \assign \alglayeredges{\indexl} \cup \Big\{\big((\loc,\region),\dtndelay,(\loc,\region')\big)\Big\}$.\label{lbl:rule1} %
            \item Rule 2: For any $(\loc,\region) \in \alglayernodes{\indexl}$ and
            $\gtatran=\gtatrandefn$ s.t. $(\loc,\region)\xrightarrow{\sigma}(\loc',\region')$,\newline
            if $\locguard = \top$, or if there exists $(\locguard, \region_\locguard) \in \alglayernodes{\indexl}$,
            \newline then do
            $\alglayernodes{\indexl} \assign \alglayernodes{\indexl} \cup \big\{(\loc',\region')\big\}$,
            and
            $\alglayeredges{\indexl} \assign \alglayeredges{\indexl} \cup \Big\{\big( (\loc,\region),\sigma,(\loc',\region')\big)\Big\}$ \label{lbl:rule2} 
        \end{itemize}
        \vspace{-1em}
        $\slot \assign \nextslot(\slot)$\;        	
            $\alglayernodes{\indexl+1} \assign \big\{(\loc,\region') \mid (\loc,\region) \in \alglayernodes{\indexl}, (\loc,\region) \xrightarrow{\radelaylabel}(\loc,\region')$ and $ \slotof(\region')= \slot \big\}$\;
            $\alglayeredges{\indexl} \assign \alglayeredges{\indexl} \cup \Big\{\big((\loc,\region),\epsilon,(\loc,\region') \big)\mid (\loc,\region) \in \alglayernodes{\indexl}$, $(\loc,\region) \xrightarrow{\radelaylabel}(\loc,\region') \land \slotof(\region')= \slot \Big\}$ \label{lbl:rule3}\;
        \label{algo:dtnregionautomaton:while-end}
	} %
    $\indexofsecondrepeatinglayer \assign \indexl$\;
            $\alglayeredges{\indexofsecondrepeatinglayer-1} \assign \alglayeredges{\indexofsecondrepeatinglayer-1} \cap \big(\alglayernodes{\indexofsecondrepeatinglayer-1} {\times} \Sigma {\times} \alglayernodes{{\indexofsecondrepeatinglayer-1}}\big)  \cup $ 
                \newline \hspace*{3em}$\Big\{
            \big((\loc,\region),\epsilon,(\loc,\region')\big) \big\vert$
                $(\loc,\region') \in \alglayernodes{\indexoffirstrepeatinglayer}, \exists r'', \exists k \in \Nats,\big((\loc,\region),\epsilon,(\loc,\region'')\big) \in$ 
								\newline \hspace*{12em}$\alglayeredges{\indexofsecondrepeatinglayer-1}  \cap \big(\alglayernodes{\indexofsecondrepeatinglayer-1} {\times} \Sigma {\times} \alglayernodes{\indexofsecondrepeatinglayer}\big),\newslot{\region'}{k} = \region''
        \Big\}$
        \label{line:loopback} 
				
	$\algtotalnodes \assign \cup_{0 \leq i \leq \indexofsecondrepeatinglayer-1 } \alglayernodes{i}, \algtotaledges\leftarrow \cup_{0 \leq i \leq \indexofsecondrepeatinglayer-1 } \alglayeredges{i}$

	\Return{$\DRAdefn{}$\label{algo:dtnregionautomaton:return}}

	\caption{Algorithm to compute \dtnregionautomaton{} of \gta{}~\TA.}
	\label{algo:dtnregionautomaton}
\end{algorithm}

The algorithm returns the DTN region automaton~$\DRA{\TA} = \DRAdefn{}$, where $\alglayernodes{}$ is the set of explored region states, 
and $\alglayeredges{}$ is the set of transitions that were added;
except that transitions leaving {$\alglayernodes{\indexofsecondrepeatinglayer-1}$} are redirected back
to $\alglayernodes{\indexoffirstrepeatinglayer}$ (lines~\ref{line:loopback}--\ref{algo:dtnregionautomaton:return}).
Redirecting such transitions means that whenever $\ra{\MFbound}{\UG{\TA}}$
has a delay transition from
$(\loc,\region)$ to $(\loc,\region')$ with $\slotof(\region) = \slotof(\alglayernodes{\indexofsecondrepeatinglayer-1})$
and $\slotof(\region') = \slotof(\alglayernodes{\indexofsecondrepeatinglayer})$, then we actually add
a transition from
$(\loc,\region)$ to $(\loc,\region'')$, where
$\region''$ is obtained from $\region'$ by shifting the slot to that of $\slotof(\alglayernodes{i_0})$; this means that 
$\eliminate{\region'}{\gclock} = \eliminate{\region''}{\gclock}$, so these define the same clock valuations except with a shifted slot.
The property $\alglayernodes{i_0} \approx \alglayernodes{l_0}$ ensures that $(\loc,\region'') \in \alglayernodes{i_0}$.

We write $\dtntransition{(\loc,\region)}{\sigma}{(\loc',\region')}$ iff $\big ((\loc,\region),\sigma,(\loc',\region') \big)\in \setEdges$.
Paths and computations are defined for the \dtnregionautomaton{}  analogously to region automata.

We now show how to construct the \summaryautomaton{} $\da$ (the step from $\DRA{\TA}$ to $\da$ in \cref{fig:data-structures}).
We define $\da$ by extending $\DRA{\TA}$ with the clocks of~$\TA$.
Moreover, 
each transition $((\loc,\region),\epsilon,(\loc,\region'))$ has the guard~$\eliminate{\region'}{\gclock}$ and no reset;
and each transition 
$((\loc,\region),\sigma,(\loc',\region'))$ with $\sigma \neq \epsilon$
has the guard
$\eliminate{\region}{\gclock}$, and resets the clocks that are equal to~0 in~$\region'$.
The intuition is that $\da$ ensures by construction that any valuation that is to take a discrete transition ($\sigma \neq \epsilon$) at location $(\loc,\region)$ belongs to~$\region$.
Notice that we omit invariants here.
Because transitions are derived from those of~$\raUGA$, the only additional behavior
we can have in $\da$ due to the absence of invariants is a computation delaying in a location $(\loc,\region)$ and reaching outside of $\region$ (without taking an $\epsilon$-transition), while no discrete transitions can be taken afterwards.
Because traces end with a discrete transition, this does not add any trace not possible in~$\proj{\Lg}{1}(\TA^\infty)$.

\subsubsection{Properties of \cref{algo:dtnregionautomaton}}
\label{section:correctness}
We explain the overview of the correctness argument for \cref{algo:dtnregionautomaton} and some of its consequences
(See \cref{appendix:algo-proofs}).

Let us first prove the termination of the algorithm, which also yields a bound on the number of iterations of the main loop (thus on $\indexofsecondrepeatinglayer$ and $\indexoffirstrepeatinglayer$).
Recall that for a given $\TA$, $N_\TA$ denotes the number of pairs $(\loc, \region)$
where $\loc \in \Loc$ and $\region$ is an $\Mbound$-region (see \cref{ss:slots}).

\begin{lemma}%
    \label{lemma:bounded-slots}
    Let~$\DRA{\TA}$ be a \dtnregionautomaton{} returned by \cref{algo:dtnregionautomaton}. Then the slots of all region states in~$\DRA{\TA}$ are bounded. Consequently, the number of iterations
    of \cref{algo:dtnregionautomaton} is bounded by $2^{N_\TA+1}$.
\end{lemma}
 
The region automaton is of exponential size.
Each iteration of \cref{algo:dtnregionautomaton} takes exponential time since one might have to go through all {region states} in the worst case. By \cref{lemma:bounded-slots}, the number of iterations is bounded by $2^{N_\TA}$,
which is doubly exponential in $\lvert \clocks \rvert$.
\cref{thm:reach-pmcp} will show how to decide the reachability PMCP in exponential space.

We now prove the correctness of the algorithm in the following sense.
\begin{theorem}
    \label{thm:dtn}
    Let $\TA$ be a \gta{}, $\DRA{\TA} = \DRAdefn{}$ its \dtnregionautomaton{}, $\da$ be the \summaryautomaton{}.
    Then we have $\proj{\Lg(\TA^\infty)}{1} = \Lg(\da)$.
\end{theorem}

To prove this, we need the following lemma that states a nontrivial property on which we rely:
if a process can reach a given location~$\loc$ at global time~$t'$, then it can also reach~$\loc$ at any global time within the slot of~$t'$.
It follows that the set of global times at which a location can be occupied by at least one process is a union of intervals. Intuitively, this is why partitioning the region states by slots is a good enough abstraction in our setting.

\begin{lemma}\label{lbl:taglobalClockLemma}    
    Consider a \gta~$\TA$ with bound function~$\MFbound$.
    Let $\regionpath=\regionpathdefn$ such that $\regionstatewithindex{0} = \regionstateinit$ be a computation in $\ra{\MFbound}{\UG{\TA}}$
     such that $\slotof(\region_l)$ is bounded. 
    For all $t' \in \slotof(\region_l)$,
    there exists a timed computation {$\tapathdefn$}
    in $\UG{\TA}$
    such that $\clockval_i \in \region_i$ for $0\leq i \leq l$, and $\clockval_l(\gclock) =t'$.
\end{lemma}

The following lemma proves one direction of \cref{thm:dtn}.
\begin{lemma}\label{lbl:dtnsoundness}
    Consider a trace $\trace = (\delta_0,\sigma_0)\ldots (\delta_{l-1},\sigma_{l-1})
    \in \Lg(\da)$.
    Let $I$ be the unique interval of the form $[k,k]$ or $(k,k+1)$ with $k \in \Nats$ that contains $\delta_0+\ldots+\delta_{l-1}$.
    \begin{enumerate}[topsep=1pt]
        \item\label{lbl:dtnsoundness:claim1} For all $t' \in I$,
        there exists $n \in \Nats$, and a computation $\gcomputation$         
        of $\network{n}$ 
        such that $\proj{\traceof{\gcomputation}}{1} = (\delta_0',\sigma_0)\ldots (\delta_{l-1}',\sigma_{l-1})$ for some $\delta_i'\geq 0$, and $\totaltime(\proj{\gcomputation}{1})=t'$.        
        \item\label{lbl:dtnsoundness:claim2} $\trace \in \proj{\Lg(A^\infty)}{1}$.
    \end{enumerate}
\end{lemma}

The following lemma establishes the inclusion in the other direction. Given a computation $\gcomputation$ in $\network{n}$, we build a timed computation
in $\da$ on the same trace. Because the total time delay of $\gcomputation$ can be larger than the bound $2^{N_\TA}$, we need to carefully calculate the slot in which they will end when projected to $\da$.

\begin{lemma}  
    \label{lemma:completeness}
    For any computation $\gcomputation$ of $\network{n}$ with $n \in \Nats$,
    $\proj{\traceof{\gcomputation}}{1} \in \Lg(\da)$.
\end{lemma}

\sj{do we want to say anything about the proofs of these two lemmas?}

\smartpar{Deciding the Reachability PMCP}  
It follows from \cref{algo:dtnregionautomaton} that 
the reachability case can be decided in exponential space.
This basically consists of running the main loop of \cref{algo:dtnregionautomaton} without storing the whole list of all $\alglayernodes{i}$, but only the last one. 
The loop needs to be repeated up to $2^{N_\TA+1}$ times (or until the target label $\sigma_0$ is found).
\begin{theorem}
    \label{thm:reach-pmcp}
    The reachability PMCP for \DTN{}s is decidable in \EXPSPACE{}.
\end{theorem}

\smartpar{Local Properties Involving Several Processes}
The algorithm described above can be extended to compute
$\proj{\Lg(\network{\infty})}{[1,a]}$.
We define the \emph{product} of $k$ timed automata 
$\TA_i$, written $\otimes_{1\leq i \leq k}\TA_i$, as the standard
product of timed automata (see \eg{} \cite{BY03}) applied to $\TA_i$ after replacing each label $\sigma$
appearing in $\TA_i$ by $(i,\sigma)$.

\begin{lemma}
    \label{lemma:da-I}
    Given \gta $\TA$, and interval $I=[1,a]$, let $\da$ be the summary automaton computed as above. Then 
    $\proj{\Lg(\network{\infty})}{I} = \Lg(\otimes_{1\leq i \leq a}\da)$.
\end{lemma}

\begin{wrapfigure}[7]{r}{0.4\textwidth}
    \centering
		\vspace{-1em}
    \begin{tikzpicture}[pta, font=\footnotesize]
        \node[location, initial] at (0, 0)  (q0) {$\locinit$};
        \node[location] at (1.9, 0) (q1) {$\loc_1$};
        \node[location] at (3, 0) (q2) {$\loc_2$};
        \node[invariant, above=of q1] {$\clock \leq 0$};

        \node[above of=q0, locguard,yshift=2em, xshift=0.7em] {$\loc_1$};

        \path (q0) edge[loop above] node[left,align=center]{$\clock=1$\\$\clock \leftarrow 0$} (q0);
        \path (q0) edge[] node[above]{$\clock \leftarrow 0$} (q1);
        \path (q1) edge[] (q2);
    \end{tikzpicture}
    \caption{An example of \gta for which liveness is not preserved by our abstraction.}
    \label{fig:counterexample-liveness}
    \vspace{-0.5em}
\end{wrapfigure}
\smartpar{Limitations} Liveness properties (\eg{} checking whether a transition can be taken an infinite number of times) are not preserved by our abstraction; since an infinite loop in the DTN region automaton may not correspond to a concrete computation in any~$\network{n}$.
In fact, consider the \gta{} in \cref{fig:counterexample-liveness}.
While there is an infinite loop on~$\locinit$ in the DTN region automaton, no concrete execution takes the loop on~$\locinit$ indefinitely, as each firing of this loop needs one more process to visit~$\loc_1$, and then to leave it forever, due to the invariant $\clock \leq 0$.

\section{Timed Lossy Broadcast Networks}
\label{section:ltba}
Systems with lossy broadcast (a.k.a.\ ``reconfigurable broadcast networks'', where the underlying network topology might change arbitrarily at runtime)
have received attention in the parameterized verification literature~\cite{DelzannoSZ10}.
In the setting with finite-state processes, lossy broadcast is known to be at least as powerful as disjunctive guards, but it is unknown if it is strictly more powerful~\cite[Section~6]{DBLP:journals/corr/abs-2109-08315}.
We show that in our \emph{timed} setting the two models are equally powerful, \ie{} they simulate each other.
Details are provided in \cref{lbl:fig-lossy-dtn,lbl:fig-dtn-lossy}, and the corresponding proofs can be found in \cref{sec:proof-ltba}.

\smartpar{Lossy Broadcast Networks}
Let $\Sigma$ be a set of labels.
A lossy broadcast timed automaton (LBTA) $B$ is a tuple $(\Loc,\locinit,\clocks,\Sigma,\Synclabels,\Transitions,\invariant)$
where $\Loc,\locinit,\clocks,\invariant$ are as for TAs, and a transition is of the form $(\loc, \guard, \resetclocks,\tatranlabel,\synclabel,\loc')$, where
$\synclabel \in \{ a!!, a?? \mid a \in \Synclabels\}$.
The \emph{synchronization label} $\synclabel$ is used for defining global transitions.
A transition with $\synclabel = a!!$ is called a \emph{sending transition}, and a transition with $\synclabel = a??$ is called a \emph{receiving transition}.

We also make \cref{assumption:nonblocking-2} and \cref{assumption:unique-sigma} for~LTBAs. 
\add{The former means that the LBTA is timelock-free when all receiving transitions are removed.}
The semantics of a network of LBTAs is a timed transition system defined similarly as for \ngtas, except
for \emph{discrete transitions} which now induce a sequence of local transitions separated by 0 delays, as follows.
Given $n>0$, configurations of $B^n$ are defined as for DTNs. %
Let $\nConfig=\dtnconfig$ be a configuration of $B^{n}$.
Consider indices $1\leq i \leq n$ and $J \subseteq \{1,\ldots, n\}\setminus\{i\}$,
and labels $\sigma, \sigma_j \in \Sigma$ for $j\in J$, such that
\begin{ienumerate}%
    \item $\transition{(\loc_i,\clockval_i)} {a!!,\tatranlabel} {(\loc_i',\clockval_i')}$ is a sending transition of $B$, and
    \item for all $j \in J$: $\transition{(\loc_j,\clockval_j)} {a??, \tatranlabel_j} {(\loc_j',\clockval_j')}$ is a receiving transition of $B$.
\end{ienumerate}
Then the timed transition system of $B^n$ contains the transition sequence
$\nConfig \xrightarrow{(0,(i, \sigma))} \nConfig' \xrightarrow{(0, (j_1, \sigma_{j_1}))} \nConfig_1' \xrightarrow{(0, (j_2, \sigma_{j_2}))}\ldots \xrightarrow{(0, (j_m, \sigma_{j_m}))}\nConfig_m'$
for all possible sequences $j_1,\ldots,j_m$ where $J = \{j_1,\ldots,j_m\}$.
Non-zero delays only occur outside of these chains of 0-delay transitions.
For a given LBTA $B$, the family of systems $B^\infty$ is called a \emph{lossy broadcast timed network} (LBTN).

\smartpar{Simulating LBTN by DTN (and vice versa)}
The following theorem states that LBTAs and \gta{}s are inter-reducible.
\begin{theorem}
    \label{thm:ltba-equivalent}
    For all \gta $\TA$, there exists an LBTA $B$ s.t.\ $\forall k {\geq} 1$, $\proj{\calL(\network{\infty})}{[1,k]} \equiv \proj{\calL(B^\infty)}{[1,k]}$.
    For all LBTA $B$, there exists a \gta $\TA$ s.t.\
    $\forall k {\geq} 1, \proj{\calL(\network{\infty})}{[1,k]} = \proj{\calL(B^\infty)}{[1,k]}$.
\end{theorem}

\begin{proof}[Sketch]
Simulation of disjunctive guards by lossy broadcast is simple: a transition from $\loc$ to $\loc'$ with location guard $\locguard$ is simulated in lossy broadcast by the sender taking a self-loop transition on $\locguard$, and the receiver having a synchronizing  transition from $\loc$ to $\loc'$.

The other direction is where we need the power of clock invariants: to simulate a lossy broadcast where the sender moves from $\loc$ to $\loc'$ and a receiver moves from $\loc_\rcv$ to $\loc_\rcv'$, in the DTN we first let the sender move to an auxiliary location $\loc_\sigma$ (from which it can later move on to $\loc'$), and have a transition from $\loc_\rcv$ to $\loc_\rcv'$ that is guarded with $\loc_\sigma$. To ensure that no time passes between the steps of sender and receiver, we add an auxiliary clock $\clock_\send$ that is reset when moving into $\loc_\sigma$, and $\loc_\sigma$ has clock invariant $\clock_\send = 0$. 

In both directions, auxiliary transitions that are only needed for the simulation are labeled with fresh symbols in $\Sigma^-$ such that they do not appear in the language of the system.
\end{proof}

Because the reduction to \DTN{}s is in linear-time, we get the following.

\begin{corollary}
    The reachability PMCP for LBTN is decidable in \EXPSPACE{}.
\end{corollary}

\section{Synchronizing Timed Networks and Timed Petri Nets}
\label{sec:stn}

We first introduce synchronizing timed networks.
Our definitions follow \cite{AJ03,DBLP:conf/lics/AbdullaDM04}, except that their model considers systems with a controller process, whereas we assume (like in our previous models) that all processes execute the same automaton.

\smartpar{Synchronizing Timed Network}
A \emph{synchronizing timed automaton} (STA) $\STA$ is a tuple $\STAdefn$ where
$\Loc$, $\locinit$, $\clocks$, $\invariant$ are as for TAs, and
$\STARules$ is a finite set of \emph{rules}, where each rule $\STArule \in \STARules$ is of the form \sj{make it a set instead of a sequence?}
        $
        \STAruledefn 
        $
        for some $m \in \Nats$ and with $(\loc_i,\guard_{\STArule,i},\STAresetclocks{\STArule}{i},\sigma_{\STArule,i},\loc'_i) \in \Loc \times \clockcons \times 2^\clocks \times  \Sigma \times \Loc$ for $1 \leq i \leq m$.

The semantics of a \emph{network of STAs} (NSTA) is defined as for~\ngtas, except for \emph{discrete transitions}, which now synchronize a subset of all processes in the following way:
Let $\STArule \in \STARules$ be a rule (of the form described above) and $\nConfig=\dtnconfig$ a configuration of $\STA^n$. Assume
\begin{ienumerate}
\item there exists an injection $h : \{1, \dots, \STArulelength\} \to \{1 \ldots \networksize \}$  such that for each $1 \leq i \leq \STArulelength$,
    \add{$\loc_{h(i)} = \loc_{\STArule,i}$,}
    $\loc_{h(i)}  \xrightarrow[]{\guard_{\STArule,i},\STAresetclocks{\STArule}{i},\tatranlabel_{\STArule,i}}  \loc'_{h(i)}$ is an element of $\STArule$, $\clockval_{h(i)}\models \guard_{\STArule,i}$ and $\clockval'_{h(i)}=\clockval_{h(i)}[\STAresetclocks{\STArule}{i} \leftarrow 0]$, and 
\item $j \not \in \text{range}(h)$, $\loc'_j = \loc_j$ and $\clockval'_j = \clockval_j$.
\end{ienumerate}
Then the timed transition system of $\STA^n$ contains the transition sequence
$\nConfig \xrightarrow{(0,(h(1), \sigma_{\STArule,1}))} \nConfig_1 \xrightarrow{(0, (h(2), \sigma_{\STArule,2}))}\ldots \xrightarrow{(0, (h(m), \sigma_{\STArule,m}))}\nConfig_m$.
That is, $m$~distinct processes take individual transitions according to the rule without delay, and the configurations of the non-participating processes remain unchanged.

Again, we also make \cref{assumption:nonblocking-2} and \cref{assumption:unique-sigma} for~STAs.
\add{The former means here that $\STA$ is timelock-free when all transitions of rules with $\STArulelength>1$ are removed.}
All other notions follow in the natural way.
\add{Given an STA~$\STA$, the family of systems $\STA^\infty$ is called a \emph{synchronizing timed network}~(STN).}

\begin{theorem}
    \label{thm:sta-equivalent}
    \add{For all \gta $\TA$ with set of locations $\Loc$, there exists an STA $\STA$ with set of locations $\Loc$ such that for every $\loc \in \Loc$: $\loc$ is reachable in $\TA$ iff $\loc$ is reachable in $\STA$.
		For all STA $\STA$ with set of locations $\Loc_\STA$, there exists a \gta $\TA$ with set of locations $\Loc_\TA \supseteq \Loc_\STA$ such that for every $\loc \in \Loc_\STA$: $\loc$ is reachable in $\TA$ iff $\loc$ is reachable in $\STA$.}
\end{theorem}

\begin{proof}[Sketch]
Simulation of disjunctive guards by STAs is simple: a transition from $\loc$ to~$\loc'$ with location guard~$\locguard$ is simulated by a pairwise synchronization, where one process takes a self-loop on~$\locguard$, and the other moves from $\loc$ to~$\loc'$.

Conversely, to simulate a rule $\STArule$ of the STA with $\STArulelength$ participating processes, we add auxiliary locations $\locmid_{\STArule,i}$, for $1 \leq i \leq \STArulelength$, each with a clock invariant (on an additional clock only used for the simulation) that ensures that no time passes during simulation. 
For each element $\loc_{\STArule,i} \xrightarrow[]{\guard_{\STArule,i},\STAresetclocks{\STArule}{i},\tatranlabel_{\STArule,i}} \loc'_{\STArule,i}$ of $\STArule$, we have a transition from $\loc_{\STArule,i}$ to $\locmid_{\STArule,i}$, and from there to $\loc'_{\STArule,i}$. 
A transition to $\locmid_{\STArule,i}$ is guarded with $\locmid_{\STArule,i-1}$ \add{(except when $i=1$)}, and with the clock constraint $\guard_{\STArule,i}$, and all transitions to $\loc_{\STArule,i}'$ are guarded with $\locmid_{\STArule,\STArulelength}$. 
This ensures that any $\loc_{\STArule,i}'$ is reachable through this construction if and only if the global configuration at the beginning would allow the STA to execute rule $\STArule$.
\add{To avoid introducing timelocks, each of the $\locmid_{\STArule,i}$ has an additional transition with a trivial location guard and no clock guard to a new sink location $\locsink$ that does not have an invariant.
I.e., if simulation of a rule is started but cannot be completed (because there are processes in some but not all of the locations $\loc_{\STArule,i}$), then processes can (and have to) move to $\locsink$.}
\LongVersion{%

}
Details are provided in \cref{lbl:fig:TN-DTN}, and a full proof can be found in \Cref{sec:simulation-STA-proof}.

\end{proof}

\begin{corollary}
    The reachability PMCP for STN is decidable in \EXPSPACE{}.
\end{corollary}

Note that the construction in our proof is in general not suitable for language equivalence, \ie{} $\proj{\calL(\network{\infty})}{[1,k]}$ might contain traces that are not in $\proj{\calL(\STA^\infty)}{[1,k]}$.

\textreplace{STNs without location invariants are equivalent to a model of timed Petri nets considered by Abdulla et~al.~\cite{AbdullaACMT18}.
In their work, they left as an open problem whether the universal safety problem, \ie{} whether a given rule/transition can eventually be enabled for any number of processes, is decidable.
This problem can be reduced to the reachability PMCP of~STNs.}
{
Abdulla et al.~\cite{AbdullaACMT18} considered the \emph{universal safety problem} of \emph{timed Petri nets} — that is, whether a given transition can eventually be fired for any number of tokens in the initial place — and solved it for the case where each token has a single clock.
The question whether the problem is decidable for tokens with multiple clocks remained open.
This problem, in the multi-clock setting,
can be reduced to the PCMP of STNs. The reduction is conceptually straightforward and computable in
polynomial time in the size of the input.

}

\begin{corollary}
    The universal safety problem for timed Petri nets with an arbitrary number of clocks is decidable in \EXPSPACE{}.
\end{corollary}

\section{Conclusion}\label{section:conclusion}

In this paper, we solved positively the parameterized model checking problem (PMCP) for finite local trace properties of disjunctive timed networks (\DTN{}s) with invariants.
We also proved that the PMCP for networks that communicate via lossy broadcast can be reduced to the PMCP for \DTN{}s, and is therefore decidable.
Additional results also allowed us to solve positively the open problem from~\cite{AbdullaACMT18} whether the universal safety problem for timed Petri nets with multiple clocks is decidable.
\cref{tab:reachability_trace} gives an overview of our results, compared to existing results for the classes of systems we consider.
\hspace*{-1.1cm} %

\begin{table}[t]
    \caption{Existing and \colorbox{green!35}{new} decidability results for location reachability (Reach) and local trace properties (Trace) for DTN, LBTN, and STN with a single ($\card{\clocks}=1$) or multiple clocks ($\card{\clocks}{>}1$), and with (\textcolor{orange}{\textit{Inv}}) or without invariants (\textcolor{orange}{\textit{\cancel{Inv}}}). Entries with \checkmark$^*$ need to satisfy \cref{assumption:nonblocking-2}.}
    \label{tab:reachability_trace}
    \centering
    \renewcommand{\arraystretch}{1.3} %
    \setlength{\tabcolsep}{1.5pt} %
    \small %
    \begin{tabular}{!{\vrule width 0.8pt} >{\centering\arraybackslash}m{1.1cm} !{\vrule width 0.8pt} 
        >{\centering\arraybackslash}m{1.1cm} | >{\centering\arraybackslash}m{1.1cm} | >{\centering\arraybackslash}m{1.1cm} !{\vrule width 0.8pt} 
        >{\centering\arraybackslash}m{1.1cm} | >{\centering\arraybackslash}m{1.1cm} | >{\centering\arraybackslash}m{1.1cm} !{\vrule width 0.8pt} 
        >{\centering\arraybackslash}m{1.1cm} | >{\centering\arraybackslash}m{1.1cm} | >{\centering\arraybackslash}m{1.1cm} !{\vrule width 0.8pt}}
        
        \Xhline{1pt} 
        \multicolumn{1}{!{\vrule width 0.8pt}c!{\vrule width 0.8pt}}{\cellcolor{gray!15} } & \multicolumn{3}{c!{\vrule width 0.8pt}}{\cellcolor{gray!15} \textbf{DTN}} & \multicolumn{3}{c!{\vrule width 0.8pt}}{\cellcolor{gray!15} \textbf{LBTN}} & \multicolumn{3}{c!{\vrule width 0.8pt}}{\cellcolor{gray!15} \textbf{STN}} \\ 
        \Xhline{1pt} 
        \cellcolor{gray!15} & 
        \cellcolor{gray!8}  $\card{\clocks}{=}1$ & \cellcolor{gray!8}  $\card{\clocks} {>}1$  & \cellcolor{gray!8}  $\card{\clocks} {>}1$ & \cellcolor{gray!8}  $\card{\clocks} =1$ & \cellcolor{gray!8}  $\card{\clocks} {>}1$  & \cellcolor{gray!8}  $\card{\clocks} {>}1$ & \cellcolor{gray!8}  $\card{\clocks} =1$ & \cellcolor{gray!8}  $\card{\clocks} {>}1$  & \cellcolor{gray!8}  $\card{\clocks} {>}1$ \\ 
        \hline
        \cellcolor{gray!15} & \cellcolor{gray!8} \textcolor{orange}{\textit{\cancel{Inv}}} & \cellcolor{gray!8} \textcolor{orange}{\textit{\cancel{Inv}}} & \cellcolor{gray!8} \textcolor{orange}{\textit{Inv}} & \cellcolor{gray!8} \textcolor{orange}{\textit{\cancel{Inv}}} & \cellcolor{gray!8} \textcolor{orange}{\textit{\cancel{Inv}}} & \cellcolor{gray!8} \textcolor{orange}{\textit{Inv}} & \cellcolor{gray!8} \textcolor{orange}{\textit{\cancel{Inv}}} & \cellcolor{gray!8} \textcolor{orange}{\textit{\cancel{Inv}}} & \cellcolor{gray!8} \textcolor{orange}{\textit{Inv}}  \\
        
        \Xhline{0.8pt} 
        \cellcolor{gray!15} \textbf{Reach} & \checkmark \cite{SpalazziS20} & \checkmark \cite{SpalazziS20} & \cellcolor{green!35} \checkmark  & \checkmark \cite{AJ03,ADFL19} & \cellcolor{green!35} \checkmark & \cellcolor{green!35} \checkmark & \checkmark \cite{AJ03} & \cellcolor{green!35} \checkmark  & \cellcolor{green!35} \checkmark \\ 
       
        \hline
        \cellcolor{gray!15} \textbf{Trace} & \checkmark \cite{SpalazziS20} & \checkmark \cite{SpalazziS20} & \cellcolor{green!35} \checkmark$^*$  & \cellcolor{green!35} \checkmark & \cellcolor{green!35} \checkmark & \cellcolor{green!35} \checkmark$^*$  & \cellcolor{gray!35} ? & \cellcolor{gray!35} ? & \cellcolor{gray!35} ?\\
        \Xhline{1pt} 
    \end{tabular}
\end{table}

In addition to the %
results presented here, we believe that our proof techniques can be extended to support timed networks with more powerful communication primitives, and in some cases to networks with controllers.

Future work will include tightening the complexity bounds for the problems considered here, as well as the development of zone-based algorithms that can be more efficient in practice than a direct implementation of the algorithms presented here.

\newpage

	\newcommand{\CCIS}{Communications in Computer and Information Science}
	\newcommand{\ENTCS}{Electronic Notes in Theoretical Computer Science}
	\newcommand{\FAC}{Formal Aspects of Computing}
	\newcommand{\FundInf}{Fundamenta Informaticae}
	\newcommand{\FMSD}{Formal Methods in System Design}
	\newcommand{\IJFCS}{International Journal of Foundations of Computer Science}
	\newcommand{\IJSSE}{International Journal of Secure Software Engineering}
	\newcommand{\IPL}{Information Processing Letters}
	\newcommand{\JAIR}{Journal of Artificial Intelligence Research}
	\newcommand{\JLAP}{Journal of Logic and Algebraic Programming}
	\newcommand{\JLAMP}{Journal of Logical and Algebraic Methods in Programming} %
	\newcommand{\JLC}{Journal of Logic and Computation}
	\newcommand{\LMCS}{Logical Methods in Computer Science}
	\newcommand{\LNCS}{Lecture Notes in Computer Science}
	\newcommand{\RESS}{Reliability Engineering \& System Safety}
	\newcommand{\RTS}{Real-Time Systems}
	\newcommand{\SCP}{Science of Computer Programming}
	\newcommand{\SOSYM}{Software and Systems Modeling ({SoSyM})}
	\newcommand{\STTT}{International Journal on Software Tools for Technology Transfer}
	\newcommand{\TCS}{Theoretical Computer Science}
	\newcommand{\TOPLAS}{{ACM} Transactions on Programming Languages and Systems ({ToPLAS})}
	\newcommand{\ToPNoC}{Transactions on {P}etri Nets and Other Models of Concurrency}
	\newcommand{\TOSEM}{{ACM} Transactions on Software Engineering and Methodology ({ToSEM})}
	\newcommand{\TSE}{{IEEE} Transactions on Software Engineering}

	\bibliographystyle{splncs04}
	\bibliography{main}

	\appendix
\newpage
\section{Omitted Formal Definitions}\label{appendix:definitions}

\subsection{Timed automata}\label{appendix:trace}

We give the formal definition of the trace of a timed path $\lcomputation = \lcomputationdefn$,
which is the sequence
of pairs of delays and labels, obtained by removing \textreplace{$\epsilon$-transitions}{transitions with a label from $\Sigma^-$} and
adding the delays of these to the following transition.
Formally, if all $\tatranlabel_j$ are \textreplace{$\epsilon$}{from $\Sigma^-$}, then
the trace is empty. Otherwise,
$\traceof{\lcomputation} = (\delta_0', \tatranlabel_0')\ldots (\delta_{m}', \tatranlabel_m')$ defined as follows.
Let $0\leq i_0<\ldots<i_m\leq l-1$ be the maximal sequence such that \textreplace{$\tatranlabel_{i_j}\neq \epsilon$}{$\tatranlabel_{i_j}\notin \Sigma^-$} for each~$j$.
Then, $\tatranlabel_j' = \tatranlabel_{i_j}$.
Moreover,  $\delta_j' = \sum_{i_{j-1} < k \leq i_j} \delta_k$ with $\delta_{-1} = -1$.

\smartpar{Zones and DBMs~\cite{HNSY-1994,BY03}}
\label{subsection:zones}
A \emph{zone} is a set of clock valuations that are defined by a clock constraint $\ccons$ (as defined in \cref{sec:model}).
In the following we will use the constraint notation and the set (of clock valuations) notation interchangeably. %
We denote by $\setofallzones$ the set of all zones.

Let $\posttime(\zone) = \{\clockval' \mid \exists \clockval \in \zone, \exists \delta \geq 0, \clockval'=\clockval+\delta\}$ denote the \emph{time successors} of~$\zone$, and for a transition~$\tatran=\tatrandefn$, let $\post_\tatran(\zone) =\{\clockval' \mid \exists \clockval \in \zone, \clockval\models \invariant(\loc) \land \guard, \clockval' =\clockval[\resetclocks \leftarrow 0], \clockval' \models  \invariant(\loc')\}$
be the \emph{immediate successors of $\zone$ via~$\tatran$}.

For a set~$\clocks$ of clocks, we denote by $\clocks_0 = \clocks \cup \{\clock_0\}$ the set $\clocks$ extended with a special variable~$\clock_0$ with the constant value~$0$.
For convenience we sometimes write~$0$ to represent the variable~$\clock_0$.

A \emph{difference bound matrix (DBM)} for a set of clocks $\clocks$ is a $|\clocks_0| \times |\clocks_0|$-matrix $\zonedbmdefn$, in which each entry $\zonedbmentry= \zonedbmentrydefn$ represents the constraint $\clock-\clock' \lessdot_{\clock \clock'} \intconstant_{\clock \clock'}$ where $\intconstant_{\clock \clock'} \in \Ints$ and $\lessdot_{\clock \clock'} \in \{<, \leq\}$, or $(\lessdot_{\clock \clock'},\intconstant_{\clock \clock'}) = (<, \infty)$.

It is known that a zone $\zone$  can be \emph{represented by a DBM} \ie{} a valuation $\clockval \in \zone$ iff $\clockval $ satisfies all the clock constraints represented by the~DBM.

We define a {total} ordering on  $(\lessdot_{\clock \clock'} \times \add{\Ints}) \cup \{(<,\infty)\}$ as follows $(\lessdot_{\clock \clock'}, \intconstant)< (<, \infty) $, 
$(\lessdot_{\clock \clock'}, \intconstant) < (\lessdot_{\clock \clock'}', \intconstant')$ if $\intconstant< \intconstant'$ and $(<,\intconstant) < (\leq,\intconstant)$.
A DBM is \emph{canonical} if none of its constraints can be strengthened (\ie{} replacing one or more entries with a strictly smaller entry based on the ordering we defined) without reducing the set of solutions.
Given a DBM $\zonedbmdefn$, we denote by $\valuationsofdbm{\zonedbmdefn}$, the set of valuations that satisfy all the clock constraints in~$\zonedbmdefn$, which is a zone. Any region and zone can be represented by a DBM.
Given two DBMs $\zonedbmdefn,\zonedbmdefnprime$
we write $\zonedbmdefn \leq \zonedbmdefnprime$ if for all $\clock,\clock'$,    
    $\zonedbmentry \leq \zonedbmentryprime$,
if~$\zone=\valuationsofdbm{\zonedbmdefn}$ and $\zone'=\valuationsofdbm{\zonedbmdefnprime}$,
then we have $\zonedbmdefn \leq \zonedbmdefnprime $ iff $\zone \subseteq \zone'$.

\subsection{Networks of TAs}\label{appendix:NTAs}
We formally define \emph{projections} of computations and traces of
\ngtas.
If $\nConfig=\big((\loc_1,\clockval_1),\ldots,(\loc_n,\clockval_n)\big)$ and $\intervalPr = \intervalPrdefn \subseteq \{1, \dots, n\}$, then $\proj{\nConfig}{\intervalPr}$ is the tuple $\big((\loc_{i_1},\clockval_{i_1}),\ldots,(\loc_{i_k},\clockval_{i_k})\big)$, and we extend this notation to computations 
$\proj{\gcomputation}{\intervalPr}$ by keeping only the discrete transitions of $ \intervalPr$ and by adding the delays of the removed discrete transitions to the delay of the following discrete transition of $\intervalPr$. Formally,
given $\intervalPr \subseteq\{1,\ldots,n\}$ and computation 
$\gcomputation =  \gcomputationdefn$,
let $\proj{\gcomputation}{\intervalPr}$ denote the projection of $\gcomputation$ to processes $\intervalPr$,
defined as follows. Let $0 \leq k_0< \ldots <k_m\leq l-1$ be 
a sequence of maximal size such that $i_{k_j} \in \intervalPr$ for all $0\leq j \leq m$.
Then $\proj{\gcomputation}{\intervalPr}
= \proj{\nConfig_0}{\intervalPr} \xrightarrow{(\delta_{k_0}',(i_{k_0},\tatranlabel_{k_0}))}
\proj{\nConfig_{k_0+1}}{\intervalPr}
\ldots \proj{\nConfig_{k_m}}{\intervalPr} \xrightarrow{(\delta_{k_m}', (i_{k_m},\tatranlabel_{k_m}))} \proj{\nConfig_{k_m+1}}{\intervalPr}
$, where each $\delta_{k_j}'=\sum_{k_{j-1}< j'\leq k_j}\delta_{j'}$
with $k_{-1} = -1$.

We formally define the \emph{composition} of computations of \ngtas.
For timed paths $\gcomputation_1$ of $\network{n_1}$ and $\gcomputation_2$ of $\network{n_2}$ with $\totaltime(\gcomputation_1) = \totaltime(\gcomputation_2)$,
we denote by $\gcomputation_1 \parallel \gcomputation_2$ their \emph{composition} into a timed path of $\TA^{n_1+n_2}$ 
whose projection to the first $n_1$ processes is~$\gcomputation_1$, and whose projection to the last $n_2$ processes is~$\gcomputation_2$.
If $\gcomputation_1$ has length~0, then we concatenate the first configuration
of $\gcomputation_1$ to $\gcomputation_2$ by extending it to $n_1+n_2$ dimensions;
and symmetrically if $\gcomputation_2$ has length~$0$.
Otherwise, for $i \in \{1,2\}$, 
let~$\gcomputation_{i} = \nConfig_0^i \xrightarrow{\delta^i_0,\sigma^i_0}
\nConfig_1^i\xrightarrow{\delta^i_1,\sigma^i_0} \ldots$.
Let~$i_0\in\{1,2\}$ be such that $\delta_0^{i_0} \leq \delta_{0}^{3-i_0}$
(and pick $i_0=1$ in case of equality).
Then the first transition of $\gcomputation_1 \parallel \gcomputation_2$
is $\nConfig_0 \xrightarrow{\delta_{0}^{i_0}, \sigma_0^{i_0}} \nConfig_1$, 
where $\nConfig_0$ is obtained by concatenating $\nConfig_0^{i_0}$ and $\nConfig_0^{3-i_0}$,
and $\nConfig_1$ is obtained by $\nConfig_1^{i_0}$ and $\nConfig_0^{3-i_0}$.
We define the rest of the path recursively, after subtracting $\delta_0^{i_0}$ from
the first delay of $\gcomputation_{3-i_0}$.

\section{Omitted Proofs}
\label{app:proofs}

\subsection{Proofs of \cref{ss:slots}}

\begin{proof}[Proof of \cref{lemma:slots}]
    Consider $\region$ represented by a DBM in canonical form.
    We define $\region'_{\gclock,\clock} = \region_{\gclock,\clock} +k$
    and $\region'_{\clock,\gclock} = \region_{\clock,\gclock}  -k$,
    and $\region'_{\clock,\clock'} = \region_{\clock,\clock'}$ for all $\clock,\clock' \in \clocks\setminus\{\gclock\}$.
    In addition, if $\region'_{\gclock,\clock} > \MFbound{\gclock}$, then it is replaced by $\infty$,
    and if  $\region'_{\clock,\gclock} < -\MFbound{\clock}$, it is replaced by $-\infty$.

    Note that the relations between $\clock$ and $\clock'$ remain unchanged for all $\clock,\clock' \in \clocks\setminus\{\gclock\}$ since
    these entries are not modified, and moreover, the canonical form also cannot assign these entries a smaller value.
    This is because $\region_{\clock,\clock'} \leq \region_{\clock,\gclock} + \region_{\gclock,\clock'}$
    by the canonical form of $\region$.
    Moreover, 
    \[
        \region'_{\clock,\gclock} + \region'_{\gclock,\clock'} = \region_{\clock,\gclock} + k + \region_{\gclock,\clock'} - k = 
        \region_{\clock,\gclock} + \region_{\gclock,\clock'},
    \]
    So we also have $\region'_{\clock,\clock'}\leq \region'_{\clock,\gclock} + \region'_{\gclock,\clock'}$.
    So is $\region'$ is also in canonical form, and we have $\eliminate{\region'}{\gclock} = \eliminate{\region}{\gclock}$.
\end{proof}

\subsection{Proofs of \cref{section:correctness}}\label{appendix:omitted-proofs}
\label{appendix:algo-proofs}
\begin{proof}[Proof of \cref{lemma:bounded-slots}]
    Note that each $\alglayernodes{i}$ is a set of {$\MFbound$-\regionstates{}} of $\UG{A}$ (thus, including the global clock). Moreover, each  {$\proj{\alglayernodes{i}}{-\gclock}$} is
    a set of pairs $(\loc,\region)$ where $\loc$ is a location,
    and $\region$ is a $\Mbound$-regions.
    Recall that $N_{\TA}$ is the number of such pairs.
    Every layer $\alglayernodes{i}$ with even index~$i$ and with a bounded slot has a singleton slot.
    So whenever $\indexofsecondrepeatinglayer\geq 2^{N_\TA+1}$, there exists $0\leq i<j \leq \indexofsecondrepeatinglayer$ such that
    $\alglayernodes{i} \approx \alglayernodes{j}$.

    Because $\MFbound(\gclock) = 2^{N_\TA+1}$, if the algorithm stops at some iteration $l$,
    then $\slotof(\alglayernodes{l})$ is indeed bounded.
\end{proof}

\begin{proof}[Proof of \cref{lbl:taglobalClockLemma}]
    Consider a computation  $\regionpath=\regionpathdefn$ of $\ra{\MFbound}{\UG{\TA}}$, and assume, w.l.o.g. that it starts with 
    a delay transition, and delays and discrete transitions alternate.
    By the properties of regions \cite{AD94}, %
    there exists a timed computation $\tapath = \tapathdefn$ that follows $\regionpath$, in the sense that $\clockval_i \in \region_i$ for all $0\leq i \leq l$.
    
    Recall that, by definition, each discrete transition with label $\tatranlabel$ in $\ra{\MFbound}{\UG{\TA}}$ is built from a transition of
    $\UG{\TA}$ with label $\tatranlabel$.

    Let $\tau_i=\epsilon$ if~$\sigma_i=\radelaylabel$, 
    and otherwise, let $\tau_i$ denote the transition $(\loc_i, \guard, \resetclocks,\tatranlabel_i,\loc_{i+1})$
    of~$\UG{\TA}$ corresponding to the $i$-th transition of $\regionpath$.

    For a set of clock valuations~$z$, we define $\posttime(z)=\{\clockval+d \mid d\geq 0, \clockval \in z\}$,
    and for a given transition $\tatran=\tatrandefn$, $\post_{\tatran}(z) = \{ \clockval[\resetclocks \leftarrow 0] \mid \clockval \in z, z \models \guard \}$.
    Consider the sequence $(z)_i$ defined by $z_0 = \regioninit$, and for all $0 \leq i \leq l-1$, 
    \begin{center}
        \begin{itemize}                
            \item $z_{i+1}=\posttime(z_i) \cap \region_{i+1}$, if $\tatran_i=\epsilon$ (delay transition);
            \item $z_{i+1}=\post_{\tatran_i}(z_i) \cap \region_i$, if $\tatran_i \neq \epsilon$ (discrete transition).
        \end{itemize}
    \end{center}
    Each $z_i$ is exactly the set of those valuations that are reachable from the initial valuation        
    by visiting $\region_i$ at step~$i$.
    In other terms, $z_l$ is the \emph{strongest post-condition} of $\regioninit$ via $\regionpath$: it is the set of states reachable from the initial state
    in $\UG{\TA}$, by following the discrete transitions of~$\regionpath$, while staying at each step inside~$r_i$.
    We have that $z_l \neq \emptyset$ since there exists 
    the computation $\tapath$ mentioned above.
    We can have however $z_i \subsetneq r_i$ since unbounded regions can contain unreachable valuations. 

    In this proof, we assume familiarity with DBMs; formal definitions are given in \cref{subsection:zones}.
    The sequence $(z_i)$ can be computed using zones, using difference bound matrices (DBM)~\cite{HNSY-1994,BY03}.
    Let $\zonedbmforproof,\regiondbmforproof$ be DBMs in canonical 
    form representing  $\lastzone, \lastregion$ respectively. 

    Since, by construction, $\lastzone \subseteq \lastregion$, it holds that for every pair of clocks $\clock,\clock'$, $\zonedbmentryforproof{c}{c'} \leq \regiondbmentryforproof{c}{c'}$.
    In particular $\zonedbmentryforproof{\gclock}{0} \leq \regiondbmentryforproof{\gclock}{0}$ and $\zonedbmentryforproof{0}{\gclock} \leq \regiondbmentryforproof{0}{\gclock}$ 
     that is   $ [-\zonedbmentryforproof{0}{\gclock},\zonedbmentryforproof{\gclock}{0}] \subseteq [-\regiondbmentryforproof{0}{\gclock},\regiondbmentryforproof{\gclock}{0}]$. And we have
     $\emptyset \neq [-\zonedbmentryforproof{0}{\gclock},\zonedbmentryforproof{\gclock}{0}]$ since $z_l\neq \emptyset$.
    But in our case the interval on the right hand side is a slot, and because it is bounded it is either a singleton or an interval of the form $(k,k+1)$. It follows that $[-\zonedbmentryforproof{0}{\gclock},\zonedbmentryforproof{\gclock}{0}]$ cannot be strictly smaller since it is nonempty.

    Now, because $\zonedbmforproof$ is in canonical form,  for all $t' \in \mQz$ such that {$t'$} $\leq \zonedbmentryforproof{\gclock}{0}$ and {$-t'$}$ \leq \zonedbmentryforproof{0}{\gclock}$, there exists a valuation $\clockval_l \in \valuationsofdbm{\zonedbmforproof}$ and $\clockval_l(\gclock)=t'$
    \cite[Lemma~1]{HerbreteauKSW11}.
    Because $\clockval_l \in z_l$, and by definition of $z_l$,
    there exists a timed computation {$\tapathdefn$} such that
    $(\loc_0,\clockval_0)=(\locinit,\mathbf{0})$
    and $\clockval_i \in \region_i$ for $i \in \{ 1 \ldots l\}$.
\end{proof}

\begin{proof}[Proof of \cref{lbl:dtnsoundness}]

    Consider the order in which the transitions $((\loc,\region),\sigma, (\loc',\region'))$
    are added to some~$\setEdges_i$ by the algorithm.
    The \emph{index} of transition $((\loc,\region),\sigma, (\loc',\region'))$, denoted by {$\algtranindex{\sigma}$},
    is equal to~$i$ if $((\loc,\region),\sigma, (\loc',\region'))$ is the $i$-th transition added by the algorithm. We define the index of a path $\dtnregionpath=\dtnregionpathdefn{}$
    {in $\DRA{\TA}$},
    denoted by $\algpathindex{\dtnregionpath}$ , as the maximum of the indices of its transitions (we define the maximum as 0 for a path of length~0).
    For a computation~$\gcomputation$ of $\network{n}$, $\algpathindex{\gcomputation}$ is the maximum over the computations of all processes.
    
    Let $\indexoffirstrepeatinglayer< \indexofsecondrepeatinglayer$ be the indices such that the algorithm stopped by condition $\alglayernodes{\indexofsecondrepeatinglayer} \approx \alglayernodes{\indexoffirstrepeatinglayer}$ such that $\slotof(\alglayernodes{i_0})=[k_{i_0},k_{i_0}]$,
    $\slotof(\alglayernodes{\indexofsecondrepeatinglayer})=[k_{l_0},k_{l_0}]$
    for some $k_{i_0},k_{l_0}\geq 0$.

    Notice that in the \dtnregionautomaton{}, delay transitions from $\setRS_{l_0-1}$ that increment the slots lead to regions with smaller slots in line~\ref{line:loopback} in \cref{algo:dtnregionautomaton}. This was done on purpose to obtain a finite construction. However, we are going to prove the lemma by building timed computations whose total delay belongs to the slot of the regions of a given path of the DTN region automaton, and this is only possible
    if slots are nondecreasing along delay transitions.
    We thus show how to \emph{fix} a given path of the \dtnregionautomaton{} by shifting the slots appropriately to make them nondecreasing.
    
    Consider a path $\dtnregionpath = \dtnregionpathdefn$.
    for each~$0 \leq i \leq l$, define $\dtnperiod{i_0,l_0}(i)$ as the number of times
    the prefix path $(\loc_0,r_0) \xRightarrow{\sigma_0} \ldots \xRightarrow{\sigma_{i-1}} (\loc_i,r_i)$ has a transition during which the slot goes from $\slotof(\setRS_{l_0-1})$ back to $\slotof(\alglayernodes{i_0})$.
    Here, we will work with $\MFboundp$-regions where $\MFboundp(\gclock) > \max(\MFbound(\gclock),(k_{l_0} - k_{i_0})\dtnperiod{i_0,l_0}(l) + k_{i_0})$,
    and $\MFboundp(\clock) = \bound{\clock}$ for all $\clock\in \clocks$.
    Because all $\MFbound$-regions of $\dtnregionpath$ have bounded slots
    (\cref{lemma:bounded-slots}), all $r_i$ are $\MFboundp$-regions as well. 
    Having a sufficiently large bound for $\gclock$ in $\MFboundp$ make sure that
    all shifted slots are bounded in \emph{all} steps of $\dtnregionpath$
    (and not just in the first $2l_0$ steps).

    We define $\dtnfix{\dtnregionpath}$ from $\dtnregionpath$
    by replacing each $r_i$ by $\newslot{(r_i)}{(k_{l_0} - k_{i_0})\dtnperiod{i_0,l_0}(i)}$.
    This simply reverts the effect of looping back to the slot of $\alglayernodes{{i_0}}$ and ensures that whenever $\sigma_i = \dtndelay$, 
    the $\MFboundp$-region $r_{i+1}$ is a time successor of~$r_i$, that is,
    $(\loc_i,r_i) \xrightarrow{\radelaylabel} (\loc_{i+1},r_{i+1})$
    in the region automaton $\ra{\MFboundp}{\UG{\TA}}$.
    We extend the definition of $\dtnfix{\cdot}$ to paths of $\da$
    by shifting the slots of the regions that appear in the locations as above
    (without changing the clock valuations or delays).
   
    We first show Claim~\ref{lbl:dtnsoundness:claim1}.
    Given a computation $\dtnregionpath$ of $\da$,
    write $\dtnfix{\dtnregionpath}= 
    ((\loc_0,r_0),v_0) \xrightarrow{\delta_0,\sigma_0} \ldots
    \xrightarrow{\delta_{l-1}, \sigma_{l-1}}
    ((\loc_l,r_l),v_l)$.
    We prove that 
    for all $t' \in \slotof(\region_l)$, there exists $n \in \Nats$ and path $\gcomputation$ of $\network{n}$
    such that %
    $\gclock(\gcomputation) = t'$
    and $\proj{\traceof{\gcomputation}}{1} = (\delta_0',\sigma_0)\ldots (\delta_{l-1}',\sigma_{l-1})$ for some $\delta_0',\ldots,\delta_{l-1}'$.
    We proceed by induction on $\algpathindex{\dtnregionpath}$.%

    This clearly holds for $\algpathindex{\dtnregionpath}=0$
    since the slot is then $[0,0]$, and the region path starts in the initial region state and has length~$0$.

    Assume that $\algpathindex{\dtnregionpath}>0$.
    Consider~$t' \in \slotof(\region_l)$.
    Computation $\dtnregionpath$ defines a path in $\ra{\MFboundp}{\UG{\TA}}$
     to which we apply 
    \cref{lbl:taglobalClockLemma}, which yields a timed computation
    $\lcomputation = \lcomputationdefn$ in $\UG{\TA}$ such that
    $\clockval_i \in \region_i$ for $0\leq i \leq l$, and $\clockval_l(\gclock) =t'$.%

    This computation has the desired trace but we still need to prove that
    it is feasible in some $\network{n}$.
    We are going to build an instance $\network{n}$ in which $\lcomputation$ will be the computation of the first process, where all location guards are satisfied.

    Let~$J$ be the set of steps~$j$ such that $\sigma_j$ is a discrete transition with a nontrivial location guard.
    Consider some $j \in J$,
    and denote its nontrivial location guard by~$\locguard$. Let 
    $t'_j = \clockval_j(\gclock) + \delta_j$ the global time at which the transition is taken.
    By the properties of $\da$, We have $t'_j \in \slotof(r_j)$.
    When the algorithm added the transition $((\loc_j,\region_j),\sigma_j,(\loc_{j+1},\region_{j+1}))$,
    some region state 
    $(\locguard,\region_\locguard)$
    with $\slotof(\region_\locguard) = \slotof(r_j)$
    was already present (by Rule 2 on line~\ref{lbl:rule2}). 
    Therefore, there exists a computation $\gcomputation^j$ in the \dtnregionautomaton{} that ends in state
    $(\locguard, r_\locguard)$
    with
    $\algtranindex{\gcomputation^j} < \algtranindex{((\loc_j,\region_j),\sigma_j,(\loc_{j+1},\region_{j+1}))}$. 
    Therefore there exists a computation
    in $\da$ as well visiting the same locations, and ending in the location $(\locguard,\region_\locguard)$.
    By induction, there exists $n_j$, and a computation
    $\gcomputation^j$ in~$\TA^{n_j}$ with 
    $\proj{\traceof{\gcomputation^j}}{1}$ ending in~$\locguard$ with
    $\gclock(\proj{\gcomputation^j}{1})= t_j'$.

    We arbitrarily extend all~$\gcomputation^j$, for $j \in J$, to global time $\totaltime(\lcomputation)$ (or further),
    which is possible by \cref{assumption:nonblocking-2}.
    We compose $\eliminate{\lcomputation}{\gclock}$ and all the $\gcomputation^j$ into a single one in $\TA^{n+1}$
    (see composition in \cref{appendix:NTAs})
    where the first process follows $\eliminate{\lcomputation}{\gclock}$,
    and the next $n=\sum_{j \in J} n_j$ follow the $\gcomputation^j$. Thus, when the first process takes a transition with a location guard $\locguard$ at time $t'_j$, there is another process at $\locguard$
    precisely at time $t'_j$.   

    Claim~\ref{lbl:dtnsoundness:claim2} is an application of Claim~\ref{lbl:dtnsoundness:claim1}. In fact,
    we can build a computation in which process~1 follows
    $\dtnfix{\dtnregionpath}$ by applying Claim~\ref{lbl:dtnsoundness:claim1} to each prefix where
    a nontrivial location guard is taken.
\end{proof}
\begin{proof}[Proof of \cref{lemma:completeness}]\label{proof:lemma:completeness}
    We prove, by induction on the length of~$\gcomputation$, a slightly stronger statement: for all $n\geq 1$, all computations $\gcomputation$ of $\network{n}$, and all $1\leq k \leq n$,
    $\proj{\traceof{\gcomputation}}{k} \in \Lg(\da)$.

    Let $\indexoffirstrepeatinglayer < \indexofsecondrepeatinglayer$ be the indices such that the algorithm stopped by condition 
    $\alglayernodes{\indexofsecondrepeatinglayer} \approx \alglayernodes{\indexoffirstrepeatinglayer}$ such that $\slotof(\alglayernodes{\indexoffirstrepeatinglayer})=[k_{\indexoffirstrepeatinglayer},k_{\indexoffirstrepeatinglayer}]$,
    $\slotof(\alglayernodes{\indexofsecondrepeatinglayer})=[k_{\indexofsecondrepeatinglayer},k_{\indexofsecondrepeatinglayer}]$
    for some $k_{\indexoffirstrepeatinglayer},k_{\indexofsecondrepeatinglayer}\geq 0$.
    Define function $\reduce(a) = a$ if 
    $a < k_{\indexofsecondrepeatinglayer}$,
    and otherwise $\reduce(a) = k_{\indexoffirstrepeatinglayer} + ((a-k_{\indexoffirstrepeatinglayer}) \mod~(k_{\indexofsecondrepeatinglayer}-k_{\indexoffirstrepeatinglayer}))$.
    This function simply removes the global time spent during the loops 
    between $\alglayernodes{\indexoffirstrepeatinglayer}$ and $\alglayernodes{\indexofsecondrepeatinglayer}$; thus, given any computation
    of $\da$ ending in $((\loc,\region),v)$, we have 
    $\reduce(v(\gclock)) \in \slotof(r)$.
     
    \smallskip
    If~$\gcomputation$ has length 0, then its trace is empty
    and thus belongs to $\Lg(\da)$.
    
    \smallskip
    Assume that the length is greater than 0.
    The proof does not depend on a particular value of~$k$, so we show the statement for $k=1$ (but induction hypotheses will use different~$k$).
    Let~$\rho = (\loc_0,v_0) \xrightarrow{\delta_0,\sigma_0} \ldots \xrightarrow{\delta_{l-1}, \sigma_{l-1}} (\loc_l,v_l)$ be a computation of $\UG{A}$
    on the trace $\traceof{\proj{\pi}{1}}$, and let $\tau_j$ be the transition 
    with label $\sigma_j$ that is taken on the $j$-th discrete transition.
 
    Assume that $\tau_{l-1}$ does not have a location guard.
    Let $\pi'$ be the prefix of~$\pi$ on the trace $(\delta_0,\sigma_0)\ldots (\delta_{l-2}, \sigma_{l-2})$
    (if $l=1$, then $\pi'=(\loc_0,v_0)$). By induction $\proj{\traceof{\pi'}}{1} \in \Lg(\da)$.
    Consider a computation of $\da$ along this trace, that ends in some configuration $((\loc,\region),v)$. By~$\pi$, the invariant of~$q$ holds
    at $v+\delta_{l-1}$, so by Rule 1 of \cref{algo:dtnregionautomaton}, $\da$ contains a region state $(\loc,\region')$ with $\eliminate{\region'}{\gclock} = \regionof{v+\delta_{l-1}}{\Mbound}$ that is reachable from $(\loc,\region)$ via $\epsilon$-transitions. 
    Moreover,
    we know by $\pi$ that the guard of $\tau_{l-1}$ is satisfied at $v+\delta_{l-1}$,
    and by Rule 2, $\da$ has a transition from $(\loc,\region')$ with label~$\sigma_{l-1}$.
    It follows that $\proj{\traceof{\pi}}{1} = \proj{\traceof{\pi'}}{1}\cdot (\delta_{l-1},\sigma_{l-1}) \in \Lg(\da)$.
 
    Assume now that $\tau_{l-1}$ has a nontrivial location guard~$\gamma$.

    Let $\pi'$ be the prefix of~$\pi$ on the trace $(\delta_0,\sigma_0)\ldots (\delta_{l-2}, \sigma_{l-2})$.
    As in the previous case, we have, by induction a computation in~$\da$
    which follows trace $\proj{\traceof{\pi'}}{1}$ and further delays
    $\delta_{l-1}$. 
    At this point, the clock guard of the transition $\tau_{l-1}$
    is also enabled. Following the same notations as above, 
    let $(\loc,\region')$ denote the location of~$\da$ reached after the additional delay $\delta_{l-1}$.
    Notice that $\reduce(\delta_0+\ldots+\delta_{l-1})\in \slotof(\region')$.
    To conclude, we just need to justify that a transition from~$(\loc,\region')$ with label $\sigma_{l-1}$ exists in $\da$.
    By \cref{algo:dtnregionautomaton}, this is the case if, and only if
    some other region state $(\gamma, s)$ exists in~$\da$ such that $\slotof(s) = \slotof(\region')$.
 
    In $\pi$, there exists some process~$k$ which is at location $\gamma$ at time
    $\delta_0+\ldots+\delta_{l-1}$ (since the last transition requires a location guard at $\locguard$). 
    Thus process $k$ is at $\gamma$ at the end of 
    $\pi'$, and remains so after the delay of $\delta_{l-1}$.
    By induction, 
    $\proj{\traceof{\pi'}}{k} \in \Lg(\da)$.
    Thus, there is a computation of $\da$ along this trace, that is, with total delay $\gclock(\pi') = \delta_0+\ldots+\delta_{l-1}$, and ending in a location~$(\gamma,s)$ for some region~$s$.
    Therefore $\reduce(\delta_0+\ldots+\delta_{l-1} )\in \slotof(s)$.
    But we also have $\reduce(\delta_0+\ldots+\delta_{l-1}) \in \slotof(\region')$
    as seen above.
    Because slots are disjoint intervals, it follows $\slotof(\region') = \slotof(s')$, which concludes the proof.
    \end{proof}
 
    \begin{proof}[Proof of \cref{thm:reach-pmcp}]
        By symmetry between processes, a label $\sigma_0$ is reachable in $\network{n}$ by some process,
        iff it is reachable by process 1 in $\network{n}$.
        By \cref{lemma:completeness,lbl:dtnsoundness},
        this is the case iff there exists $0\leq i \leq 2^{N_\TA+1}$
        and a region $\region$ such that $(\loc,\region) \in \alglayernodes{i}$ and $((\loc,\region),\sigma_0,(\loc',\region')) \in \alglayernodes{i}$ for some $(\loc',\region')$. 
        
        Notice also that because $\gclock$ has the same slot in each  $\alglayernodes{i}$, the size of $\alglayernodes{i}$ is bounded by the number of region states for a bound function where $\gclock$ is either equal to~$0$ (for slots of the form $[k,k]$), 
        or is in the interval $(0,1)$ (for slots of the form $(k,k+1)$). This is exponential in $\lvert \clocks \rvert$.
    
        Moreover, each $\alglayernodes{i}$ can be constructed only using 
        $\alglayernodes{i-1}$ and $\TA$, that is, does not require the whole sequence $\alglayernodes{0},\alglayernodes{1}, \ldots, \alglayernodes{i-1}$.
    
        The \EXPSPACE{} algorithm basically executes the main loop of \cref{algo:dtnregionautomaton} but only stores
        $\alglayernodes{i}$ at iteration~$i$. It has a binary counter to count up to $2^{N_\TA+1}$. If it encounters the target label $\sigma_0$, it stops and returns yes. Otherwise, it stops after $2^{N_\TA+1}$ iterations, and returns no.
    \end{proof}

    \begin{proof}[Proof of \cref{lemma:da-I}]
        Consider a trace 
        $\trace = (\delta_0, (i_0, \sigma_0)) \ldots (\delta_{l-1}, (i_{l-1}, \sigma_{l-1}))$ in $\proj{\Lg(\network{\infty})}{I}$
        and let $\trace_j$ denote its projection to process~$j$.
        By \cref{thm:dtn}, $\trace_j \in \Lg(\da)$, so there is a computation $\lcomputation_j$
        in $\da$ with trace $\trace_j$. Due to the labeling of the symbols with indices,
        there is no synchronization in $\otimes_{1\leq j \leq a} \da$ between different copies of $\da$ (except on time delays),
        so we can execute each $\lcomputation_j$ in the $j$-th copy in the product, and this yields  a computation with trace $\trace$.

        Conversely, consider a trace $\trace$ of $\otimes_{1\leq j \leq a} \da$.
        Similarly, it follows that there is a computation $\lcomputation_j$ in $\da$
        on trace $\proj{\trace}{j}$. By \cref{thm:dtn}, there exists $n_j\geq 1$
        such that $\proj{\trace}{j} \in \proj{\Lg(\network{n_j})}{j}$. 
        Let $\gcomputation_j$ be the computation in $\network{n_j}$ with a trace whose projection to~$j$
        is equal to $\proj{\trace}{j}$. We compose the computations $\gcomputation_j$, which yields
        a computation $\gcomputation$ in $\network{n_1+\ldots+n_a}$ such that
        $\proj{\gcomputation}{[1,a]} = \trace$.
    \end{proof}

\subsection{Proofs of \cref{section:ltba}}\label{sec:proof-ltba}
In order to simulate location guards in the lossy broadcast setting,
we use $\Synclabels = \Loc$ and transitions that require a nontrivial location guard~$\locguard$ have synchronization label $\locguard??$; we add for each location $\locguard$ a self-loop with synchronization label $\locguard!!$ \add{and with a fresh label $\iota$ from $\Sigma^-$ (such that this transition will not appear in the traces of the system)}.

The other direction is slightly more involved: given an LTBA~$B$, a \gta $\TA$ is constructed from $B$ starting with the same locations and clock invariants, and adding an auxiliary clock $\clock_\send$.
Then, for every broadcast sending transition $(\loc,\guard,\resetclocks,\tatranlabel,a!!,\loc')$ we do the following:
\begin{itemize}[topsep=0pt]
    \item add an auxiliary location $\loc_{\tatranlabel}$ with $\invariant(\loc_{\tatranlabel}) = (\clock_\send = 0)$ (\ie{} it has to be left again without time passing), and transitions $(\loc,\guard,\{\clock_\send\},\tatranlabel,\top,\loc_{\tatranlabel})$ and $(\loc_{\tatranlabel},\top,\resetclocks,\iota,\top,\loc')$ for a fresh $\iota \in \Sigma^-$;
    \item for every corresponding broadcast receiving transition $(\loc_\rcv,\guard_\rcv,\resetclocks^\rcv,\tatranlabel',a??,\loc'_\rcv)$ we add a disjunctive guarded transition
    $(\loc_\rcv,\guard_\rcv,\resetclocks^\rcv,\tatranlabel',\loc_{\tatranlabel},\loc'_\rcv)$, \ie{} receivers can only take the transition if the sender has moved to $\loc_{\tatranlabel}$.
\end{itemize}
Note that this construction relies on the fact that a label~$\Sigma$ uniquely determines the transition (\cref{assumption:unique-sigma} also applies here).

\begin{figure*}[ht]
	\centering
    \vspace*{-0.2cm}
\newcommand{\ft}{\footnotesize}

\begin{tikzpicture}[font=\footnotesize]

    \node (p0-0) [location] at (0,0) {$\loc_3$};
    \node (p0-0inv) [invariant, below=of p0-0] {$\varphi_3$};
    
    \node (p0-1) [location] at (3,0) {$\loc_4$};
    \node (p0-1inv) [invariant, below=of p0-1] {$\varphi_4$};

    \node (p1-0) [location] at (5,0) {$\loc_3$};
    \node (p1-0inv) [invariant, below=of p1-0] {$\varphi_3$};

    \node (p1-1) [location] at (8,0) {$\loc_4$};
    \node (p1-1inv) [invariant, below=of p1-1] {$\varphi_4$};
    
    \node (p1-2) [location] at (5,-1) {$\gamma$};
    \node (p1-2inv) [invariant, below=of p1-2] {$\varphi_{\gamma}$};
    
    \node (p1-3) [location] at (8,-1) {$\gamma$};
    \node (p1-3inv) [invariant, below=of p1-3] {$\varphi_{\gamma}$};
    
    \node (p2-0) [location] at (0,2) {$\loc_1$};
    \node (p2-0inv) [invariant, below=of p2-0] {$\varphi_1$};
    \node (p2-1) [location] at (3,2) {$\loc_2$};
    \node (p2-1inv) [invariant, below=of p2-1] {$\varphi_2$};
    \node (p3-0) [location] at (5,2) {$\loc_1$};
    \node (p3-0inv) [invariant, below=of p3-0] {$\varphi_1$};
    \node (p3-1) [location] at (8,2) {$\loc_2$};
    \node (p3-1inv) [invariant, below=of p3-1] {$\varphi_2$};



    
 \draw[dashed] (-1,1)--(9,1);
    \path[->]
        (p0-0) edge []
            node [below] {$\ft \guard, \resetclocks,\tatranlabel $} 
            node [above,locguard] {$\ft \gamma$} (p0-1)
    ;

    \path[->]
        (p1-0) edge [] node [below] {$\ft \guard, \resetclocks,\tatranlabel$}  node [above] {$\ft \gamma??$} (p1-1)
    ;
  \path[->]
        (p1-2) edge [] node [below] {$\ft \top, \emptyset,\iota$}  node [above] {$\ft \gamma!!$} (p1-3)
    ;
  
\path[->]
        (p2-0) edge [] node [below] {$\ft \guard, \resetclocks,\tatranlabel$}  node [above,locguard] {$\ft \gamma = \top$} (p2-1)
    ;
    
    \path[->]
        (p3-0) edge [] node [below] {$\ft \guard, \resetclocks,\tatranlabel$}  node [above] {$\ft \loc_1 !!$} (p3-1)
    ;
\end{tikzpicture}
	\caption{Gadgets for constructing an LBTA $B$ from a \gta $\TA$. The upper half shows the case of a transition with a trivial location guard in the \gta $\TA$ given on the left, for which we produce a transition in the LBTA $B$ shown on the right.
    The lower half shows the case of a transition of the \gta $\TA$ with a non-trivial location guard, given on the left, for which we produce two transitions shown on the right in the LBTA $B$.
    }
	\label{lbl:fig-dtn-lossy}
\end{figure*}
\begin{lemma}
    \label{lem:gtatolossy}
    For every \gta $\TA$, there exists an LBTA~$B$ such that for every $k \geq 1$, $\proj{\calL(\network{\infty})}{[1,k]} \equiv \proj{\calL(B^\infty)}{[1,k]}$.
    \end{lemma}
    \begin{proof}
        
        $B$ is constructed from $\TA$ by keeping locations and clock invariants, 
        {setting $\Synclabels = \Loc$}, and modifying the transitions in the following way:
        \begin{itemize}
        \item each transition $\gtatrandefn$ of $A$ with $\gamma=\top$ is simulated by a 
        sending transition 
        $(\loc,\guard,\resetclocks,\tatranlabel,{\loc!!},\loc')$, 
        \item each transition $\gtatrandefn$ of $A$ with $\gamma \neq \top$ is simulated by a receiving transition $(\loc,\guard,\resetclocks,\tatranlabel,{\locguard??},\loc')$ together with a sending transition $(\locguard,\top,\emptyset,\iota,{\locguard!!},\locguard)$,
        \end{itemize}
		\cref{lbl:fig-dtn-lossy} shows the idea of the construction.
		
        We prove that $\proj{\calL(\network{\infty})}{1} = \proj{\calL(B^\infty)}{1}$, the lemma statement follows as for \ngtas in \cref{lemma:da-I}.
    
        First, let 
        $\trace = (\delta_0, \sigma_0) \ldots (\delta_{l-1}, \sigma_{l-1})$ be the trace of a computation $\rho$ of $\network{n}$ for some $n$.
        We prove inductively that there exists a computation $\rho'$ of $B^n$ that has the same trace $\trace$ and ends in the same configuration as $\rho$.
    
        Base case: $i=0$.
        If $(\delta_0, \sigma_0)$ is a trace of $\network{n}$, then there must exist a transition $\nConfig \xrightarrow{\delta, (j,\tatranlabel)} \nConfig'$ such that $\nConfig$ is the initial configuration of $\TA$, and after a delay $\delta$ some process $j$ takes a discrete transition on label $\tatranlabel$, which may be guarded by a location $\locguard$.
        Note that by construction $\nConfig$ is also an initial configuration of $B$, and we can take the same delay $\delta$ in $\nConfig$.
        Since the transition on $(j,\tatranlabel)$ is possible after $\delta$ in $A$, there must be a transition $\gtatrandefn$ of $A$ such that $\loc=\locinit$, $\mathbf{0}+\delta$ satisfies $g$, and $\locguard$ is either $\locinit$ or $\top$.
        If $\locguard=\top$, then by construction of $B$ there exists a sending transition $(\initloc,\guard,\resetclocks,\tatranlabel, \locguard!!,\loc')$. 
        If $\locguard \neq \top$, then in $B$ there exists a receiving transition $(\locinit,\guard,\resetclocks,\tatranlabel,\locguard??,\loc')$ and a sending transition $(\locguard,\top,\locinit,\iota,\locguard!!,\locguard)$.
        In both cases, one process moves into $\loc'$ and the transition label $\tatranlabel$ is the same as for the transition in $\TA$ \add{(and the transition of $B$ labeled with $\iota$ does not appear in the trace)}.
        Therefore, the resulting configuration and trace is the same as in $\TA$.
    
        Step: $i \rightarrow i+1$.
        Assume that the property holds for the first $i$ steps.
        Then the inductive argument is the same as above, except that we are not starting from an initial configuration, but equal configurations in $\TA^n$ and $B^n$ that we get by induction hypothesis. 
    
        Now, let 
        $\trace = (\delta_0, \sigma_0) \ldots (\delta_{l-1}, \sigma_{l-1})$ be the trace of a computation $\rho$ of $B^n$ for some $n$.
        With the same proof structure above, we can show that there exists a computation $\rho'$ of $\network{n}$ that has the same trace $\trace$ and ends in the same configuration.
    \end{proof}
    
\begin{figure*}[ht]
	\centering
    \vspace*{-0.2cm}
\newcommand{\ft}{\footnotesize}

\begin{tikzpicture}[font=\footnotesize]

    \node (p0-0) [location] at (-1,0) {$\loc$};
    \node (p0-0inv) [invariant, below=of p0-0] {$\varphi$};

    \node (p1-0) [location] at (-1,-1.5) {$\loc_\rcv$};
    \node (p1-0inv) [invariant, below=of p1-0] {$\varphi_r$};

    \node (p0-1) [location] at (1.5,0) {$\loc'$};
    \node (p0-1inv) [invariant, below=of p0-1] {$\varphi'$};

    \node (p1-1) [location] at (1.5,-1.5) {$\loc'_\rcv$};
    \node (p1-1inv) [invariant, below=of p1-1] {$\varphi_r'$};

    \path[->]
        (p0-0) edge []
            node [below] {$\ft \guard, \resetclocks,\tatranlabel $} 
            node [above] {$\ft a!!$} (p0-1)
    ;

    \path[->]
        (p1-0) edge [] node [below] {$\ft \guard_\rcv,\resetclocks^\rcv,\tatranlabel'$}  node [above] {$\ft a??$} (p1-1)
    ;

    \node (p0-2) [location] at (4,0) {$\loc$};
    \node (p0-2inv) [invariant, below=of p0-2] {$\varphi$};

    \node (p1-2) [location] at (4,-1.5) {$\loc_\rcv$};
    \node (p1-2inv) [invariant, below=of p1-2] {$\varphi_r$};

    \node (p0-3) [location] at (6.5,0) {$\loc_{\tatranlabel}$};
    \node (p0-3-c) [invariant,above =of p0-3] {$\ft \clock_\send = 0$};

  \node (p0-4) [location] at (9,0) {$\loc'$};
  \node (p0-4inv) [invariant, below=of p0-4] {$\varphi'$};
  \node (p1-4) [location] at (9,-1.5) {$\loc_\rcv'$};
  \node (p1-4inv) [invariant, below=of p1-4] {$\varphi_r'$};
 \path[->] (p0-2) edge []
            node [above] {}  
            node [below] {$\ft \guard, \{\clock_\send\} ,\tatranlabel$} (p0-3)
   ;
   
\path[->] (p0-3) edge []
            node [above] {}  
            node [below] {$\ft \top, \resetclocks,\iota$} (p0-4)
   ;

    \path[->]
        (p1-2) edge [] node [above,locguard] {$\ft  \loc_{\tatranlabel}$} node [below] {$\ft \guard_\rcv,\resetclocks^\rcv,\tatranlabel'$}(p1-4)
    ;

\end{tikzpicture}
	\caption{Gadgets for constructing a \gta $\TA$ from an LBTA~$B$.
    Given a sending transition of the $B$ shown on top left, we produce the sequence of transitions in the \gta $\TA$ as shown on top right.
    For every receiving transition of the LBTA $B$ with the corresponding label $a??$, we produce a transition shown on bottom right in the \gta $\TA$.
    }
	\label{lbl:fig-lossy-dtn}
\end{figure*}
    
    \begin{lemma}
        \label{lem:lossytogta}
        For every LBTA $B$, there exists a \gta $\TA$ such that $\proj{\calL(\network{\infty})}{[1,k]} = \proj{\calL(B^\infty)}{[1,k]}$.
        \end{lemma}
    \begin{proof}
        $\TA$ is constructed from $B$ by starting with the same locations and clock invariants, and adding an auxiliary clock $\clock_\send$.
        Then, for every broadcast sending transition $(\loc,\guard,\resetclocks,\tatranlabel,a!!,\loc')$ we do the following:
        \begin{itemize}
        \item add an auxiliary location $\loc_{\tatranlabel}$ with $\invariant(\loc_{\tatranlabel}) = (\clock_\send = 0)$ (\ie{} it has to be left again without time passing), and transitions $(\loc,\guard, \{\clock_\send\},\tatranlabel,\top,\loc_{\tatranlabel})$ and $(\loc_{\tatranlabel},\top,\resetclocks ,\iota,\top,\loc')$ \add{for a fresh $\iota \in \Sigma^-$}; 
        \item for every corresponding broadcast receiving transition $(\loc_\rcv,\guard_\rcv,\resetclocks^\rcv,\tatranlabel',a??,\loc'_\rcv)$ we add a disjunctive guarded transition
        $(\loc_\rcv,\guard_\rcv,\resetclocks^\rcv,\tatranlabel',\loc_{\tatranlabel},\loc'_\rcv)$, \ie{} receivers can only take the transition if the sender has moved to $\loc_{\tatranlabel}$.
        \end{itemize}
				\cref{lbl:fig-lossy-dtn} shows the idea of the construction.
        Note that the construction relies on the fact that a label~$\Sigma$ uniquely determines the transition (\cref{assumption:unique-sigma} also applies here).
    
        Like in the proof of \cref{lem:gtatolossy}, the claim follows from proving inductively that for every computation $\rho$ of $B^n$ with trace $\trace$ there exists a computation $\rho'$ of $\network{n}$ with trace $\trace$.
				\add{In this case the proof relies on the fact that for a lossy broadcast transition, the timed transition system of $B^n$ contains the sequence of steps in any order for the receivers.}
    \end{proof}
		
		Note that~\cite{ADFL19} claims that location reachability is undecidable for automata with $2$ clocks in a model that is very similar (and may be equivalent) to LBTN; we have reasonable doubts regarding that result (which comes without a full proof), but the discrepancy might come from differences in the models as well.

\subsection{Proofs of \cref{sec:stn}}\label{sec:simulation-STA-proof}

As mentioned in the proof idea of \cref{thm:sta-equivalent}, simulating a \gta by an STA is simple.
For the other direction, let us define the \gta that will simulate a given STA.
This \gta is based on the gadget shown in \cref{lbl:fig:TN-DTN} (for one rule~$r$ of the STA), and formally defined in the following.

    \begin{figure}[h]
        \centering
\newcommand{\ft}{\footnotesize}

\begin{tikzpicture}[font=\footnotesize]

    \node (p1-0) [location] at (-.5,0) {$\loc_{\STArule,1}$};
    \node (p1-1) [location] at (3,0) {$\loc_{\STArule,1}'$};

    \node (p2-0) [location] at (-.5,1.5) {$\loc_{\STArule,2}$};
    \node (p2-1) [location] at (3,1.5) {$\loc_{\STArule,2}'$};

    \node (p3-0) [location] at (5,0) {$\loc_{\STArule,1}$};
    \node (p3-1) [location] at (9,0) {$\locmid_{\STArule,1}$};
    \node (p3-1inv) [invariant, below=of p3-1] {$\clocksync \leq 0$};
    \node (p3-2) [location] at (11.5,0) {$\loc_{\STArule,1}'$};

    \node (p4-0) [location] at (5,1.5) {$\loc_{\STArule,2}$};
    \node (p4-1) [location] at (9,1.5) {$\locmid_{\STArule,2}$};
    \node (p4-1inv) [invariant, below=of p4-1] {$\clocksync \leq 0$};
    \node (p4-2) [location] at (11.5,1.5) {$\loc_{\STArule,2}'$};

    \node (pn-0) [location] at (-.5,4) {$\loc_{\STArule,\STArulelength}$};
    \node (pn-1) [location] at (3,4) {$\loc_{\STArule,\STArulelength}'$};
    \node (pn-2) [location] at (5,4) {$\loc_{\STArule,\STArulelength}$};
    \node (pn-3) [location] at (9,4) {$\locmid_{\STArule,{\STArulelength}}$};
    \node (pn-3inv) [invariant, below=of pn-3] {$\clocksync \leq 0$};
    \node (pn-4) [location] at (11.5,4) {$\loc_{\STArule,\STArulelength}'$};
    
    \node (vdots1) at (1.5,3){$\vdots$};
    \node (vdots2) at (7.5,3){$\vdots$};
    
    \draw[dashed] (4,-1)--(4,5);

    \path[->]
        (p1-0) edge [] node [above] { $\guard_{\STArule,1}, \STAresetclocks{\STArule}{1}, \tatranlabel_{\STArule,1}$ } (p1-1);
    \path[->]
        (p2-0) edge [] node [above] {$ \guard_{\STArule,2}, \STAresetclocks{\STArule}{2}, \tatranlabel_{\STArule,2} $} (p2-1);      
    \path[->]
        (p3-0) edge []  node [above] {$ \guard_{\STArule,1}, \STAresetclocks{\STArule}{1} \cup \{\clocksync\},\dummytranlabel_{\STArule,1} $} (p3-1);
    \path[->]
        (p3-1) edge [] node [below,locguard] {$\locmid_{\STArule,\STArulelength}$} node [above] {$\tatranlabel_{\STArule,1}$} (p3-2);

    \path[->]
        (p4-0) edge [] node [below,locguard] {$\locmid_{\STArule,1}$}node [above] {$ \guard_{\STArule,2}, \STAresetclocks{\STArule}{2} \cup \{\clocksync\},\dummytranlabel_{\STArule,2} $} (p4-1);
    \path[->]
        (p4-1) edge [] node [below,locguard] {$\locmid_{\STArule,\STArulelength}$}node [above] {$\tatranlabel_{\STArule,2}$} (p4-2);

    \path[->]
        (pn-2) edge []node [below,locguard] {$\locmid_{\STArule,\STArulelength-1}$} node [above] {$ \guard_{\STArule,\STArulelength}, \STAresetclocks{\STArule}{\STArulelength} \cup \{\clocksync\},\dummytranlabel_{\STArule,\STArulelength}$} (pn-3);
    \path[->]
        (pn-0) edge [] node [above] { $ \guard_{\STArule,\STArulelength}, \STAresetclocks{\STArule}{\STArulelength}, \tatranlabel_{\STArule,\STArulelength} $} (pn-1);
    \path[->]
        (pn-3) edge []  node [above] {$\tatranlabel_{\STArule,\STArulelength}$} (pn-4);
  
\end{tikzpicture}
            \caption{On the left-hand side is a rule $\STArule$ in an STA $\STA$. On the right-hand side is the corresponding \gta-gadget, where location guards ensure that $\locmid_{\STArule,\STArulelength}$ is only reachable if all $\locmid_{\STArule,i}$ are reachable, and it follows that all $\loc'_i$ are reachable if and only if $\locmid_{\STArule,\STArulelength}$ is reachable. Furthermore, invariants on $\locmid_{\STArule,i}$ ensure that there cannot be any delay between the transitions in the gadget. Not displayed are transitions from every $\locmid_{\STArule,i}$ without any guards to the sink location $\locsink$ (with fresh labels from $\Sigma^-$).}
            \label{lbl:fig:TN-DTN}
        \end{figure}

        \begin{definition}[Corresponding \gta for a given STA]\label{lbl:corrTAdefn}
            For a given STA $\STA = \STAdefn$, its \emph{corresponding \gta{}} is defined as
            \(
            \corrTA = \corrTAdefn,
            \)
            where:
            \begin{itemize}
                \item $\Loc_{\corrTA} = \Loc \cupdot \Locmid \cupdot \{\locsink\}$, with
                    $\Locmid = \{\locmid_{\STArule, i} \mid \STArule \in \STARules, 1 \leq i \leq \STArulelength~\text{where }~ \STArule = \STAruledefn\}$,
                        \item \add{$\clocksync$ is an auxiliary clock that does not appear in $\clocks$}
                \item $\Transitions = \bigcup_{\STArule \in \STARules} \Tranleftgadget{\STArule} \cup \Tranrightgadget{\STArule} \cup T_{\STArule,\bot} $ for $\STArule=\STAruledefn$, where:
                \begin{itemize}
                    \item $\begin{array}{ll}
														\Tranleftgadget{\STArule} = 	&  \{(\loc_{\STArule,1}, \guard_{\STArule,1}, \STAresetclocks{\STArule}{1} \cup \{\clocksync\}, \dummytranlabel_{r,1}, \top, \locmid_{\STArule, 1})\}\\
														& \cup \bigcup_{2 \leq i \leq \STArulelength} \{\tranleftgadget\}
													\end{array}$\\
												  (incoming transitions of locations $\locmid_{\STArule, i}$ in \cref{lbl:fig:TN-DTN}),
                    \item $\begin{array}{ll}
														\Tranrightgadget{\STArule} = 	& \bigcup_{1 \leq i \leq \STArulelength-1} \{\tranrightgadget\} \\
														& \cup \{ (\locmid_{\STArule, m}, \top, \emptyset, \tatranlabel_{\STArule,m}, \top, \loc'_{\STArule,m})\}
													\end{array}$\\
														(outgoing transitions of locations $\locmid_{\STArule, i}$ in \cref{lbl:fig:TN-DTN}),
                    \item $T_{\STArule,\bot} = \bigcup_{1 \leq i \leq \STArulelength} \{ (\locmid_{\STArule,i},\top,\emptyset,\iota,\top,\locsink)\}$\\ (transitions to sink location, not shown in \cref{lbl:fig:TN-DTN})
                \end{itemize}
                \item $\invariant_{\corrTA}: \Loc_{\corrTA} \rightarrow \STAclockcons$ defined as  $\invariant_{\corrTA}(\loc)=\invariant(\loc)$ for every $\loc \in \Loc$, $\invariant_{\corrTA}(\locmid)= \clocksync \leq 0$ for $\locmid \in \Locmid$ and $\invariant_{\corrTA}(\locsink)=\top$.
            \end{itemize}
        \end{definition}

\begin{definition}[Stable and Intermediate Configurations]
    For an \ngta $\corrnetwork{m}$ (where $\corrTA$ is the guarded timed automaton defined above), the set of configurations is partitioned into: %
    \begin{itemize}
        \item \emph{Stable Configurations}: A configuration $\nConfig = \nConfigdefn$ is \emph{stable} if $\loc_i \in \Loc $ for all $1 \leq i \leq m$, \ie{} in stable configurations, no process is in a location $\locmid_{\STArule, i}$ or in the sink location $\locsink$.
        \sj{if we want \cref{thm:sta-equivalent} to directly follow from \cref{lem:corresponding-reachability1}-\cref{lem:corresponding-reachability2}, we need stable configurations to allow processes in the sink state; need to check that the proof still works!}
        \item \emph{Intermediate Configurations}: A configuration is \emph{intermediate} if it is not stable.
    \end{itemize}
		
		We say that a stable configuration $\corrConfig$ of $\corrTA^n$ \emph{corresponds to} a configuration $\STAConfig$ of $\STA^n$ if for every configuration $(\loc,\clockval)$ of $\corrTA$ we have $(\loc,\clockval) \in \corrConfig$ if and only if $(\loc,\eliminate{\clockval}{\clock_\sync}) \in \STAConfig$
\end{definition}

For the following lemma, we extend the projection $\eliminate{\clockval}{\clock}$ of a clock valuation $\clockval$ onto the clocks different from a clock $\clock$ to configurations in the expected way, \ie{} for $\nConfig = \big((\loc_1,\clockval_1), \ldots (\loc_{\networksize},\clockval_{\networksize})\big)$, let $\eliminate{\nConfig}{\clock}= \big( (\loc_1,\eliminate{\clockval_1}{\clock}),\ldots, (\loc_{\networksize},\eliminate{\clockval_{\networksize}}{\clock})\big)$.

\begin{lemma}
    \label{lem:corresponding-reachability1}
    Consider an STA $\STA = \STAdefn$ and its corresponding \gta $\corrTA=\corrTAdefn$.
    If a configuration $\STAConfig$ is reachable in $\STAnetwork{n}$ for some $\networksize \in \Nats$ then there exists a stable configuration $\corrConfig$ in $\corrnetwork{n}$ such that $\eliminate{\corrConfig}{\clock_\sync} = \STAConfig$ (in particular, $\corrConfig$ corresponds to $\STAConfig$).
\end{lemma}
\begin{proof}

        Let $\STAcomputation$ be a computation of $\STAnetwork{n}$ that ends in $\STAConfig$\add{, and $l$ the number of \emph{blocks} in~$\STAcomputation$, where a block is either a non-zero delay transition or a sequence of discrete transitions that correspond to the execution of a single rule $\STArule \in \STARules$. }
				The proof is by induction on~$l$, the number of blocks of~$\STAcomputation$.
        \begin{enumerate}
            \item \textbf{Base case ($l=0$)}: Let $\nConfig_{\STA,0}$  and  $\nConfig_{\corrTA,0}$ be the initial configurations of networks $\STAnetwork{n}$ and $\corrnetwork{n}$ respectively.
            Since all processes in $\nConfig_{\corrTA,0}$ and $\nConfig_{\STA,0}$ start in the initial location and all clocks are initialized to zero, we have $\eliminate{\nConfig_{\corrTA,0}}{\clocksync}=\nConfig_{\STA,0}$.
            
            \item \textbf{Induction step} \add{($l \Rightarrow l+1$)}: 
            
						We consider two cases: the final block in the computation could be a delay transition or the execution of a rule $\STArule \in \STARules$.
						
            \begin{enumerate}
						
						\item \textbf{Delay transition}:
						
						In this case $\STAcomputation$ is of the form $\nConfig_{\STA,0} \rightarrow^* \nConfig_{\STA,pre}\xrightarrow[]{\delta}\STAConfig$. 
						By induction hypothesis, there is a stable configuration $\nConfig_{\corrTA,pre}$ with $\eliminate{\nConfig_{\corrTA,pre}}{\clock_\sync}=\nConfig_{\STA,pre}$.
                        Because the $\nConfig_{\corrTA,pre}$ is stable, $\nConfig_{\STA,pre}$ has the same locations, thus the same invariants
                        as $\nConfig_{\corrTA,pre}$, so the same delay $\delta$ is possible from $\nConfig_{\corrTA,pre}$ as well.
						If $\nConfig_{\corrTA}$ is the configuration reached by a delay of $\delta$ from $\nConfig_{\corrTA,pre}$, then clearly we have $\eliminate{\corrConfig}{\clock_\sync} = \STAConfig$.
						                
            \item \textbf{Execution of rule $\STArule \in \STARules$}:
						
						In this case $\STAcomputation$ is of the form $\nConfig_{\STA,0} \rightarrow^* \nConfig_{\STA,pre} \rightarrow^* \STAConfig$, where the sequence of transitions $\nConfig_{\STA,pre} \rightarrow^* \STAConfig$ is an execution of rule $\STArule$.            
            By induction hypothesis, if 
								$\nConfig_{\STA,pre}= \STAConfigdefnpre$ is reachable in  $\STAnetwork{n}$ then a stable configuration $\nConfig_{\corrTA,pre} =\corrConfigdefnpre$ is reachable in $\corrnetwork{n}$ with $\eliminate{\altclockval_{i}}{\clock}= \clockval_{i} ~\text{for}~ 1 \leq i \leq \networksize$.
            To show that a stable configuration $\nConfig_{\corrTA}$ with the desired property can be reached from $\nConfig_{\corrTA,pre}$, we construct a computation $\gcomputation_{pre}$
            of the form $\gcomputation_{pre}=\intercomputationdefn$, where
						in the first part $\nConfig_{\corrTA,pre} \rightarrow^* \nConfig_{\corrTA,mid}$ the discrete transitions from 
						$\Tranleftgadget{\STArule}$ (for the given rule $r$, compare \cref{lbl:corrTAdefn} and \cref{lbl:fig:TN-DTN}) are executed without delay, and in the second part $\nConfig_{\corrTA,mid} \rightarrow^* \nConfig_{\corrTA}$ the discrete transitions from $\Tranrightgadget{\STArule}$ are executed.
						We analyze the properties of this computation in the following.
           
            \begin{itemize}

                \item \textbf{Part 1}: $\nConfig_{\corrTA,pre} \rightarrow^* \nConfig_{\corrTA,mid}$, executing ~$\Tranleftgadget{\STArule}$: 
								
								Assume w.l.o.g. that the transitions in $\nConfig_{\STA,pre} \rightarrow^* \STAConfig$ in $\STAnetwork{n}$ are executed by the first $m$ processes, and process $i$ takes element $\loc_{\STArule,i}  \xrightarrow[]{\guard_{\STArule,i},\STAresetclocks{\STArule}{i},\tatranlabel_{\STArule,i}}  \loc'_{\STArule,i}$ of the rule $\STArule$. 
								Then, for $1 \leq i \leq \STArulelength$ and $(\loc_i,\clockval_i) \in \nConfig_{\STA,pre}$ we know that $\clockval_i \models \guard_{\STArule,i}$. 
                Now consider $\nConfig_{\corrTA,pre}$: since $\eliminate{\altclockval_{i}}{\clock}= \clockval_{i} ~\text{for}~ 1 \leq i \leq \networksize$, we have $\altclockval_{i} \models \guard_{\STArule,i}$ for $1 \leq i \leq \STArulelength$, \ie{} process $i$ for $1 \leq i \leq m$ satisfies the clock guard of the $i$th transition in $\Tranleftgadget{\STArule}$.
								By taking the transitions in increasing order of $i$, also all location guards are satisfied. 
								In addition, note that the set of clock resets for the $i$th transition in $\Tranleftgadget{\STArule}$ is the same as for the $i$th element of rule $\STArule$, except for $\clocksync$ (which is reset on all transitions in $\Tranleftgadget{\STArule}$).
                Since $\eliminate{\altclockval_{i}}{\clock}= \clockval_{i}$ for $1 \leq i \leq \networksize$, we get that the clock values in $\nConfig_{\corrTA,mid}$ are equal (up to $\clocksync$) to those in $\STAConfig$.
								Finally, note that in $\nConfig_{\corrTA,mid}$, process $i$ occupies location $\locmid_{\STArule,i}$ for $1 \leq i \leq \STArulelength$ (and the other processes have not changed their configuration).

                \item \textbf{Part 2}: $\nConfig_{\corrTA,mid} \rightarrow^* \nConfig_{\corrTA}$, executing ~$\Tranrightgadget{\STArule}$: 
								
								Starting from $\nConfig_{\corrTA, mid}$, each process $i$ for $1 \leq i \leq \STArulelength$ takes the $i$th transition in $\Tranrightgadget{\STArule}$ and moves to $\loc_{\STArule,i}'$, say in increasing order of $i$.
                 Note that this is possible because the guard location $\locmid_{\STArule, \STArulelength}$ is occupied in $\nConfig_{\STArule, mid}$, and will stay occupied until the $m$th process takes the $m$th transition from $\Tranrightgadget{\STArule}$. 
								In the resulting configuration $\nConfig_{\corrTA}$, the first $m$ processes will be in locations $\loc_{\STArule,i}'$ according to rule $\STArule$ and the other processes will be in the same location as in $\nConfig_{\corrTA,pre}$ (and therefore the same as in $\nConfig_{\STA,pre}$).
								Therefore, all processes will be in the same location as in $\STAConfig$.		
            Moreover, since none of these transitions alters clock valuations of any process, we get that 
						$\eliminate{\nConfig_{\corrTA}}{\clocksync}=\STAConfig$.
						Finally, observe that $\nConfig_{\corrTA}$ is a stable configuration, proving the desired property.
            \end{itemize} 
						
        \end{enumerate}
        \end{enumerate}
        
    \end{proof}

    We now prove the converse direction. Here, starting from a stable configuration of $\corrnetwork{n_1}$, we will build a corresponding reachable configuration
    in $\STAnetwork{n_2}$ for some $n_2$. The reason of this discrepancy is that in STAs, each rule is applied to exactly $m$ processes,
    while in GTAs, nothing prevents more processes to cross the gadget of \cref{lbl:fig:TN-DTN}. 
		
		As a concrete example, consider the STA $\STA$ on the left side of \cref{lbl:fig:STN-DTN} and the gadgets for its two rules that appear in $\corrTA$ on the right side:
		To reach location $\loc_3$ in $\corrTA$, three processes are sufficient---in the gadget for rule $\STArule_1$ (at the right bottom), one process moves from $\locinit$ via $\locmid_{\STArule_1,1}$ to $\loc_1$, and two processes move $\locinit$ via $\locmid_{\STArule_1,2}$ to $\loc_2$, and then these two processes can execute the gadget of $\STArule_2$ (right top) such that one of them arrives in $\loc_3$.
		In the STA $\STA$, at least four processes are needed to make one of them reach $\loc_3$--- upon firing $\STArule_1$ a single time, only one process is in $\loc_2$ (and another has moved to $\loc_1$), such that we need two additional processes in $\locinit$ to fire $\STArule_1$ a second time, and only after that can $\STArule_2$ be fired and one process reaches $\loc_3$.
		
		Accordingly, 
    the statement is also weaker: given $\corrConfig$, the lemma shows that there exists a reachable configuration $\STAConfig$ 
    that corresponds to $\corrConfig$ (but without necessarily having $\eliminate{\corrConfig}{\clock_\sync} = \STAConfig$).
    The proof is a bit more involved, and requires a copycat lemma given at the end of this section.

\newcommand{\ft}{\footnotesize}
\begin{figure}
\centering
    \begin{tikzpicture}[font=\footnotesize]

        \node[location, initial] at (0, 2)  (q0) {$\locinit$};
        \node[location] at (1.5, 2) (q1) {$\loc_1$};
        \node[location] at (1.5, 4) (q2) {$\loc_2$};
        \node[location] at (3.5, 4) (q3) {$\loc_3$};


        \path[->] (q0) edge[] node[above]{$\STArule_1, \sigma_1$} (q1);
        \path[->] (q0) edge[] node[above,sloped]{$\STArule_1, \sigma_2$} (q2);
        \path[->] (q2) edge[] node[above,sloped]{$\STArule_2, \sigma_3$} (q1);
        \path[->] (q2) edge[] node[above]{$\STArule_2, \sigma_4$} (q3);

    \node (p1-0) [location] at (4.5,0) {$\locinit$};
    \node (p1-1) [location] at (6.5,0) {$\locmid_{\STArule_1,1}$};
    \node (p1-2) [location] at (8.5,0) {$\loc_1$};

    \node (p2-0) [location] at (4.5,2) {$\locinit$};
    \node (p2-1) [location] at (6.5,2) {$\locmid_{\STArule_1,2}$};
    \node (p2-2) [location] at (8.5,2) {$\loc_2$};

\node (p3-0) [location] at (4.5,4) {$\loc_2$};
\node (p3-1) [location] at (6.5,4) {$\locmid_{\STArule_2,1}$};
\node (p3-2) [location] at (8.5,4) {$\loc_1$};

\node (p4-0) [location] at (4.5,6) {$\loc_2$};
\node (p4-1) [location] at (6.5,6) {$\locmid_{\STArule_2,2}$};
\node (p4-2) [location] at (8.5,6) {$\loc_3$};
    \path[->]
    (p2-0) edge [] node [below,locguard] {$\locmid_{\STArule_1,1}$}   (p2-1);
\path[->]
    (p2-1) edge [] node [above] {$\tatranlabel_{2}$}  (p2-2);

    \path[->]
    (p1-0) edge []   (p1-1);
\path[->]
    (p1-1) edge [] node [above] {$\tatranlabel_{1}$} node [below,locguard] {$\locmid_{\STArule_1,2}$}(p1-2);

 \path[->]
 (p4-0) edge [] node [below,locguard] {$\locmid_{\STArule_2,1}$}   (p4-1);
\path[->]
 (p4-1) edge [] node [above] {$\tatranlabel_{4}$}  (p4-2);

 \path[->]
 (p3-0) edge []   (p3-1);
\path[->]
 (p3-1) edge [] node [above] {$\tatranlabel_{3}$} node [below,locguard] {$\locmid_{\STArule_2,2}$}(p3-2);

 \draw[dashed] (4,-1)--(4,7);
 \draw[dashed] (4,3)--(11,3);
   
\end{tikzpicture}
\caption{A STA $\STA$ template for which atleast 4 processes are required to reach $\loc_3$ and 
on the right of the figure are gadgets corresponding to its rules. Note that $\loc_3$ is reachable in $\corrnetwork{3}$
where $\corrTA$ is the  corresponding guarded timed automata
}   \label{lbl:fig:STN-DTN}
        \end{figure}
		
    \begin{lemma}
        \label{lem:corresponding-reachability2}
        Consider an STA $\STA = \STAdefn$ and its corresponding \gta $\corrTA=\corrTAdefn$.
        If a stable configuration $\corrConfig$ is reachable in $\corrnetwork{n_1}$ for some $\networksize_1 \in \Nats$, then a configuration $\STAConfig$ is reachable in $\STAnetwork{n_2}$ for some $n_2 \in \Nats$ such that $\corrConfig$ corresponds to $\STAConfig$.
    \end{lemma}
    \begin{proof}
        Let $\stablecomputation$ be a computation of $\corrnetwork{\networksize_1}$ that ends in a stable configuration $\nConfig_{\corrTA}$.
        The proof is by induction on $l$, the number of non-zero delay transitions in $\stablecomputation$.
        \begin{enumerate}
         
            \item \textbf{Base case ($l=0$)}:						
						Let $\stablecomputationbasecase=\stablecomputationbasecasedefn$ be the computation to $\nConfig_{\corrTA}$ which in general can have multiple discrete but no non-zero delay transitions.
            Let $\transequence=\langle \tatran_{1}\ldots \tatran_{k}\rangle$ be the sequence of transitions of $\corrTA$ that appear on $\stablecomputation$, 
            and let $\transequenceset = \{ \tatran_{1}, \ldots, \tatran_{k}\}$ be the set of the transitions in $\transequence$.
            For a given rule $\STArule \in \STARules$ of $\STA$, let $\Trangadget{\STArule}=\Tranleftgadget{\STArule} \cup \Tranrightgadget{\STArule}$, with $\Tranleftgadget{\STArule}, \Tranrightgadget{\STArule}$ as defined in \cref{lbl:corrTAdefn}.
            
						We distinguish two cases: 
            \begin{ienumerate}
                \item $\transequenceset \subseteq \Trangadget{\STArule}$ for some rule $\STArule= \STAruledefn$ \sj{maybe change notation of rules from arrow notation to tuple, s.t. transitions for all types of automata are uniform}, \ie{} all the transitions occuring in $\transequence$
                are a part of only one rule.
                \item $\transequenceset \subseteq \bigcup_{\STArule \in \STARules'} \Trangadget{\STArule}$ for some $\STARules' \subseteq \STARules$
                with $|\STARules'| > 1$.
            \end{ienumerate}

             \begin{enumerate}
                \item \textbf{Case $\transequenceset \subseteq \Trangadget{\STArule}$ for some $\STArule \in \STARules$}:
								
                We claim that in this case $\transequenceset = \Trangadget{\STArule}$, \ie{} if some of the transitions of $\Trangadget{\STArule}$ appear in $\transequence$, then all of them must appear because of following observations:
                \begin{itemize}
                    \item First, note that all transitions in $\Tranrightgadget{\STArule}$ have $\locmid_{\STArule,\STArulelength}$ as a location guard, so $\locmid_{\STArule, \STArulelength}$ has to be occupied to take any of these transitions. 
                    So if a transition from $\Tranrightgadget{\STArule}$
                    appears in $\transequenceset$,
                    since $(\loc_{\STArule,\STArulelength}, \guard_{\STArule,\STArulelength}, \STAresetclocks{\STArule}{\STArulelength} \cup \{\clocksync\}, \dummytranlabel_{\STArule,\STArulelength}, \{\locmid_{\STArule, \STArulelength-1}\}, \locmid_{\STArule, \STArulelength}) \in \Tranleftgadget{\STArule}$ is the only incoming transition to $\locmid_{\STArule, \STArulelength}$, then 
                    we know that it must also be in $\transequenceset$.
                    But for every $i>1$ we have that $\tranleftgadget \in \Tranleftgadget{\STArule}$ has 
                    $\locmid_{\STArule, i-1}$ as a location guard, so all elements of $\Tranleftgadget{\STArule}$ must be in $\transequenceset$.
										
                    \item If for some $i \in \{1,\ldots,\STArulelength\}$ a transition $\tranleftgadget \in \Tranleftgadget{\STArule}$ appears in $\transequence$ 
                    then $\tranrightgadget \in \Tranrightgadget{\STArule}$ must also appear in $\transequence$, otherwise the process that takes the former transition would either be stuck in the auxiliary location $\locmid_{\STArule,i}$ or be in the sink location $\locsink$ by the end of $\stablecomputationbasecase$, contradicting the assumption that $\nConfig_{\corrTA}$ is stable.
                     
                \end{itemize}

                Thus, if one of the transitions in $\Trangadget{\STArule}$ appears in $\transequence$, we can assume that 
                all of them appear.  
								Furthermore, since all transitions are taken without a delay, we know that the initial configuration $\nConfig_{\corrTA,0}$ satisfies the clock guards of all transitions in $\Tranleftgadget{\STArule}$, and therefore all clock guards of rule $\STArule$, \ie{} rule $\STArule$ can be fired in $\STA$ from $\nConfig_{\STA,0}$.
								To prove the claim, consider which local configurations can be contained in $\nConfig_{\corrTA}$: since the transitions in $\transequence$ are exactly those from $\Trangadget{\STArule}$ and we arrive in a stable configuration, $\nConfig_{\corrTA}$ must contain $(\loc_{\STArule,i}',\mathbf{0})$ for all $1 \leq i \leq m$. Moreover, note that if $\networksize_1 > m$ then it may contain $(\locinit,\mathbf{0})$ (but it does not have to, since every path in the gadget can be taken by an arbitrary number of processes).
								For both cases, we can find a suitable $\networksize_2$ such that $\STA$ reaches a global configuration that contains exactly the same local configurations: if $n_2 = m$, then the resulting $\nConfig_\STA$ contains exactly all the $(\loc_{\STArule,i}',\mathbf{0})$, and if $n_2 > m$, then it additionally contains $(\locinit,\mathbf{0})$.

                \item \textbf{Case $\transequenceset \subseteq \bigcup_{\STArule \in \STARules'} \Trangadget{\STArule}$ for some $\STARules' \subseteq \STARules$
                with $|\STARules'| > 1$:
                }
								
								For simplicity, consider $\transequenceset \subseteq \Trangadget{\STArule_1} \cup \Trangadget{\STArule_2} $   for $r_1\neq r_2 \in \STARules$.
                Like in the previous case we can argue that all transitions in $\Trangadget{\STArule_1} \cup \Trangadget{\STArule_2}$ have to appear in $\transequence$. 
								\add{
								To prove the claim, we will reorder the sequence of transitions.
								Assume w.l.o.g. that $\locmid_{\STArule_1,\STArulelength_1}$ is reached before $\locmid_{\STArule_2,\STArulelength_2}$ on 
								$\stablecomputationbasecase$.
								We will show that we can reorder the computation such that all transitions from  $\Trangadget{\STArule_1}$ are taken before all transitions from $\Trangadget{\STArule_2}$, and we reach the same configuration $\nConfig_{\corrTA}$.} 
								\add{The overall construction of this proof case is depicted in \cref{lbl:fig-dtn-stn-run}.}
								
\begin{figure}[h]
    \centering

\begin{tikzpicture}
    \node (CA01) at (0,4) {$\nConfig_{\corrTA,0}$};
    \node (CA11) at (6,4) {$\nConfig_{\corrTA}$};
		
    \node (CA0) at (0,2) {$\nConfig_{\corrTA,0}$};
    \node (CA1) at (3,2) {$\nConfig_{\corrTA,1}$};
    \node (CA) at (6,2) {$\nConfig_{\corrTA}$};
    
    \node (CS0) at (0,0) {$\nConfig_{\STA,0}$};
    \node (CS1) at (3,0) {$\nConfig_{\STA,1}$};
    \node (CS2) at (3,-1.5) {$\nConfig_{\STA,2}$};
    \node (CS) at (6,-1.5) {$\STAConfig$};
    \node (star1) at (2.5,2.1) {*};
    \node (star2) at (5.65,2.1) {*};
    \node (star2) at (5.65,4.1) {*};

    \draw[->] (CA01) --  node[above] {$\transequence$}  (CA11);
    \draw[->] (CA0) --  node[above] {$\proj{\transequence}{\Trangadget{\STArule_1}}$}  (CA1);
    \draw[->] (CA1) -- node[above] {$\proj{\transequence}{\Trangadget{\STArule_2}}$} (CA);

    \draw[->] (CS0) -- node[above] {$\STArule_1$} (CS1);
    \draw[decorate, decoration={snake, amplitude=1mm, segment length=2mm}, ->] (CS1) -- node[right] {\text{\cref{lem:team-copycat}} (Copycat)} (CS2);
    \draw[->] (CS2) -- node[above] {$\STArule_2$} (CS);

\end{tikzpicture}
\caption{Constructing a computation of a network of STAs (of possibly a larger size) from a run of a network of corresponding \gta{}, with an intermediate step of reordering the given run}
\label{lbl:fig-dtn-stn-run}
\end{figure}
								
								\medskip
								To see this, we first show that all transitions in $\Tranleftgadget{\STArule_1}$ are enabled in $\nConfig_{\corrTA,0}$, \ie{} for every element
								$\loc_{\STArule_1,i}  \xrightarrow[]{\guard_{\STArule_1,i},\STAresetclocks{\STArule_1}{i},\tatranlabel_{\STArule_1,i}}  \loc'_{\STArule_1,i}$ of $\STArule_1$ we have $(\loc_{\STArule_1,i},\clockval_{\STArule_1,i}) \in \nConfig_{\corrTA,0}$ with $\clockval_{\STArule_1,i} \models \guard_{\STArule_1,i}$ (and since $\nConfig_{\corrTA,0}$ is the initial configuration we know that $\clockval_{\STArule_1,i}=\mathbf{0}$).
								To see this, remember that $\locmid_{\STArule_1,\STArulelength_1}$ is reached before $\locmid_{\STArule_2,\STArulelength_2}$ on 
								$\stablecomputationbasecase$ and note that none of the transitions of $\Tranrightgadget{\STArule_2}$ can appear before $\locmid_{\STArule_2,\STArulelength_2}$ is reached, and therefore none of the $\loc_{\STArule_2,i}'$ are available when the transitions in $\Tranleftgadget{\STArule_1}$ are taken.
								Since moreover $\Loc \cap \Locmid = \emptyset$ (auxiliary locations do not appear in rules $\STA$) and no time passes in $\stablecomputationbasecase$, it must be the case that all transitions in $\Tranleftgadget{\STArule_1}$ are already enabled in $\nConfig_{\corrTA,0}$.
								
								\medskip
								Regarding transitions in $\Tranrightgadget{\STArule_1}$, note that these start in auxiliary locations that are unique to $\Trangadget{\STArule_1}$ and \add{by construction} the only location guard is $\locmid_{\STArule_1,\STArulelength_1}$ (and again, no time passes). 
								Therefore all transitions in $\Tranrightgadget{\STArule_1}$ are enabled as soon as $\locmid_{\STArule_1,\STArulelength_1}$ is reached, independently of any transitions from~$\Trangadget{\STArule_2}$.

	\medskip
								Thus, we can consider a different sequence of transitions $\transequence'=\proj{\transequence}{\Trangadget{\STArule_1}} \cdot \proj{\transequence}{\Trangadget{\STArule_2}}$, which is a concatenation of the subsequences obtained by projecting the original sequence onto the transitions in $\Trangadget{\STArule_1}$ or $\Trangadget{\STArule_2}$, respectively.
								Let $\stablecomputationbasecase'$ be a computation with $\transequence'$ its sequence of transitions, and where each transition is taken by the same process as in $\stablecomputationbasecase$.
								Therefore, $\stablecomputationbasecase'$ ends in the same configuration $\nConfig_{\corrTA}$.
								
								Note that, since $\nConfig_{\corrTA}$ is a stable configuration, we must also reach a stable configuration $\nConfig_{\corrTA,1}$ after executing $\proj{\transequence}{\Trangadget{\STArule_1}}$ from $\nConfig_{\corrTA,0}$ (if $\nConfig_{\corrTA,1}$ was not stable, either an auxiliary location of $\Trangadget{\STArule_1}$ or the sink location would be occupied in $\nConfig_{\corrTA,1}$, and the same would still be true after executing $\proj{\transequence}{\Trangadget{\STArule_2}}$ and reaching $\nConfig_{\corrTA}$).
										
								\medskip
								To construct a computation of $\STA$ that reaches a configuration $\nConfig_\STA$ with the desired property, \add{we show that from a suitably chosen $\nConfig_{\STA,0}$ we can first execute $\STArule_1$ and then $\STArule_2$. To this end,}
								first note that by the properties of $\nConfig_{\corrTA,0}$ established above we know that $\STArule_1$ can be executed from any initial configuration $\nConfig_{\STA,0}$ that has at least $\STArulelength_1$ processes.
								We want to ensure that $\nConfig_{\corrTA,1}$ corresponds to the configuration $\nConfig_{\STA,1}$ that is reached after executing $\STArule_1$ in $\STA$.
									Similar to what we explained in the first case above, after executing $\proj{\transequence}{\Trangadget{\STArule_1}}$ from $\nConfig_{\corrTA,0}$, the resulting configuration $\nConfig_{\corrTA,1}$ must contain $(\loc_{\STArule_1, i}',\mathbf{0})$ for all $1 \leq i \leq \STArulelength_1$, and it may or may not contain $(\locinit,\mathbf{0})$.
									If it does contain $(\locinit,\mathbf{0})$, then we get to a corresponding $\nConfig_{\STA,1}$ by starting with $\STArulelength_1+1$ processes, otherwise by starting with $\STArulelength_1$ processes.
									
									\medskip
									\add{To see that $\STArule_2$ can be executed from $\nConfig_{\STA,1}$, first} note that $\proj{\transequence}{\Trangadget{\STArule_2}}$ can be executed from $\nConfig_{\corrTA,1}$ in $\corrTA$.
									From this we can conclude that for every element $\loc_{\STArule_2,i}  \xrightarrow[]{\guard_{\STArule_2,i},\STAresetclocks{\STArule_2}{i},\tatranlabel_{\STArule_2,i}}  \loc'_{\STArule_2,i}$ of $\STArule_2$, we have $(\loc_{\STArule_2,i},\clockval_{\STArule_2,i}) \in \nConfig_{\corrTA,1}$ for some $\clockval_{\STArule_2,i}$ that satisfies $\guard_{\STArule_2,i}$ (actually, since no time passed, all clock valuations are $\mathbf{0}$).
									As $\nConfig_{\corrTA,1}$ corresponds to $\nConfig_{\STA,1}$, we have the same property for $\nConfig_{\STA,1}$.
									However, note that multiple processes occupying the same $(\loc_{\STArule_2,i},\clockval_{\STArule_2,i})$ might be necessary to execute $\STArule_2$, e.g., if $\loc_{\STArule_2,i}$ appears on the left-hand side of multiple elements of $\STArule$.
									By \cref{lem:team-copycat}, for any lower bound on the number of processes required in each local configuration, we can reach a global configuration $\nConfig_{\STA,2}$ that has sufficiently many processes in every local configuration that is needed.
									As a consequence, we can execute $\STArule_2$ from $\nConfig_{\STA,2}$.
									
									\medskip
									\add{Finally, }to ensure that the resulting configuration $\nConfig_\STA$ corresponds to $\nConfig_{\corrTA}$, we need to pick $\nConfig_{\STA,2}$ with the right number of processes in each local configuration by a similar argument as above.
									
														             \end{enumerate}
           \item \textbf{Induction step} ($l \Rightarrow l+1$): 
					
						Let $\stablecomputation = \nConfig_{\corrTA,0} \rightarrow^* \nConfig_{\corrTA,1} \xrightarrow[]{\delta} \nConfig_{\corrTA,2} \rightarrow^* \nConfig_{\corrTA}$ be a computation with $l+1$ non-zero delay transitions, where the last non-zero delay happens after $\nConfig_{\corrTA,1}$, and between $\nConfig_{\corrTA,2}$ and $\nConfig_{\corrTA}$ there is a (possibly empty) sequence of discrete transitions with zero delay between them.
						Note that since we assume that $\nConfig_{\corrTA}$ is stable, also $\nConfig_{\corrTA,1}$ and $\nConfig_{\corrTA,2}$ need to be stable configurations (since a non-zero delay is not possible if any location from $P$ is occupied, and if the sink location is occupied it will remain occupied forever).
						Then the timed path $\nConfig_{\corrTA,0} \rightarrow^* \nConfig_{\corrTA,1}$ has only $l$ non-zero delays and by induction hypothesis we get that in $\STA$ we can reach a configuration $\nConfig_{\STA,1}$ that corresponds to $\nConfig_{\corrTA,1}$.
						
						Note that if we can take a delay transition with delay $\delta$ from $\nConfig_{\corrTA,1}$, then we can also take it from $\nConfig_{\STA,1}$, and the resulting configuration $\nConfig_{\STA,2}$ will again correspond to $\nConfig_{\corrTA,2}$.
						What remains to be shown is that from $\nConfig_{\STA,2}$ we can reach a configuration $\nConfig_{\STA}$ that corresponds to $\nConfig_{\corrTA}$.
						This works essentially in the same way as in the base case, except that now we might need to invoke \cref{lem:team-copycat} even if only a single rule is executed (since in the base case we can freely choose how many processes should be in the initial configuration $\nConfig_{\STA,0}$, whereas here we need to prove that we can bring sufficiently many processes to the local configurations that are needed). 
        \end{enumerate}
\end{proof}

		\add{\cref{lem:corresponding-reachability1,lem:corresponding-reachability2} consider reachability of \emph{stable} configurations of $\corrnetwork{n}$ that contain given configurations $\nConfig_{\corrTA}$ of $\corrTA$.
		We now show that the reachability of stable configurations containing a location $\loc$ is equivalent to reachability of any configuration containing $\loc$.}
						
		\begin{lemma}
		\label{lem:reachability-stable}
		Let $\STA$ be an STA with set of locations $\Loc$ and $\corrTA$ its corresponding \gta.
		For any $\loc \in \Loc$, a configuration $\nConfig$ with $\loc \in \nConfig$ is reachable in $\corrTA$ if and only if a stable configuration $\nConfig'$ with $\loc \in \nConfig'$ is reachable in $\corrTA$.
		\end{lemma}
		
		\begin{proof}
		$\Leftarrow$: immediate
		
		$\Rightarrow$: Suppose an intermediate configuration $\nConfig$ with $\loc \in \nConfig$ is reachable, 
		and let $\gcomputation = \nConfig_0 \rightarrow^* \nConfig_s \rightarrow^* \nConfig$ be the computation that ends in $\nConfig$, where $\nConfig_s$ is the last stable configuration on $\gcomputation$.
		Since $\nConfig$ is not stable, there are some processes that are in locations $\locmid_{\STArule,i}$ for some $\STArule, i$, or in $\locsink$.
		For each of these processes, modify the local timed path of the process (\ie{} the projection of $\gcomputation$ onto this process) in the following way:
		Let $p$ be the location that it occupies at the end. 
		At the point in time where the process finally moves to $p$, let the process instead do anything else such that it remains in one of the locations in $\Loc$. (This is possible since we assumed the absence of timelocks in $\STA$)
		
		Then let $\gcomputation'$ be the result of replacing the local timed paths of all these processes with the modified versions.
		Note that:
		\begin{enumerate}
		\item $\gcomputation'$ is still a valid computation (as those processes that have been modified cannot be necessary for the steps of other processes after they have moved to their $p$)
		\item We did not change the local timed path of the process that occupies $q$ at the end of $\gcomputation$, therefore it also occupies $q$ at the end of $\gcomputation'$
		\item $\gcomputation'$ ends in a stable configuration
		\end{enumerate}
		\end{proof}

\cref{lem:corresponding-reachability1,lem:corresponding-reachability2,lem:reachability-stable} together imply \cref{thm:sta-equivalent}.

Note that our construction is in general not sufficient to prove language equivalence.
\cref{fig:counterexample-language} shows an STA $\STA$ with three locations and two rules: $r_1$ with $m_1=1$, \ie{} a single participating process and $r_2$ with $m_2=2$, \ie{} two participating processes.
An invariant in location $\locinit$ forces any process to take a transition after at most $1$ time unit.

However, in the corresponding \gta $\corrTA$, a process can move from $\locinit$ to $\locmid_{\STArule_2,1}$ and from there to $\locsink$, both without delay and with labels in $\Sigma^-$, \ie{} that will not appear in the trace.
Therefore, the trace $(1, (1,\sigma_0)),(1, (1,\sigma_0))$ is in $\proj{\calL(\corrTA^\infty)}{[1,2]}$ (where process $2$ moves to $\locsink$ and never uses a transition that is visible in the trace), but not in $\proj{\calL(\STA^\infty)}{[1,2]}$ (where the second process has to take a visible transition after at most $1$ time unit).

\begin{figure}
    \centering
		\vspace{-1em}
    \begin{tikzpicture}[pta, font=\footnotesize]
        \node[location, initial] at (0, 0)  (q0) {$\locinit$};
        \node[invariant, below=of q0] {$\clock \leq 1$};
        \node[location] at (2, 0) (q1) {$\loc_1$};
        \node[location] at (2, 1) (q2) {$\loc_2$};

        \path (q0) edge[loop above] node[align=center]{$r_1, \sigma_0$\\$\clock=1$\\$\clock \leftarrow 0$} (q0);
        \path (q0) edge[] node[below]{$r_2, \sigma_1$} (q1);
        \path (q0) edge[] node[above]{$r_2, \sigma_2$} (q2);
    \end{tikzpicture}
    \caption{An example of an STA $\STA$ for which $\proj{\calL(\STA^\infty)}{[1,2]} \neq \proj{\calL(\corrTA^\infty)}{[1,2]}$.}
    \label{fig:counterexample-language}
\end{figure}

\begin{definition}
    For a given configuration $\nConfig$ of $\STA^n$ and a configuration $(\loc,\clockval)$ of $\STA$, let $\numOfOcc{\nConfig}{(\loc,\clockval)}$ indicate the number of occurrences of $(\loc,\clockval)$ in $\nConfig$.
\end{definition}

\begin{lemma}
\label{lem:team-copycat} 
        \textbf{Team Copycat in network of STAs}: In a network of STAs $\STAnetwork{\networksize}$, if there is a reachable configuration $\STAConfig$ such that
        $\numOfOcc{\STAConfig}{(\loc_i,\clockval_i)}=k_i$ for $1 \leq i \leq \networksize$,
        then $\forall 1 \leq j \leq \networksize$ there exists an $\networksize' \in \Nats$ and a reachable configuration $\STAConfig'$ in $\STAnetwork{\networksize'}$ such that
        $\numOfOcc{\STAConfig'}{(\loc_i,\clockval_i)} \geq k_i$ for $1 \leq i \leq \networksize$
        and $\numOfOcc{\STAConfig'}{(\loc_j,\clockval_j)} \geq k_j +1$.
    \end{lemma}
    \begin{proof}
\begin{enumerate}
    Let $\STAcomputation$ be a computation of $\STAnetwork{n}$ that ends in $\STAConfig$.
    In order to prove the lemma, we prove the following \emph{doubling} property:
     Given a reachable configuration $\STAConfig= \big((\loc_1, \clockval_1) \ldots (\loc_{\networksize},\clockval_{\networksize})\big) $ in $\STAnetwork{\networksize}$,
     we prove the configuration
     $\STAConfig'=\big((\loc_1, \clockval_1) \ldots (\loc_{\networksize},\clockval_{\networksize}),(\loc_{\networksize+1},\clockval_{\networksize+1}) \ldots (\loc_{2 \cdot \networksize},\clockval_{2 \cdot \networksize}) \big)$,
    where $(\loc_{j},\clockval_{j})=(\loc_{j-\networksize},\clockval_{j-\networksize})$ for  $\networksize+1 \leq j \leq 2 \cdot \networksize$, is  reachable in $\STAnetwork{2 \cdot \networksize}$.
    
    Notice that this property implies there exists a reachable configuaration, namely $\STAConfig'$,
    such that  for every $(\loc,\clockval) \in \STAConfig$, $\numOfOcc{\STAConfig'}{(\loc,\clockval)} \geq \numOfOcc{\STAConfig}{(\loc,\clockval)}+1$  and therefore proving the lemma.

Now we prove the doubling property.
\sj{the rest is pretty obvious. summarize?}
    \item 

    \add{Let $l$ be the number of \emph{blocks} in~$\STAcomputation$, where a block is either a non-zero delay transition or a sequence of discrete transitions that correspond to the execution of a single rule $\STArule \in \STARules$.}
	The proof is by induction on~$l$, the number of blocks of~$\STAcomputation$.

\begin{itemize}
    \item \textbf{Base case ($l=0$)}: Consider the network of size $2 \cdot \networksize$ and the initial configuration of this new network satisfies the desired property.
    \item \textbf{Induction step} \add{($l \Rightarrow l+1$)}: 
    We consider two cases: the final block in the computation could be a delay transition or the execution of a rule $\STArule \in \STARules$.
    \begin{itemize}
        \item \textbf{Delay transition}:
        
        In this case $\STAcomputation$ is of the form $\nConfig_{\STA,0}\xrightarrow[]{}\ldots \nConfig_{\STA,pre}\xrightarrow[]{\delta}\STAConfig$. 
        From hypothesis it follows  $\nConfigdefn_{\STA,pre}'=\big((\loc_1, \clockval_1) \ldots (\loc_{\networksize},\clockval_{\networksize}),(\loc_{\networksize+1},\clockval_{\networksize+1}) \ldots (\loc_{2 \cdot \networksize},\clockval_{2 \cdot \networksize}) \big)$
        where $(\loc_{j},\clockval_{j})=(\loc_{j-\networksize},\clockval_{j-\networksize})$ for  $\networksize+1 \leq j \leq 2 \cdot \networksize$. Delaying $\delta$ from this configuration results in $\STAConfig'$ 
        which has the desired property. 

        \item \textbf{Execution of rule $\STArule \in \STARules$}:
        
        In this case $\STAcomputation$ is of the form $\nConfig_{\STA,0}\xrightarrow[]{}\ldots \nConfig_{\STA,pre} \rightarrow^* \STAConfig$, where the sequence of transitions $\nConfig_{\STA,pre} \rightarrow^* \STAConfig$ is an execution of rule $\STArule$, where $\STArule=\STAruledefn$.      
       
        From hypothesis it follows  $\nConfig_{\STA,pre}'=\big((\loc_1, \clockval_1) \ldots (\loc_{\networksize},\clockval_{\networksize}),(\loc_{\networksize+1},\clockval_{\networksize+1}) \ldots (\loc_{2 \cdot \networksize},\clockval_{2 \cdot \networksize}) \big)$
        where $(\loc_{j},\clockval_{j})=(\loc_{j-\networksize},\clockval_{j-\networksize})$ for  $\networksize+1 \leq j \leq 2 \cdot \networksize$ is reachable.

        Let the computation to $\nConfig_{\STA,pre}'$ be $\pi_{\STA,pre}' = \nConfig_{\STA,0}\xrightarrow[]{}\ldots \nConfig_{\STA,pre}'$.
        We extend $\pi_{\STA,pre}'$ with two additional blocks each of which correspond to execution of $\STArule$.
        Let the resulting computation be $\STAcomputation'= \nConfig_{\STA,0}\xrightarrow[]{}\ldots \nConfig_{\STA,pre}'\rightarrow^* \nConfig_{\STA,mid}'\rightarrow^* \STAConfig'$. More precisely, $\STAConfig'$ is obtained from $\nConfig_{\STA,pre}'$ as follows:
        \begin{itemize}
            \item We first apply the sequence of transitions in $\STArule$ on the configuration $\nConfig_{\STA,pre}'$(which is of size $2 \cdot \networksize$) during which the processes $1 \ldots \STArulelength$ participate.
            Note that this  is possible because, by assumption $\nConfig_{\STA,pre} \rightarrow^* \STAConfig$, and therefore the first $\STArulelength$ processes (w.l.o.g) satisfy the transitons of $\STArule$.
            Also the first $\STArulelength$ configurations of $\nConfig_{\STA,pre}'$ are same as those in 
            $\nConfig_{\STA,pre}$.
            \item  We then apply again the sequence of transitions in $\STArule$, without any delay in between, starting from $\nConfig_{\STA,mid}'$. Now the processes $\networksize+1, \ldots, \networksize +\STArulelength$
            participate (this is possible because  the configuration of $j$th process is same as $(j-\networksize)$th process in $\nConfig_{\STA,pre}'$ for $\networksize+1 \leq j \leq 2 \cdot \networksize$)
            to finally obtain $\STAConfig'$.
        \end{itemize}
        
         The obtained configuration $\STAConfig'$ has the desired property because in the process of obtaining $\STAConfig'$ from $\nConfig_{\STA,pre}'$ above, both the $k$th and $(k-\networksize)$th, for $\networksize \leq k \leq \networksize+\STArulelength$, processes took the same
        transition from rule $\STArule$ (without any delay in between), while the rest of the processes did not take any transition. Therefore the doubling property is proved, hence proving the lemma.
        
    \end{itemize}

\end{itemize}

\end{enumerate}

    \end{proof}

\end{document}